\documentclass[10pt,journal,twocolumn,final,twoside]{IEEEtran}
\usepackage{ifpdf}
\usepackage{cite}
\usepackage{graphicx}
\usepackage[cmex10]{amsmath}
\usepackage{algorithmic}
\usepackage{array}
\usepackage{mdwmath}
\usepackage{mdwtab}
\usepackage{eqparbox}
\usepackage[tight,footnotesize]{subfigure}
\usepackage{stfloats}
\usepackage{amssymb}
\usepackage{bm}
\usepackage{overpic}
\usepackage{color}
\usepackage{dsfont}
\usepackage{cases}
\usepackage{xr}
\newtheorem{theorem}{Theorem}
\newtheorem{lemma}{Lemma}

\newtheorem{definition}{Definition}

\newtheorem{assumption}{Assumption}
\newtheorem{example}{Example}

\def\nn{\nonumber}

\def\Expt{\mathbb{E}}
\def\E{\mathbb{E}}
\def\mb{\mathbb}
\def\defeq{\triangleq}
\def\mc{\mathcal}
\def\col{\mathrm{col}}

\def\diag{\mathrm{diag}}
\def\cent{\mathrm{cent}}
\def\Tr{\mathrm{Tr}}
\def\bP{\bar{P}}
\def\Pm{P^{(4)}}
\def\one{\mathds{1}}

\def\bd{\bm{d}}

\def\bn{\bm{n}}
\def\bs{\bm{s}}
\def\bu{\bm{u}}
\def\bv{\bm{v}}
\def\bw{\bm{w}}
\def\bx{\bm{x}}
\def\by{\bm{y}}
\def\bz{\bm{z}}

\def\mA{\mc{A}}
\def\mD{\mc{D}}
\def\mF{\mc{F}}
\def\mM{\mc{M}}
\def\mN{\mc{N}}
\def\mR{\mc{R}}
\def\mS{\mc{S}}

\def\mU{\mc{U}}
\def\mW{\mc{W}}

\def\s{\bm{s}}

\def\u{\bm{u}}
\def\w{\bm{w}}
\def\x{\bm{x}}

\def\z{\bm{z}}

\allowdisplaybreaks

\title{On the Learning Behavior of Adaptive Networks --- Part II: Performance Analysis}

\begin{document}

\author{Jianshu~Chen,~\IEEEmembership{Member,~IEEE,}%
        ~and~Ali~H.~Sayed,~\IEEEmembership{Fellow,~IEEE}% <-this % stops a space
%\thanks{Copyright (c) 2013 IEEE. Personal use of this material is permitted. However, permission to use this material for any other purposes must be obtained from the IEEE by sending a request to pubs-permissions@ieee.org.}
\thanks{
Manuscript received December 28, 2013; revised November 21, 2014; accepted
March 29, 2015. This work was supported in part by NSF grants CCF-1011918 and ECCS- 1407712.
A short early version of limited parts of this work appears in the conference publication%
\cite{chen2013EUSIPCObenefits} without proofs and under more restrictive conditions.}
\thanks{J. Chen was with Department of Electrical Engineering, University of California, Los Angeles, and is currently with Microsoft Research, Redmond, WA 98052. This work was performed while he was a PhD student at UCLA. Email: cjs09@ucla.edu.
}
\thanks{A. H. Sayed is with Department of Electrical Engineering,
University of California, Los Angeles, CA 90095. Email: sayed@ee.ucla.edu.
}
\thanks{Communicated by Prof. Nicol\`{o} Cesa-Bianchi,  Associate Editor for Pattern Recognition, Statistical Learning, and Inference.}
\thanks{Copyright (c) 2014 IEEE.
Personal use of this material is permitted. However, permission to use this material for any other
purposes must be obtained from the IEEE by sending a request to pubs-permissions@ieee.org.
}% <-this % stops a space
%\thanks{Color versions of one or more of the figures in this paper are available online
%at http://ieeexplore.ieee.org.}%
}

\markboth{IEEE Transactions on Information Theory, VOL.~XX, NO.~XX, MONTH~2015}{Chen and Sayed: 
On the Learning Behavior of Adaptive Networks --- Part II: Performance Analysis}

\maketitle

%\vspace{-4em}
\begin{abstract}
Part I\cite{chen2013learningPart1} of this work examined the mean-square stability and convergence of the learning process of distributed strategies over graphs. The results identified conditions on the network topology, utilities, and data in order to ensure stability; the results also identified three distinct stages in the learning behavior of multi-agent networks related to transient phases I and II and the steady-state phase. This Part II examines the steady-state phase of distributed learning by networked agents. Apart from characterizing the performance of the individual agents, it is shown that the network induces a useful equalization effect across all agents. In this way, the performance of noisier agents is enhanced to the same level as the performance of agents with less noisy data. It is further shown that in the small step-size regime, each agent in the network is able to achieve the same performance level as that of a centralized strategy corresponding to a fully connected network. The results in this part reveal explicitly which aspects of the network topology and operation influence performance and provide important insights into the design of effective mechanisms for the processing and diffusion of information over networks.
\end{abstract}
%

%%%%%%%%%%%%%%%%%%%%%%%%%%%%%%%%%%%%%%%%%%%%%%%%%%%%%%%%%%%%%%%%%%%%%%%%%%%%%%%%
\section{INTRODUCTION}
\label{Sec:Intro}

In Part I of this work\cite{chen2013learningPart1}, we carried out a detailed transient analysis of the global learning behavior of multi-agent networks.
The analysis revealed interesting results about the learning abilities of distributed strategies when {\em constant} step-sizes are used to ensure continuous tracking of drifts in the data. It was noted that when constant step-sizes are employed to drive the learning process, the dynamics of the distributed strategies is modified in a critical manner. Specifically, components that relate to gradient noise are not annihilated any longer, as happens when diminishing step-sizes are used.  These noise components remain persistently active throughout the adaptation process and it becomes necessary to examine their impact on network performance, such as examining questions of the following nature: (a) can these persistent noise components drive the network unstable? (b) can the degradation in performance be controlled and minimized? (c) what is the size of the degradation?  Motivated by these questions, we provided in Part I \cite{chen2013learningPart1} detailed answers to the following three inquiries: (i) where  does the distributed strategy converge to? (ii) under what conditions on the data and network topology does it converge? (iii) and what are the rates of convergence of the learning process? In particular, we showed in Part I \cite{chen2013learningPart1} that  there always exist sufficiently small constant step-sizes that ensure the mean-square convergence of the learning process to a well-defined limit point even in the presence of persistent gradient noise.

We characterized this limit point as the {\em unique} fixed point solution of a nonlinear algebraic equation consisting of the weighted sum of individual update vectors. The scaling weights were shown to be given by the entries of the right-eigenvector of the network combination policy corresponding to the eigenvalue at one (also called the Perron eigenvector; its entries are normalized to add up to one and are all strictly positive for strongly-connected networks). The analysis from Part I \cite{chen2013learningPart1} further revealed that  the learning curve of the multi-agent network exhibits \emph{three} distinct phases. In the first phase (Transient Phase I), the convergence rate of the network is determined by the second largest eigenvalue of the combination policy in magnitude, which is related to  the degree of network connectivity. In the second phase (Transient Phase II), the convergence rate is determined by the Perron eigenvector. And, in the third phase (the steady-state phase) the mean-square error (MSE) performance attains a bound on the order of step-size parameters.

In this Part II of the work,  we address in some detail two additional questions related to network performance, namely,  iv) how close do the individual agents get to { the limit point of the distributed strategies over the network}? and v) can the system of networked agents be made to match the learning performance of a centralized solution where all information is collected and processed centrally by a fusion center? In the process of answering these questions, we shall derive a closed-form expression for the steady-state MSE of each agent. This closed-form expression turns out to be a revealing result; it amounts to a non-trivial extension of a classical result for stand-alone adaptive agents \cite{widrow1976stationary,jones1982analysis,gardner1984learning,feuer1985convergence} to the more demanding context of networked agents and for cost functions that are not necessarily quadratic or of the mean-square-error type. As we are going to explain in the sequel, the closed-form expression of the steady-state MSE captures the effect of the network topology (through the Perron vector of the combination matrix), gradient noise, and data characteristics in an integrated manner and shows  how these various factors influence performance. The derived results in this paper applies to connected networks under fairly general conditions and for fairly general aggregate cost functions.

We shall also explain later in Sections \ref{Sec:LearnBehav:GlobalBehav} and \ref{Sec:Benefits} of this part that, as long as the network is strongly connected, a left-stochastic combination matrix can always be constructed to have any desired Perron-eigenvector. This observation has an important ramification for the following reason. Starting from any collection of $N$ agents, there exists a finite number  of topologies that can link these agents together. And for each possible topology, there are infinitely many combination policies that can be used to train the network. Since the performance of the network is dependent on the Perron-eigenvector of its combination policy, one of the important conclusions that will follow is that regardless of the network topology, there will always exist choices for the respective combination policies such that the steady-state performance of all topologies can be made identical to each other to first-order in $\mu_{\max}$, which is the largest step-size across agents. In other words, no matter how the agents are connected to each other, there is always a way to select the combination weights such that the performance of the network is invariant to the topology. This will also mean that, for any connected topology, there is always a way to select the combination weights such that the performance of the network matches that of the centralized stochastic-approximation (since a centralized solution can be viewed as corresponding to a fully-connected network).

\noindent
{\bf Notation}. We adopt the same notation from Part I\cite{chen2013learningPart1}.
All vectors are column vectors. We use boldface letters to denote random quantities (such as $\bm{u}_{k,i}$) and regular font to denote their realizations or deterministic variables (such as $u_{k,i}$). We use $\mathrm{diag}\{x_1,\ldots,x_N\}$ to denote a (block) diagonal matrix consisting of diagonal entries (blocks) $x_1,\ldots,x_N$, and use $\mathrm{col}\{x_1,\ldots,x_N\}$ to denote a column vector formed by stacking $x_1,\ldots,x_N$ on top of each other. The notation $x \preceq y$ means each entry of the vector $x$ is less than or equal to the corresponding entry of the vector $y$, and the notation $X \preceq Y$ means each entry of the matrix $X$ is less than or equal to the corresponding entry of the matrix $Y$. The notation $x=\mathrm{vec}(X)$ denotes the vectorization operation that stacks the columns of a matrix $X$ on top of each other to form a vector $x$, and $X=\mathrm{vec}^{-1}(x)$  is the inverse operation. The operators $\nabla_w$ and $\nabla_{w^T}$ denote the column and row gradient vectors with respect to $w$. When $\nabla_{w^T}$ is applied to a  column vector $s$, it generates a 
matrix. The notation $a(\mu)\sim b(\mu)$ means that $\lim_{\mu\rightarrow 0} a(\mu)/b(\mu)=1$, $a(\mu) = o( b(\mu) )$ means that $\lim_{\mu\rightarrow 0} a(\mu)/b(\mu)=0$, and $a(\mu) = O( b(\mu) )$ means that there exists a constant $C>0$ such
that $a(\mu) \le  C \cdot b(\mu)$. The notation $a = \Theta( b )$ means there exist constants $C_1$ and $C_2$  independent of $a$ and $b$ such that $C_1 \cdot b \le a \le C_2 \cdot b$.

%%%%%%%%%%%%%%%%%%%%%%%%%%%%%%%%%%%%%%%%%%%%%%%%%%%%%%%%%%%%%%%%%%%%%%%%%%%%%%%%
\section{Family of Distributed Strategies }
\label{Sec:ProblemFormulation}

\subsection{Distributed Strategies: Consensus and Diffusion}
\label{Sec:ProblemFormulation:Dist}

	\begin{figure}[t]
		\centering
		\includegraphics[width=0.28\textwidth]{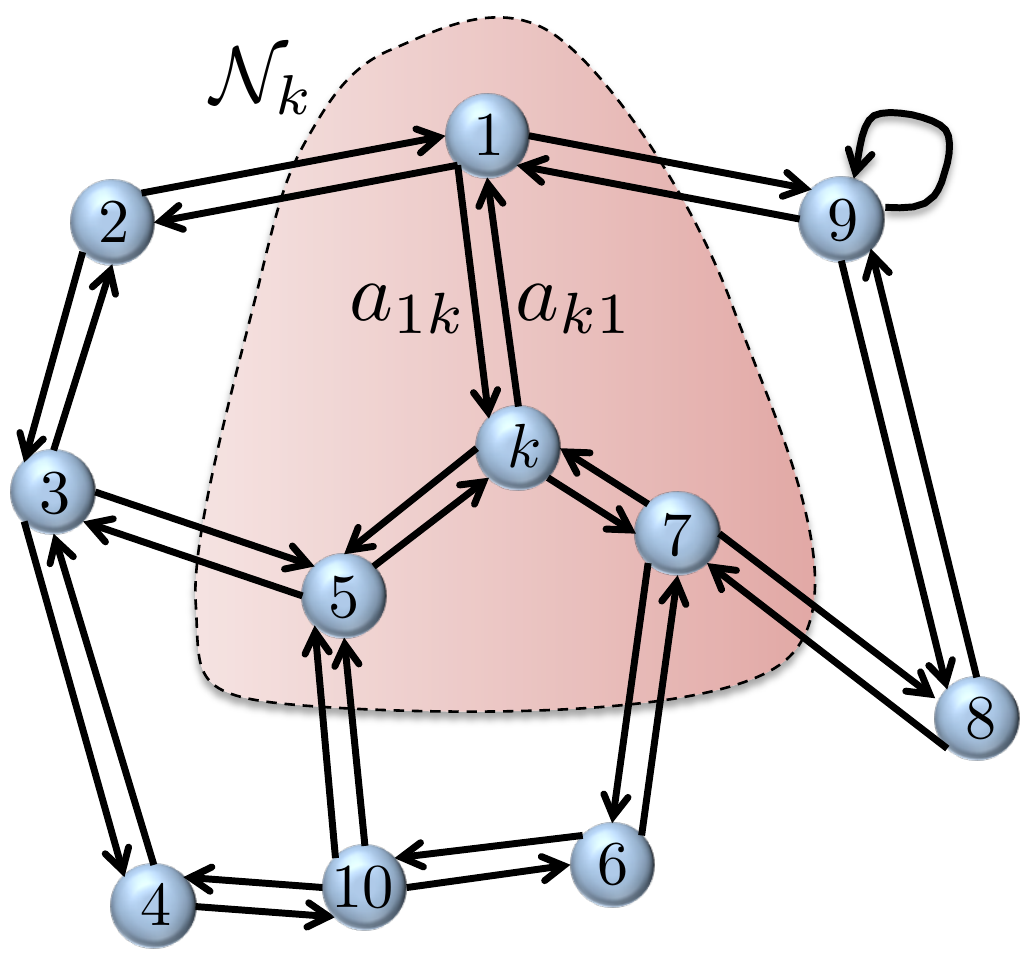}
		\caption{A network representing a multi-agent system. The set of all agents that
		can communicate with node $k$ is 
		denoted by $\mc{N}_k$.The edge linking any two agents is 
		represented by two directed arrows to emphasize that 
		information can flow in both directions.}
		\label{Fig:Fig_Network}
	\end{figure}

We consider a connected network of $N$ agents that are linked together through a topology --- see Fig. \ref{Fig:Fig_Network}. Each agent $k$ implements a distributed algorithm of the following form to update its state vector from $\bw_{k,i-1}$ to $\bw_{k,i}$:
	\begin{align}
		\label{Equ:ProblemFormulation:DisAlg_Comb1}
		\bm{\phi}_{k,i-1}	&=	\sum_{l=1}^N a_{1,lk} \bm{w}_{l,i-1}		\\
		\label{Equ:ProblemFormulation:DisAlg_Adapt}
		\bm{\psi}_{k,i}		&=	\sum_{l=1}^N a_{0,lk} \bm{\phi}_{l,i-1}
								-
								\mu_k \hat{\bm{s}}_{k,i}(\bm{\phi}_{k,i-1})	\\
		\label{Equ:ProblemFormulation:DisAlg_Comb2}
		\bm{w}_{k,i}		&=	\sum_{l=1}^N a_{2,lk} \bm{\psi}_{l,i}
	\end{align}
where $\bm{w}_{k,i} \in \mb{R}^M$ is the state of agent $k$ at time $i$, usually an estimate for the solution of some optimization problem, $\bm{\phi}_{k,i-1} \in \mb{R}^M$ and $\bm{\psi}_{k,i} \in \mb{R}^M$ are intermediate variables generated at node $k$ before updating to $\bm{w}_{k,i}$, $\mu_k$ is a non-negative constant step-size parameter used by node $k$, and $\hat{\bm{s}}_{k,i}(\cdot)$ is an $M \times 1$ update vector function at node $k$. We explained in Part I \cite{chen2013learningPart1} that in deterministic optimization problems, the update vectors $\hat{\bm{s}}_{k,i}(\cdot)$  can be selected as the gradient or Newton steps associated with the individual utility functions at the agents\cite{nedic2009distributed}. On the other hand, in stocastic approximation problems, such as adaptation, learning and estimation problems \cite{tsitsiklis1986distributed,ram2010distributed,srivastava2011distributed,kar2011converegence,kar2008distributed,kar2013distributed,dimakis2010gossip,theodoridis2011adaptive,dini2012cooperative,lopes2008diffusion,Cattivelli10,zhao2012performance,chen2011TSPdiffopt,chen2013JSTSPpareto,sayed2012diffbookchapter,chouvardas2011adaptive,gharenhshiran2013jstsp,sayed2014adaptation,sayed2014proc},
the update vectors $\hat{\s}_{k,i}(\cdot)$ are usually computed from realizations of data samples that arrive sequentially at the nodes. In the stochastic setting, the quantities appearing in \eqref{Equ:ProblemFormulation:DisAlg_Comb1}--\eqref{Equ:ProblemFormulation:DisAlg_Comb2}  become random variables and we shall use boldface letters to highlight their stochastic nature. In Example \ref{P1-Ex:UpdateVector} of Part I\cite{chen2013learningPart1}, we illustrated various choices for $\hat{\bm{s}}_{k,i}(w)$ in different contexts.

The combination coefficients $a_{1,lk}$, $a_{0,lk}$ and $a_{2,lk}$ in \eqref{Equ:ProblemFormulation:DisAlg_Comb1}--\eqref{Equ:ProblemFormulation:DisAlg_Comb2} are nonnegative convex-combination weights that each node $k$ assigns to the information arriving from node $l$ and will be zero if agent $l$ is not in the neighborhood of agent $k$. Therefore, each summation in \eqref{Equ:ProblemFormulation:DisAlg_Comb1}--\eqref{Equ:ProblemFormulation:DisAlg_Comb2} is actually confined to the neighborhood of node $k$. We let $A_1$, $A_0$ and $A_2$ denote the $N \times N$ matrices that collect the coefficients $\{a_{1,lk}\}$, $\{a_{0,lk}\}$ and $\{a_{2,lk}\}$. Then, the matrices $A_1$, $A_0$ and $A_2$ satisfy
	\begin{align}
		\label{Equ:ProblemFormulation:a_lk:convex_matrixform}
		A_1^T \mathds{1} = \mathds{1},	\quad
		A_0^T \mathds{1} = \mathds{1},	\quad
		A_2^T \mathds{1} = \mathds{1}
	\end{align}
where $\one$ is the $N\times 1$ vector with all its entries equal to one. Condition \eqref{Equ:ProblemFormulation:a_lk:convex_matrixform} means that the matrices $\{A_0,A_1,A_2\}$ are left-stochastic (i.e., the entries on each of their columns add up to one). We also explained in Part I\cite{chen2013learningPart1} that different choices for $A_1$, $A_0$ and $A_2$ correspond to different distributed strategies, such as the such as the traditional consensus\cite{nedic2009distributed,nedic2010cooperative,tsitsiklis1986distributed,kar2011converegence,kar2008distributed,kar2013distributed,dimakis2010gossip} and diffusion (ATC and CTA) \cite{lopes2008diffusion,sayed2012diffbookchapter, Cattivelli10,zhao2012performance,chen2011TSPdiffopt,chen2013JSTSPpareto,sayed2014adaptation,sayed2014proc} algorithms --- see Table \ref{Tab:ChoiceOfMatrixA}.
In our analysis, we will proceed with the general form \eqref{Equ:ProblemFormulation:DisAlg_Comb1}--\eqref{Equ:ProblemFormulation:DisAlg_Comb2} to study all three schemes, and other
possibilities, within a unifying framework.

\begin{table}
	\centering
	\caption{Different choices for $A_1$, $A_0$ and $A_2$ correspond to different
	distributed strategies.}
	\label{Tab:ChoiceOfMatrixA}
	\begin{tabular}{c|ccc|c}
		\hline\hline
		Distributed Strategeis	&	$A_1$	&	$A_0$	&	$A_2$	&	$A_1A_0A_2$	\\
		\hline
		Consensus				&	$I$		&	$A$		&	$I$		&	$A$			\\
		ATC diffusion			&	$I$		&	$I$		&	$A$		&	$A$			\\
		CTA	diffusion			&	$A$		&	$I$		&	$I$		&	$A$			\\
		\hline
	\end{tabular}
	\end{table}

\subsection{Review of the Main Results from Part I\cite{chen2013learningPart1}}
\label{Sec:FamilyDistStrategy:ReviewPartI}

Due the coupled nature of the social and self-learning steps in \eqref{Equ:ProblemFormulation:DisAlg_Comb1}--\eqref{Equ:ProblemFormulation:DisAlg_Comb2}, information derived from local data at agent $k$ will be propagated to its neighbors and from there to their neighbors in a diffusive learning process. It is expected that some global performance pattern will emerge from these localized interactions in the multi-agent system. As mentioned in the introductory remarks, in Part I \cite{chen2013learningPart1} and in this Part II, we examine the following five questions:
	\begin{itemize}
		\item
			\underline{Limit point}: where does each state $\bm{w}_{k,i}$ converge to?
		\item
			\underline{Stability}: under which condition does convergence occur?
		\item	
			\underline{Learning rate}: how fast does convergence occur?
		\item
			\underline{Performance}: how close does $\bm{w}_{k,i}$ get to the limit point?
		\item
			\underline{Generalization}: can $\w_{k,i}$ match the performance of a centralized solution?
	\end{itemize}
In Part I\cite{chen2013learningPart1}, we  addressed the first three questions in detail and derived expressions that fully characterize the answer in each case. One of the main conclusions established in Part I\cite{chen2013learningPart1} is that for general \emph{left-stochastic} matrices $\{A_1,A_0,A_2\}$, the agents in the network will have their iterates $\bw_{k,i}$ converge, in the mean-square-error sense, to the \emph{same} limit vector $w^o$ that corresponds to the unique solution of the following
algebraic equation:
	\begin{align}
		\sum_{k=1}^N p_k s_k(w)	=	0
		\label{Equ:LearnBehav:FixedPointEqu}
	\end{align}
where the update functions $s_k(\cdot)$ are defined further ahead in \eqref{Equ:Assumption:Randomness:MDS} as the conditional means of the update directions $\hat{\s}_{k,i}(\cdot)$ used in \eqref{Equ:ProblemFormulation:DisAlg_Comb1}--\eqref{Equ:ProblemFormulation:DisAlg_Comb2}, and each positive coefficient $p_k$ is the $k$th entry of the following vector:
	\begin{align}
		p	&=		
					\col\left\{ 
						\frac{\mu_1}{\mu_{\max}} \pi_1, \ldots, \frac{\mu_N}{\mu_{\max}} \pi_N
					\right\}
		\label{Equ:ProblemForm:p_def}
	\end{align}
Here, $\mu_{\max}$ is the largest step-size among all agents, $\pi_k$ is the $k$th entry of the vector $\pi \defeq A_2 \theta$, and $\theta$ is the right eigenvector of $A \defeq A_1A_0A_2$ corresponding to the eigenvalue at one with its entries normalized to add up to one, i.e., 
	\begin{align}
		A \theta = \theta, \qquad \one^T \theta = 1
		\label{Equ:ProblemForm:PerronVector_A}
	\end{align}
We refer  to $\theta$ as the Perron eigenvector of $A$. The unique solution $w^o$ of \eqref{Equ:LearnBehav:FixedPointEqu} has the interpretation of a Pareto optimal solution corresponding to the weights $\{p_k\}$ \cite{chen2013learningPart1,chen2013JSTSPpareto,boyd2004convex}. By  selecting different combination policies $A$, or even different topologies, the entries $\{p_k\}$ can be made to 
change (since $\theta$ will change) and the limit point $w^o$ resulting from \eqref{Equ:LearnBehav:FixedPointEqu} can  be steered towards different Pareto optimal solutions.

\begin{figure}[t]
		\centering
		\includegraphics[width=0.45\textwidth]{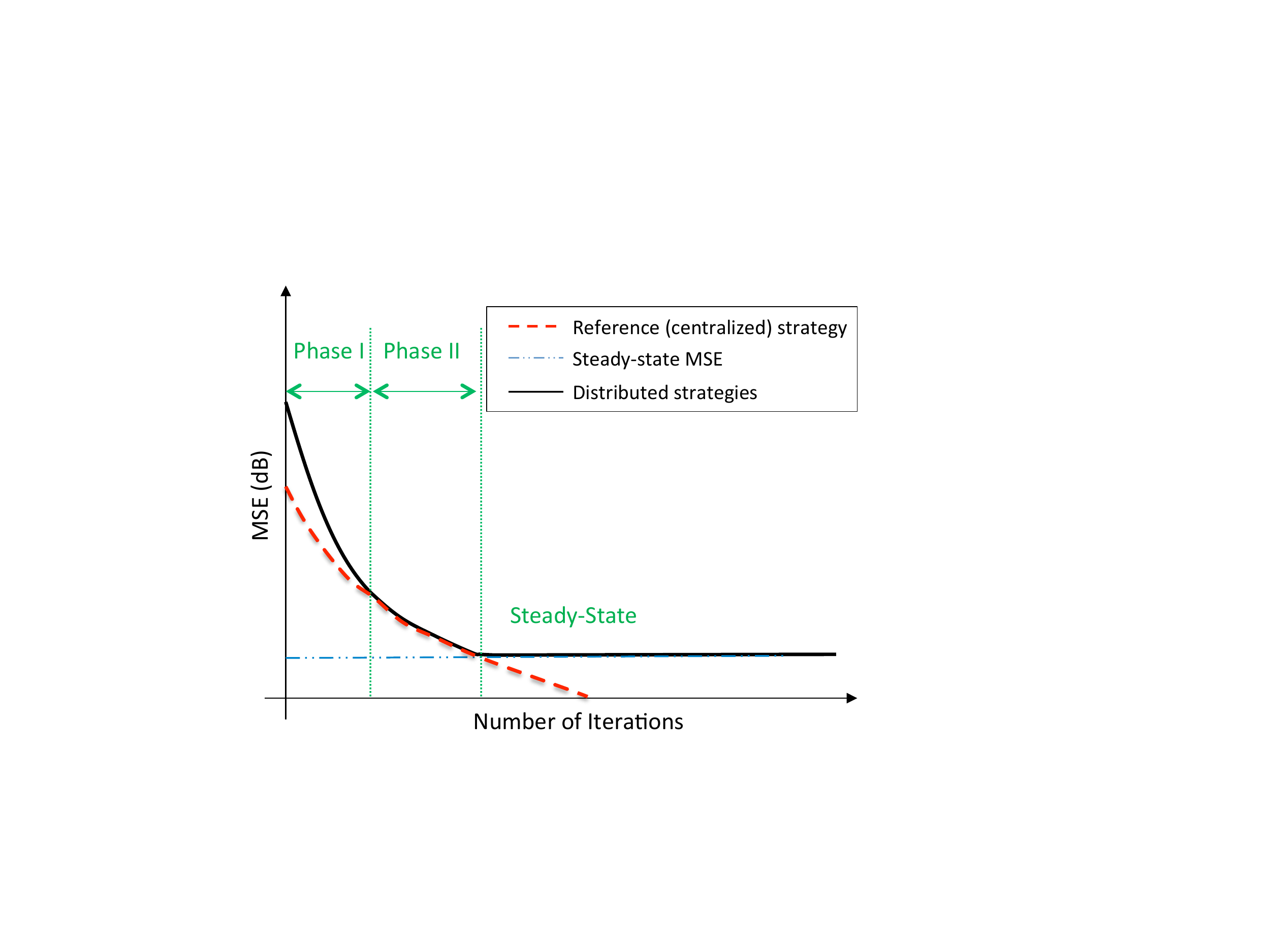}
		\caption{A typical mean-square-error (MSE) 
		learning curve includes a transient stage that consists of
		two phases and a steady-state phase. The plot shows 
		how the learning curve of a network of agents compares 
		to the learning curve of a centralized reference solution. 
		The analysis in this work, and in the accompanying Part I \cite{chen2013learningPart1} characterizes in detail the 
		parameters that determine the behavior of the network 
		(rate, stability, and performance) during each phase of the 
		learning process.}
		\label{Fig:TypicalLearningCurve}
		\vspace{-1\baselineskip}
	\end{figure}

The second major conclusion from Part I \cite{chen2013learningPart1} is that, during the convergence process towards the limit point $w^o$, the learning
curve at each agent exhibits \emph{three} distinct phases (see Fig. \ref{Fig:TypicalLearningCurve}): Transient Phase I, Transient Phase II, and Steady-State Phase. These phases were shown in Part I \cite{chen2013learningPart1} to have the following features:
	\begin{itemize}
		\item
		{\bf Transient Phase I:}\\
		If the agents are initialized at different values, then the iterates at
		the various agents will initially evolve in such a way to make
		each $\bw_{k,i}$ get closer to the following
		\emph{reference} (centralized) recursion $\bar{w}_{c,i}$:
			\begin{align}
				\label{Equ:LearnBehav:RefRec}
				\bar{w}_{c,i}	=	\bar{w}_{c,i-1} - \mu_{\max}\sum_{k=1}^N p_k s_k(\bar{w}_{c,i-1})
			\end{align}
		which is initialized at 
			\begin{align}
				\bar{w}_{c,0}	=	\sum_{k=1}^N \theta_k w_{k,0}
				\label{Equ:LearnBehav:wc0bar_initRefRecur}
			\end{align}
		where $w_{k,0}$ is the initial value of the distributed
		strategy at agent $k$. The rate at which the agents approach
		$\bar{w}_{c,i}$ is { geometric (linear) and is} determined by $|\lambda_2(A)|$, 
		the second largest eigenvalue of $A$ in magnitude. 
		If the agents are initialized at the
		same value, say, e.g., $\bw_{k,0}=0$, then the learning curves
		start at Transient Phase II directly.

		\item
		{\bf Transient Phase II:}\\
		In this phase, the trajectories of all agents are uniformly 
		close to the trajectory of the reference recursion; they
		converge in a coordinated manner to steady-state { in geometric (linear) rate}. The
		learning curves at this phase are well modeled by the 
		same reference recursion \eqref{Equ:LearnBehav:RefRec} since we showed in  
		\eqref{P1-Equ:MSStability:MSE_wki_PhaseII} from Part I
		\cite{chen2013learningPart1} that:
			\begin{align}
				\E\|\tilde{\w}_{k,i}\|^2	=	\|\tilde{w}_{c,i}\|^2  
									+ 
									O(\mu_{\max}^{1/2}) \cdot \gamma_c^i 
									+ 
									O(\mu_{\max})
				\label{Equ:MSStability:MSE_wki_PhaseII}
			\end{align}
		where the error vectors are defined by $\tilde{\w}_{k,i}=w^o-\w_{k,i}$ and $\tilde{w}_{c,i}=w^o-\bar{w}_{c,i}$. 
		Furthermore,  for small step-sizes and during the later stages of this phase, $\bar{w}_{c,i}$ will
		be close enough to $w^o$ and the convergence rate $r$ was shown in expression 
		\eqref{P1-Equ:DistProc:r_RefRec} from Part I\cite{chen2013learningPart1} to be given by
			\begin{align}
				r		&=		\big[
							\rho(I_M - \mu_{\max} H_c)
						\big]^2
						+ 
						O\big( (\mu_{\max} \epsilon )^{\frac{1}{2(M-1)}} \big)
				\label{Equ:LearnBehav:r}
			\end{align}
		where $\rho(\cdot)$ denotes the spectral radius of its matrix argument, $\epsilon$ is
		an arbitrarily small positive number, and $H_c$ is defined as the aggregate (Hessian-type) sum:
			\begin{align}
				H_c	\defeq	\sum_{k=1}^N p_k \nabla_{w^T} s_k(w^o)
				\label{Equ:LearnBehav:Hc_def}
			\end{align}
		
		\item
		{\bf Steady-State Phase:}\\
		The reference recursion \eqref{Equ:LearnBehav:RefRec}
		continues converging towards $w^o$ so that $\|\tilde{w}_{c,i}\|^2$ will converge to zero ($-\infty$ dB 
		in Fig. \ref{Fig:TypicalLearningCurve}).	
		However, for the distributed strategy \eqref{Equ:ProblemFormulation:DisAlg_Comb1}--%
		\eqref{Equ:ProblemFormulation:DisAlg_Comb2}, the mean-square-error
		$\E\|\tilde{\bw}_{k,i}\|^2$ at each
		agent $k$ will converge to a \emph{finite} steady-state value that is on the order of $O(\mu_{\max})$:
			\begin{align}
				\limsup_{i \rightarrow \infty} \E \| \tilde{\w}_{k,i} \|^2 	\le 	O(\mu_{\max})
				\label{Equ:LearnBehav:limsup_MSE_UB}
			\end{align}				
	\end{itemize}

Note that the bound \eqref{Equ:LearnBehav:limsup_MSE_UB} provides a partial answer to the fourth question we are interested in, namely, how close the $\w_{k,i}$ get to the network limit point $w^o$. Expression \eqref{Equ:LearnBehav:limsup_MSE_UB} indicates that the mean-square error is on the order of $\mu_{\max}$. However, in this Part II, we will examine this mean-square error more closely and provide a more accurate characterization of the steady-state MSE value by deriving a closed-form expression for it. In particular, we will be able to characterize this MSE value in terms of the vector $p$ as follows\footnote{The interpretation of the limit in \eqref{Equ:LearnBehav:MSE} is explained in more detail in Sec. \ref{Sec:LearnBehav:SteadyStateAnal}.}:
			\begin{align}
				\lim_{i\rightarrow\infty}\!
					\E\|\tilde{\bw}_{k,i}\|^2	
						&=	
							\mu_{\max}\!\cdot\!
							\Tr\left\{
								X
								(p^T \!\! \otimes\! I_M) \cdot \mc{R}_v \cdot (p\!\otimes\! I_M)
							\right\} 
							\nn\\
							&\quad
							+ 
							o(\mu_{\max})
				\label{Equ:LearnBehav:MSE}
			\end{align}
		where $X$ is the solution to a certain Lyapunov equation described
		later in \eqref{Equ:SteadyState:ContinuousLyapunovEqu_final}
		(when $\Sigma=I$), $\mR_v$ is a gradient noise covariance matrix defined below in 
		\eqref{Equ:Assumption:Rv:R_v_limit}, and $o(\mu_{\max})$ denotes a strictly higher order term
		of $\mu_{\max}$. 
		Expression \eqref{Equ:LearnBehav:MSE} is a most 
		revealing result; it captures the effect of the network topology through the eigenvector $p$, 
		and it captures the effects of gradient noise and data characteristics
		through the matrices $\mc{R}_v$ and $X$, respectively. Expression \eqref{Equ:LearnBehav:MSE}
		is a non-trivial extension of a classical 
		and famous result pertaining to the mean-square-error performance 
		of stand-alone adaptive agents \cite{widrow1976stationary,jones1982analysis,
		gardner1984learning,feuer1985convergence} to the more demanding 
		context of networked agents. In particular, it can be easily verified that 
		\eqref{Equ:LearnBehav:MSE} reduces to the well-known $\mu M\sigma_v^2/2$ 
		expression for the mean-square deviation of single LMS learners when the network 
		size is set to $N=1$ and the topology is removed \cite{widrow1976stationary,
		 jones1982analysis,gardner1984learning,feuer1985convergence}. However, expression 
		 \eqref{Equ:LearnBehav:MSE} is not limited to single agents or to mean-square-error costs. 
		 It applies to rather general connected networks and to fairly general cost functions.

\subsection{Relation to Prior Work}
\label{Sec:ProblemFormulation:PriorWork}

As pointed out in Part I\cite{chen2013learningPart1} (see Sec. \ref{P1-Sec:ProblemFormulation:PriorWork}), most prior works in the literature\cite{tsitsiklis1986distributed,
kar2011converegence,kar2008distributed,kar2013distributed,dimakis2010gossip,ram2010distributed,
srivastava2011distributed,nedic2009distributed,lee2013distributed,bianchi2012performance,johansson2008subgradient,braca2008running,stankovic2011decentralized} focus on studying the performance and convergence of their respective distributed strategies under {\em diminishing} step-size conditions and for {\em doubly-stochastic} combination policies. In contrast, we focus on \emph{constant} step-sizes in order to enable continuous adaptation and learning under drifting conditions. We also focus on \emph{left-stochastic combination matrices} in order to induce flexibility about the network limit point; this is because doubly-stochastic policies force the network to converge to the {\em same} limit point, while left-stochastic policies enable the networks to converge to any of infinitely many Pareto optimal solutions. Moreover, the value of the limit point can be controlled through the selection of the Perron eigenvector.

Furthermore, the performance of distributed strategies has usually been characterized in terms of bounds on their steady-state mean-square-error performance --- see, e.g.,  \cite{nedic2010cooperative,tsitsiklis1986distributed,ram2010distributed,srivastava2011distributed,nedic2009distributed,lee2013distributed,johansson2008subgradient,stankovic2011decentralized}. In Part I \cite{chen2013learningPart1} of the work, as a byproduct of our study of the three stages of the learning process, we were able to derive performance bounds for the steady-state MSE of a fairly general class of distributed strategies under broader (weaker) conditions than normally considered in the literature. In this Part II, we push the analysis noticeably further and derive a closed-form expression for the steady-state MSE in the slow adaptation regime, such as expression \eqref{Equ:LearnBehav:MSE}, which captures in an integrated manner how various network parameters (topology, combination policy, utilities) influence performance. 

Other useful and related works in the literature appear in \cite{kar2011converegence,kar2008distributed,kar2013distributed, bianchi2012performance}. These works, however, study the distribution of the error vector in steady-state under {\em diminishing} step-size conditions and using central limit theorem (CLT) arguments. They established a Gaussian distribution for the error quantities in steady-state and derived an expression for the error variance but the expression tends to zero as $i\rightarrow\infty$ since, under the conditions assumed in these works, the error vector $\tilde{\w}_{k,i}$ approaches zero almost surely. Such results are possible because, in the diminishing step-size case, the influence of gradient noise is annihilated by the decaying step-size. However, in the \emph{constant} step-size regime, the influence of gradient noise is always present and seeps into the operation of the algorithm. In this case, the error vector does {\em not} approach zero any longer and its variance  approaches instead a steady-state {\em positive-definite} value. Our objective is to characterize this steady-state value and to examine how it is influenced by the network topology, by the persistent gradient noise conditions, and by the data characteristics and utility functions.  In the constant step-size regime, CLT arguments cannot be employed anymore because the Gaussianity result does not hold any longer. Indeed, reference \cite{zhao2011probability} illustrates this situation clearly; it derived an expression for the characteristic function of the limiting error distribution in the case of mean-square-error estimation and it was shown that the distribution is not Gaussian. For these reasons, the analysis in this work is based on alternative techniques that do not pursue any specific form for the steady-state distribution and that rely instead on the use of energy conservation arguments \cite{chen2011TSPdiffopt,Sayed08,sayed2012diffbookchapter}.   As the analysis and detailed derivations in the appendices show, this is a formidable task to pursue due to the coupling among the agents and the persistent noise conditions. Nevertheless, under certain conditions that are generally weaker than similar conditions used in related contexts in the literature, we will be able to derive accurate expressions for the network MSE performance and its convergence rate in small constant step-size regime.

We finally remark that the analysis in this paper and its accompanying Part I\cite{chen2013learningPart1} is \emph{not} focused on the solution of \emph{deterministic} distributed optimization problems, although algorithm \eqref{Equ:ProblemFormulation:DisAlg_Comb1}--\eqref{Equ:ProblemFormulation:DisAlg_Comb2} can still be applied for that purpose (see future Sec. \ref{Sec:Benefits:ParetoOpt}). Instead, we consider a stochastic setting where each individual cost $J_k(w)$ is generally expressed as the expectation of some loss function, say, as
	\begin{align}
		J_k(w)	=	\Expt Q_k(w; \bm{x}_{k,i})
		\label{Equ:FamilyDistStrategy:J_k}
	\end{align}
and the objective is to minimize the aggregate stochastic cost:
	\begin{align}
		J^{\mathrm{glob}}(w)		=	\sum_{k=1}^N J_k(w)
		\label{Equ:FamilyDistStrategy:J_glob}
	\end{align}
In such problems, we usually do not know the exact form of the cost function because we do not have prior knowledge about the exact statistical distribution of the data $\bm{x}_{k,i}$. What is generally available to each agent $k$ is a stream of data points $\bm{x}_{k,0}, \bm{x}_{k,1}, \ldots$ that arrives at agent $k$ sequentially over time. The agents in the network then use stochastic gradients constructed as $\nabla Q(w; \bm{x}_{k,i})$ (or from variations thereof), in place of the the actual gradients, $\nabla J_k(w)$, to learn from the streaming data.  Because of the stochastic nature of the learning algorithms, they will exhibit different convergence behavior than deterministic optimization algorithms. For example, even with a constant step-size, stochastic gradient distributed strategies can still converge at a geometric rate towards a small MSE in steady-state,  whereas diminishing step-sizes of the form $\mu(i)=\mu_o/i$, ensure a slower  almost sure convergence rate of $O(1/i)$.

\section{MODELING ASSUMPTIONS}
\label{Sec:Assump}

In this section, we first recall the assumptions used in Part I \cite{chen2013learningPart1} and then introduce two conditions that are required to carry out the MSE analysis in this part. We already explained in Sec. \ref{Sec:Assump} of Part I \cite{chen2013learningPart1} how the assumptions listed below relate to, and extend, similar conditions used in the literature. 
\vspace{0.5em}
	\begin{assumption}[Strongly-connected network]
		\label{Assumption:Network}
		The $N \times N$ matrix product $A \defeq A_1A_0A_2$ is assumed to be a primitive 
		left-stochastic matrix, i.e., $A^T\mathds{1}=\mathds{1}$ 
		and there exists a finite integer $j_o$ such that 
		all entries of $A^{j_o}$ are strictly positive. 
		\hfill\QED
	\end{assumption}
\vspace{0.5em}	

	\vspace{0.5em}
	\begin{assumption}[Update vector: Randomness]
		\label{Assumption:UpdateVectorRandomness}
		There exists an $M \times 1$ deterministic vector function 
		$s_k(w)$ such that, for all $M \times 1$ vectors $\bm{w}$ in the filtration $\mc{F}_{i-1}$
		generated by  the past history of iterates $\{\bm{w}_{k,j}\}$ for $j \le i-1$
		and all $k$, it holds that
			\begin{align}
				\label{Equ:Assumption:Randomness:MDS}
				&\E\left\{
					\hat{\bm{s}}_{k,i}(\bm{w}) | \mc{F}_{i-1}
				\right\}
						=	s_k(\bm{w})						
			\end{align}
		for all $i,k$. Furthermore, there exist
		$\alpha \ge 0$ and $\sigma_{v}^2 \ge 0$ such that
		for all $i,k$ and $\bm{w} \in \mc{F}_{i-1}$:
			\begin{align}
				\label{Equ:Assumption:Randomness:RelAbsNoise}
				&\E\left\{\!
					\left\| 
						\hat{\bm{s}}_{k,i}(\bm{w})\!-\!s_{k}(\bm{w})
					\right\|^2
					\big|
					\mF_{i-1}
				\right\}
						\le \alpha
							\!\cdot\!
							\| \w \|^2 \!+\! \sigma_{v}^2							
			\end{align}
		holds with probability one.
		\hfill\QED		
	\end{assumption}
	\vspace{0.5em}
	
	\vspace{0.5em}
	\begin{assumption}[Update vector: Lipschitz]
		\label{Assumption:UpdateVectorLipschitz}
		There exists a nonnegative $\lambda_{U}$ such that for 
		all $x,y \in \mb{R}^M$ and all $k$:
			\begin{align}
				\label{Equ:Assumption:Lipschitz}
				\|s_k(x)-s_k(y)\|	\le 	\lambda_{U} \cdot \|x-y\|
			\end{align}
		where the subscript ``$U$'' in $\lambda_{U}$ means
		``upper bound''.
		\hfill\QED
	\end{assumption}	
	\vspace{0.5em}

	\vspace{0.5em}
	\begin{assumption}[Update vector: Strong monotonicity]
		\label{Assumption:UpdateVectorMonot}
		Let $p_k$ denote the $k$th entry of the vector $p$ 
		defined in \eqref{Equ:ProblemForm:p_def}.
		There exists $\lambda_L > 0$ such that for
		all $x,y \in \mb{R}^M$:
			\begin{align}
				\label{Equ:Assumption:StrongMonotone}
				(x-y)^T \cdot
						\sum_{k=1}^N
						p_k 
						\Big[
								s_k(x)-s_k(y)
						\Big]
						\ge 	\lambda_L \cdot \|x-y\|^2
			\end{align}
		where the subscript ``$L$'' in $\lambda_L$ means ``lower bound'',
		{ 
		and $\lambda_L$ may depend on $\{p_k\}$.
		}
		\hfill\QED
	\end{assumption}
	\vspace{0.5em}

		\begin{assumption}[Jacobian matrix: Lipschitz]
		\label{Assumption:JacobianUpdatVectorLipschitz}
		Let $w^o$ denote the limit point of the distributed strategy
		\eqref{Equ:ProblemFormulation:DisAlg_Comb1}--\eqref{Equ:ProblemFormulation:DisAlg_Comb2},
		which was defined earlier as 
		the unique solution to \eqref{Equ:LearnBehav:FixedPointEqu} and was characterized in Theorem \ref{P1-Thm:LimitPoint} 
		of Part I \cite{chen2013learningPart1}. Then, in a small 
		neighborhood around $w^o$, we assume that $s_k(w)$ is differentiable with 
		respect to $w$ and satisfies
			\begin{align}
				\label{Equ:Assumption:LipschitzJacobian}
				\|\nabla_{w^T} s_k(w^o+\delta w)-\nabla_{w^T} s_k(w^o)\|\leq\;\lambda_H\cdot \|\delta w\|	
			\end{align}
		for all $\|\delta w\|\leq r_H$ for some small $r_H$, and where $\lambda_H$ is a 
		nonnegative number independent of $\delta w$. 
		
		\hfill\QED
	\end{assumption}
	\vspace{0.5em}
The following lemma gives the equivalent forms of Assumptions \ref{Assumption:UpdateVectorLipschitz}--\ref{Assumption:UpdateVectorMonot} when the $\{s_k(w)\}$ happen to be differentiable.
	\begin{lemma}[Equivalent conditions on update vectors]
		\label{Lemma:EquivCond_updateVec}
		Suppose $\{s_k(w)\}$ are differentiable in an open set $\mS \subseteq \mb{R}^M$. Then,
		having conditions \eqref{Equ:Assumption:Lipschitz} and \eqref{Equ:Assumption:StrongMonotone} hold on $\mS$ is 
		equivalent to the following conditions, respectively,
			\begin{align}
				\| \nabla_{w^T} s_k(w) \|		&\le 		\lambda_U				
				\label{Equ:Lemma:EquivCondUpdateVec:LipschitzUpdate_HessianUB}
															\\
				\frac{1}{2} [ H_c(w) + H_c^T(w) ] 
												&\ge 		\lambda_L \cdot I_M
				\label{Equ:Lemma:EquivCondUpdateVec:StrongMono_HessianLB}
			\end{align}
		for any $w \in \mS$, where $\|\cdot\|$ denotes the $2$-induced norm (largest singular value) of 
		its matrix argument and 
			\begin{align}
				H_c(w)		\defeq		\sum_{k=1}^n p_k \nabla_{w^T} s_k(w)
				\label{Equ:Lemma:EquivCondUpdateVec:Hc_def}
			\end{align}
	\end{lemma}
	\begin{proof}
		See Appendix \ref{P1-Appendix:Proof_Lemma_EquivCondUpdateVec} in Part I\cite{chen2013learningPart1}.
	\end{proof}

Next, we introduce two new assumptions on $\hat{\bs}_{k,i}(\bw)$, which are needed for the MSE analysis of this Part II.  Assumption \ref{Assumption:Rv} below has been used before in the stochastic approximation literature --- see, for example, \cite{sacks1958asymptotic} and  Eq. (6.2) in Theorem 6.1 of \cite[p.147]{nevelson1972stochastic}. 
{
Before we state the assumptions, we first introduce some useful quantities. Let $\bm{v}_i(x)$ denote the $MN \times 1$ global vector that collects the statistical fluctuations in the stochastic update vectors across all agents:
				\begin{align}
					{\bm{v}}_{i}(x) \defeq 	
									\col\{ 
												\hat{\bm{s}}_{1,i}(x_1)-s_{1}(x_1), \;
												\ldots, \;
												\hat{\bm{s}}_{N,i}(x_N)-s_{N}(x_N)
										\}
					\label{Equ:Assumption:Rv:vi_def}
				\end{align}
		where we are using the vector $x$ to denote a block vector consisting of entries $x_k$ of 
		size $M\times 1$ each, i.e., $x \defeq \col\{ x_1, \ldots, x_N\}$. 
		For any $\x_{k} \in \mc{F}_{i-1}$, $1 \le k \le N$, we introduce the covariance matrix:
				\begin{align}
					\mc{R}_{v,i}(\bm{x})	\defeq	\E\left\{ {\bm{v}}_{i}(\bm{x}) 
												 {\bm{v}}_{i}^T(\bm{x})
												 \big|
												 \mc{F}_{i-1}
											\right\}
					\label{Equ:Assumption:Rv:Rvi_def}
				\end{align}
		where, again, we are using the notation $\x$ to refer to the block vector 
		$\x=\mbox{\rm col}\{\x_1,\ldots,\x_N\}$ with stochastic entries of size $M\times 1$ each. 
		Note that $\mR_{v,i}(\bm{x})$ generally depends on time $i$.
		This is because the distribution of $\hat{\bm{s}}_{k,i}(\cdot)$
		given $\mc{F}_{i-1}$ usually varies with time. The following assumption requires that,
		in the limit, this second-order moment of the distribution tends to a constant value.
		
}
	\begin{assumption}[Second-order moment of gradient noise]
		\label{Assumption:Rv}
		We assume that, 
		in the limit, $\mR_{v,i}(\bm{x})$ becomes
		invariant and tends to a deterministic constant value when evaluated at $\x=\one\otimes w^o$ 
		with probability one (almost surely):
			\begin{align}
				\lim_{i\rightarrow\infty} \mc{R}_{v,i}(\one \otimes w^o)		\defeq	\mc{R}_v
				\label{Equ:Assumption:Rv:R_v_limit}
			\end{align}		
		Furthermore, in a small 
		neighborhood around $\one \otimes w^o$, we assume that there exists deterministic
		constants $\lambda_v \ge 0$, 
		$r_V > 0$, and $\kappa \in (0, 4]$ such that for all $i \ge 0$:
			\begin{align}
				\label{Equ:Assumption:RvLipschitz}
				\left\|\mc{R}_{v,i}(\one \otimes w^o + \delta x) - \mc{R}_{v,i}(\one \otimes w^o)\right\|
						&\le 		\lambda_v \cdot \| \delta x\|^\kappa
			\end{align}
		for all $\|\delta x\| \le r_V$ with probability one. \hfill\QED
	\end{assumption}

	\begin{example}
		\label{Ex:Rv_LMS}
		We illustrate how Assumption \ref{Assumption:Rv} holds automatically
		in the context of distributed least-mean-squares estimation.
		Suppose each agent $k$ receives a stream
		of data samples $\{\bu_{k,i},\bd_{k}(i)\}$
		that are generated by the following linear model:
			\begin{align}
				\bd_{k}(i)	=	\bu_{k,i} w^o + \bn_{k}(i)
				\label{Equ:Example:LinearModel}
			\end{align}
		where the $1 \times M$ regressors $\{\bm{u}_{k,i}\}$ are zero mean and independent over time
		and space  with
	    covariance matrix $R_{u,k}=\E \{\bm{u}_{k,i}^T\bm{u}_{k,i}\} \ge 0$ and
	    the noise sequence $\{\bm{n}_l(j)\}$ is also zero mean, white, with variance 
	    $\sigma_{n,l}^2$, and independent of the regressors $\{\bm{u}_{k,i}\}$
	    for all $l,k,i,j$. 
	    The objective is to estimate the $M\times 1$ 
    		parameter vector $w^o$ by minimizing the following global cost function
			\begin{align}
				J^{\mathrm{glob}}(w)		=	\sum_{k=1}^N J_k(w)
				\label{Equ:ProblemForm:J_glob_def}
			\end{align}
    		where
    			\begin{align}
    				\label{Equ:Example:J_k_LMS}
    				J_k(w)	=	\E|\bd_k(i)-\bu_{k,i} w|^2
    			\end{align}
    		In this case, the actual gradient vector when evaluated at an $M\times 1$ vector $x_k$ 
		is given by
    			\begin{align}
				s_{k}(x_k)=\nabla_w \E|\bd_k(i)-\bu_{k,i}x_k|^2
			\end{align}
		and it can be
    		replaced by the instantaneous approximation
    			\begin{align}
    				\hat{\bs}_{k,i}(x_k)
    						=   -2\bm{u}_{l,i}^T[\bm{d}_{l}(i)-\bm{u}_{l,i} x_k]
    				\label{Equ:Example:ski_LMS}
    			\end{align}
		(Recall from \eqref{Equ:ProblemFormulation:DisAlg_Adapt} that the stochastic gradient at each agent
		$k$ is evaluated at $\phi_{k,i-1}$ and in this case $x_k = \phi_{k,i-1}$.)
		It follows that the gradient noise vector $\bv_{k,i}(x_k)$ evaluated at $x_k$, at
		each agent $k$ is given by
	        \begin{align}
	            \label{Equ:PerformanceAnalysis:GradientNoise_LMS}
	            \bm{v}_{k,i}(x_k) =   2(R_{u,k}-\bm{u}_{k,i}^T\bm{u}_{k,i})
	            						(w^o-x_k) - 2\bm{u}_{k,i}^T \bm{n}_k(i)
	        \end{align}
	    and it is straightforward to verify that
	    		\begin{align}
	    			\mc{R}_{v,i}(\one \otimes w^o)=	\diag\{
	    										4 \sigma_{n,1}^2 R_{u,1},
	    										\cdots,
	    										4 \sigma_{n,N}^2 R_{u,N}
	    									\}
	    			\label{Equ:Example:Rv_LMS:Rv}
	    		\end{align}
	    	which is independent of $i$ and, therefore, condition
	    	\eqref{Equ:Assumption:Rv:R_v_limit} holds with $\mc{R}_v$ given by
	    	\eqref{Equ:Example:Rv_LMS:Rv}. 
	    	 Furthermore, condition \eqref{Equ:Assumption:RvLipschitz} is also satisfied. 
		 Indeed, let $x = \col\{x_1,\ldots, x_N\} \in \mb{R}^{MN}$, and from 
		 \eqref{Equ:PerformanceAnalysis:GradientNoise_LMS} we find that
	    	 	\begin{align}
	    	 		\mc{R}_{v,i}(x)
	    	 			&=		\diag\{ G_1,\ldots,G_N\}
    							+
    							\mc{R}_{v,i}(\one \otimes w^o)
				\nn
	    	 	\end{align}
	    	 where each $G_k$ is a function of $w^o-x_k$ and is given by
	    	 	\begin{align}
	    	 		G_k			&\defeq	4 \cdot
	    	 							\E\big\{ 
	    	 								(R_{u,k}-\bm{u}_{k,i}^T\bm{u}_{k,i})
	            								(w^o-x_k)
										\nn\\
										&\qquad\quad
										(w^o - x_k)^T
	            								(R_{u,k}-\bm{u}_{k,i}^T\bm{u}_{k,i})^T
	            							\big\}
								\nn
	    	 	\end{align}
	    	 Note that
	    	 	\begin{align}
	    	 		\|G_k\|			&\le 	4 \cdot							
	    	 								\E\left\| 
	    	 									R_{u,k}-\bm{u}_{k,i}^T\bm{u}_{k,i}
	    	 								\right\|^2
	    	 								\cdot
	    	 								\| w^o-x_k\|^2
										\nn
	    	 	\end{align}
	    	 so that
	    	 	\begin{align}
	    	 		\big\| &\mc{R}_{v,i}(x) - \mc{R}_{v,i}(\one \otimes w^o)\big\|
							\nn\\
	    	 			&=	 	\max_{1\le k \le N}\|G_k\|
	    	 					\nonumber\\
	    	 			&\le 	\max_{1\le k \le N}
	    	 					\left\{
	    	 							4 \cdot						
    	 								\E\| 
    	 									R_{u,k}-\bm{u}_{k,i}^T\bm{u}_{k,i}
    	 								\|^2
    	 								\cdot
    	 								\| w^o-x_k\|^2
    	 						\right\}
    	 						\nonumber\\
    	 				&\le 	\max_{1\le k \le N}
	    	 					\left\{
	    	 							4 \cdot						
    	 								\E\| 
    	 									R_{u,k}-\bm{u}_{k,i}^T\bm{u}_{k,i}
    	 								\|^2
    	 						\right\}
    	 						\cdot
    	 						\max_{1\le k \le N}
    	 						\| w^o-x_k\|^2
    	 						\nonumber\\
    	 				&\le 	\max_{1\le k \le N}
	    	 					\left\{
	    	 							4 \cdot						
    	 								\E\| 
    	 									R_{u,k}-\bm{u}_{k,i}^T\bm{u}_{k,i}
    	 								\|^2
    	 						\right\}
    	 						\cdot
    	 						\sum_{k=1}^N
    	 						\| w^o-x_k\|^2
    	 						\nonumber\\
    	 				&=		\max_{1\le k \le N}
	    	 					\left\{
	    	 							4 \cdot						
    	 								\E\| 
    	 									R_{u,k} \!-\! \bm{u}_{k,i}^T\bm{u}_{k,i}
    	 								\|^2
    	 						\right\}
    	 						\cdot
    	 						\| \one \otimes w^o \!-\! x\|^2
				\label{Equ:Example:Rv_LMS:Rv_Lipschitzbound}
	    	 	\end{align}
	    	 In other words,
	    	 condition \eqref{Equ:Assumption:RvLipschitz} holds for the least-mean-squares estimation case with $\kappa = 2$. 
		\hfill\QED
	\end{example}
	\vspace{0.5em}
	\begin{assumption}[Fourth-order moment of gradient noise]
		\label{Assumption:GradientNoise4thOrderMoment}
		There exist nonnegative numbers $\alpha_4$ and $\sigma_{v4}^2$ such that for any $M \times 1$
		random vector $\w \in \mF_{i-1}$,
			\begin{align}
				\E \left\{
					\|\bv_{k,i}(\w)\|^4
				\big| \mF_{i-1}
				\right\}
						 	&\le 	\alpha_4 \cdot \|\w\|^4 + \sigma_{v4}^4
				\label{Equ:Assumption:GradientNoise4thOrderMoment}
			\end{align}
		holds with probability one.
		\hfill \QED
	\end{assumption}
This assumption will be used in the analysis for constant step-size adaptation to arrive at accurate expressions for the steady-state MSE of the agents. By assuming that the fourth-order moment of the gradient noise is bounded as in \eqref{Equ:Assumption:GradientNoise4thOrderMoment}, it becomes possible to derive MSE expressions that can be shown to be at most $O\big(\mu_{\max}^{\min(3/2, 1+\kappa/2)}\big)$ away from the actual MSE performance. When the step-sizes are sufficiently small, the size of the term $O(\mu_{\max}^{\min(3/2, 1+\kappa/2)})$ is even smaller and, for all practical purposes, this term is negligible --- see expressions \eqref{Equ:SteadyState:WMSE_limsup_finalfinal}--\eqref{Equ:SteadyState:WMSE_liminf_finalfinal} in Theorem \ref{Thm:SteadyStatePerformance} (and also \eqref{Equ:SteadyState:WMSE_final}).

\begin{example}
\label{Ex:LMS_4thOrder}

It turns out that condition \eqref{Equ:Assumption:GradientNoise4thOrderMoment} is automatically satisfied in the context of distributed least-mean-squares estimation. We continue with the setting of Example \ref{Ex:Rv_LMS}.  From expression \eqref{Equ:PerformanceAnalysis:GradientNoise_LMS}, we have that for any $M \times 1$ random vector $\w \in \mF_{i-1}$,
	\begin{align}
		\| \bv_{k,i}( \w ) \|^4		
							&=	
									16
									\left\|
										(R_{u,k}-\bm{u}_{k,i}^T\bm{u}_{k,i})
	            								(w^o-\w) 
										- 
										\bm{u}_{k,i}^T \bm{n}_k(i)
									\right\|^4
									\nn\\
							&\overset{(a)}{\le}
									16 \times 8
									\Big(
										\left\|
											R_{u,k}-\bm{u}_{k,i}^T\bm{u}_{k,i}
										\right\|^4
										\cdot
										\left\|
		            								w^o-\w
										\right\|^4
										\nn\\
										&\qquad\qquad
										+
										\left\|
											\bm{u}_{k,i}
										\right\|^4
										\cdot
										\left\|
											\bm{n}_k(i)
										\right\|^4
									\Big)
									\nn\\
							&\overset{(b)}{\le}
									128
									\Big(
										8
										\left\|
											R_{u,k}-\bm{u}_{k,i}^T\bm{u}_{k,i}
										\right\|^4
										\cdot
										\| \w \|^4
										\nn\\
										&\qquad\qquad
										+
										8
										\left\|
											R_{u,k}-\bm{u}_{k,i}^T\bm{u}_{k,i}
										\right\|^4
										\cdot
										\| w^o \|^4
										\nn\\
										&\qquad\qquad
										+
										\left\|
											\bm{u}_{k,i}
										\right\|^4
										\cdot
										\left\|
											\bm{n}_k(i)
										\right\|^4
									\Big)
		\label{Equ:Ex:gradNoise_4thOrder_interm1}
	\end{align}
where steps (a) and (b) use the inequality $\| x + y \|^4 \le 8 \| x \|^4 + 8 \| y\|^4$, which can be obtained by applying Jensen's inequality to the convex function $\| \cdot \|^4$. Applying the expectation operator conditioned on $\mF_{i-1}$, we obtain
	\begin{align}
		\E &\left\{ \| \bv_{k,i}( \w ) \|^4 | \mF_{i-1} \right\}
									\nn\\
							&\overset{(a)}{\le}
									1024
									\cdot
									\E
									\left\{
										\left\|
											R_{u,k}-\bm{u}_{k,i}^T\bm{u}_{k,i}
										\right\|^4
										|
										\mF_{i-1}
									\right\}
									\cdot
									\| \w \|^4
									+
									\nn\\
									&\quad\;\;
									1024
									\cdot
									\E
									\left\{
										\left\|
											R_{u,k}-\bm{u}_{k,i}^T\bm{u}_{k,i}
										\right\|^4
										|
										\mF_{i-1}
									\right\}
									\cdot
									\| w^o \|^4
									+
									\nn\\
									&\quad\;\;
									128
									\cdot
									\E 
									\left\{
										\left\|
											\bm{u}_{k,i}
										\right\|^4
										|
										\mF_{i-1}
									\right\}
									\cdot
									\E
									\left\{
										\left\|
											\bm{n}_k(i)
										\right\|^4
										|
										\mF_{i-1}
									\right\}
									\nn\\
						&\overset{(b)}{=} 
									1024
									\cdot
									\E
									\left\{
										\left\|
											R_{u,k}-\bm{u}_{k,i}^T\bm{u}_{k,i}
										\right\|^4
									\right\}
									\cdot
									\| \w \|^4
									+
									\nn\\
									&\quad\;\;
									1024
									\cdot
									\E
									\left\{
										\left\|
											R_{u,k}-\bm{u}_{k,i}^T\bm{u}_{k,i}
										\right\|^4
									\right\}
									\cdot
									\| w^o \|^4
									+
									\nn\\
									&\quad\;\;
									128
									\cdot
									\E 
									\left\|
										\bm{u}_{k,i}
									\right\|^4
									\cdot
									\E
									\left\|
										\bm{n}_k(i)
									\right\|^4
									\nn\\
						&\defeq
									\alpha_4 \cdot \|\w\|^4 + \sigma_{v4}^4
									\nn
	\end{align}
where step (a) uses the fact that $\w \in \mF_{i-1}$ and is thus determined given $\mF_{i-1}$, and step (b) uses the fact that $\u_{k,i}$ and $\bv_{k,i}(i)$ are independent of $\mF_{i-1}$.
\hfill \QED
\end{example}

\section{Performance of Multi-Agent Learning Strategy}
\label{Sec:LearnBehav:SteadyStateAnal}

\subsection{Main Results}
\label{Sec:LearnBehav:SteadyStateAnal:MainResults}

In this section, we are interested in evaluating 
$\E\|\tilde{\bw}_{k,i}\|_{\Sigma}^2$ as $i\rightarrow\infty$
for arbitrary positive semi-definite weighting matrices $\Sigma$.
The main result is summarized in the following theorem.

\begin{figure}
	\centering
	\includegraphics[width=0.45\textwidth]{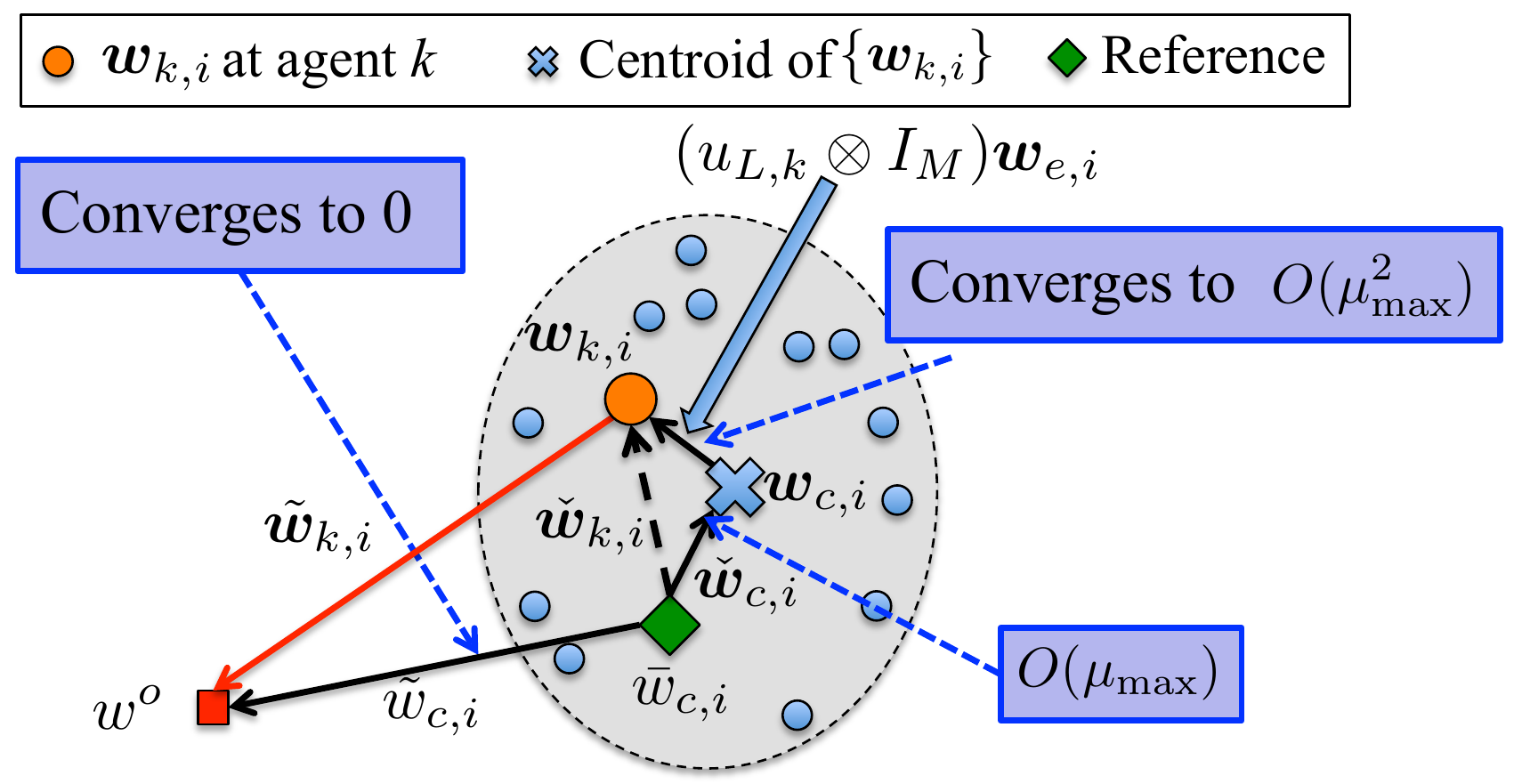}
	\caption{
	Decomposition of the error quantity $\tilde{\w}_{k,i}$ (best viewed in color). The  error
	quantity $\tilde{\w}_{k,i}$ (solid red line) can be decomposed into three terms (the solid back
	lines): (i) the offset of $\w_{k,i}$ from the centroid of $\{\w_{k,i}\}$, 
	denoted by $(u_{L,k} \otimes I_M) \bm{w}_{e,i}$, (ii) the offset of the centroid
	from the reference recursion \eqref{Equ:LearnBehav:RefRec}, denoted by $\check{\w}_{c,i}$, 
	and (iii) the error between the reference recursion and the optimal solution $w^o$, denoted
	by $\tilde{w}_{c,i}$.
	}
	\label{Fig:ErrorDecomposition}
\end{figure}

	\begin{theorem}[Steady-state performance]
		\label{Thm:SteadyStatePerformance}
		When Assumptions \ref{Assumption:Network}--\ref{Assumption:GradientNoise4thOrderMoment} 
		hold and the step-sizes are sufficiently small so that the distributed strategy
		\eqref{Equ:ProblemFormulation:DisAlg_Comb1}--\eqref{Equ:ProblemFormulation:DisAlg_Comb2} 
		is mean-square stable\footnote{The explicit condition for mean-square stability is given by 
		\eqref{P1-Equ:Thm_NonAsympBound:StepSize}  in Part I \cite{chen2013learningPart1}.}, the weighted mean-square-error of
		\eqref{Equ:ProblemFormulation:DisAlg_Comb1}--\eqref{Equ:ProblemFormulation:DisAlg_Comb2}
		(which includes diffusion and consensus algorithms as special cases)
		satisfies
			\begin{align}
				\limsup_{i\rightarrow\infty} \E\|\tilde{\bw}_{k,i}\|_{\Sigma}^2
						&\le		
								\mu_{\max}\cdot
								\Tr\left\{
										X (p^T\otimes I_M)\cdot \mc{R}_v \cdot (p \otimes I_M)
								\right\}
								\nn\\
								&\quad
								+
								O\big(\mu_{\max}^{\min(3/2, 1+\kappa/2)}\big)
				\label{Equ:SteadyState:WMSE_limsup_finalfinal}
								\\
					\liminf_{i\rightarrow\infty} \E\|\tilde{\bw}_{k,i}\|_{\Sigma}^2
						&\ge		
								\mu_{\max}\cdot
								\Tr\left\{
										X (p^T\otimes I_M)\cdot \mc{R}_v \cdot (p \otimes I_M)
								\right\}
								\nn\\
								&\quad
								-
								O\big(\mu_{\max}^{\min(3/2, 1+\kappa/2)}\big)
					\label{Equ:SteadyState:WMSE_liminf_finalfinal}
			\end{align}
		where $\Sigma$ is any positive semi-definite weighting matrix, and
		$X$ is the unique positive semi-definite solution to the following 
		Lyapunov equation:
			\begin{align}
				H_c^T X + X H_c = \Sigma
				\label{Equ:SteadyState:ContinuousLyapunovEqu_final}
			\end{align}
		where $H_c$ was defined earlier in \eqref{Equ:LearnBehav:Hc_def}.
		The unique solution of \eqref{Equ:SteadyState:ContinuousLyapunovEqu_final}
		can be represented by the integral expression\cite[p.769]{kailath2000linear}:
			\begin{align}
				X	=	\int_0^{\infty} 
						e^{-H_c^T t} \cdot \Sigma \cdot e^{-H_c t} 
						dt
				\label{Equ:SteadyState:ContinuousLyapunovEqu_Solution}
			\end{align}
		Moreover, if $\Sigma$ is strictly positive-definite, then $X$ is also
		strictly positive-definite.
	\end{theorem}
	\begin{proof}
		The argument is nontrivial and involves several steps. The details are provided in 
		Appendix \ref{Appendix:Proof_Thm_SteadyStatePerf}. We briefly describe the main 
		steps of the proof here: 
			\begin{enumerate}
			
			\item
			By following the network transformation introduced in Part I \cite{chen2013learningPart1},
			we decompose the error vector $\tilde{\w}_{k,i}$ into three terms, as illustrated in Fig. 
			\ref{Fig:ErrorDecomposition}: (i) $(u_{L,k} \otimes I_M) \w_{e,i}$, the offset of $\w_{k,i}$ 
			from the centroid of $\{\w_{k,i}\}$, defined as
				\begin{align}
					\bw_{c,i}	=	\sum_{k=1}^N \theta_k \bw_{k,i}
								\nn
				\end{align}
			where $\theta_k$ is the $k$th entry of the Perron vector defined in 
			\eqref{Equ:ProblemForm:PerronVector_A},
			(ii) $\check{\w}_{c,i}$, the offset of the centroid from the reference recursion 
			\eqref{Equ:LearnBehav:RefRec}, 
			and (iii) $\tilde{w}_{c,i}$, the error between the reference recursion and the optimal solution $w^o$. 
			
			\item
			Only the second term, $\check{\w}_{k,i}$, contributes to the steady-state MSE,
			which we already know from \eqref{Equ:LearnBehav:limsup_MSE_UB}
			(see also \eqref{P1-Equ:MSStability:limsup_MSE_UB} in Part I \cite{chen2013learningPart1})
			that it is $O(\mu_{\max})$. For the other two terms, $\tilde{w}_{c,i}$ converges to zero
			and $(u_{L,k} \otimes I_M) \w_{e,i}$ converges to a higher-order term in $\mu_{\max}$. 
			In Sections A and B of Appendix
			\ref{Appendix:Proof_Thm_SteadyStatePerf}, we make this argument rigorous by
			deriving the gap between
			the error covariance matrices of $\tilde{\w}_{k,i}$ and $\check{\w}_{c,i}$ and showing that
			it is indeed a higher-order term.
			
			\item
			Next, we show that the recursion for $\check{\w}_{c,i}$ can be viewed as a perturbed
			version of a \emph{linear} dynamic system driven by the gradient noise term. 
			In Section C of Appendix 
			\ref{Appendix:Proof_Thm_SteadyStatePerf}, we bound the gap between
			these two recursions and show that it is also a higher-order term. This would require us to 
			bound the fourth-order moments of the error quantity $\tilde{\w}_{k,i}$, which are derived
			in Appendices \ref{Appendix_SketchProof_4thOrderMoment}--\ref{Appendix:Proof_UsefulBound}.
				
			\item	
			Then, in Section D of Appendix \ref{Appendix:Proof_Thm_SteadyStatePerf}, we examine 
			the covariance matrix of the linear dynamic model and find
			a closed-form expression for it. 	 
			
			\item
			Finally, in Section E of Appendix \ref{Appendix:Proof_Thm_SteadyStatePerf}, we combine
			all results together to obtain the closed-form expression for the steady-state MSE of the network.
			
			\end{enumerate}
			
	\end{proof}

Strictly speaking, the limit of $\E\|\tilde{\w}_{k,i}\|_{\Sigma}^2$ may not exist as it requires the $\limsup$ and the $\liminf$ of $\E\|\tilde{\w}_{k,i}\|_{\Sigma}^2$ to be equal to each other. However, note from \eqref{Equ:SteadyState:WMSE_limsup_finalfinal} and \eqref{Equ:SteadyState:WMSE_liminf_finalfinal} that the first-order terms of $\mu_{\max}$ in both $\limsup$ and $\liminf$ expressions are the same. When the step-size $\mu_{\max}$ is small, the $\limsup$ and the $\liminf$ bounds will be dominated by this same first-order term, and the steady-state MSE will be tightly sandwiched between \eqref{Equ:SteadyState:WMSE_limsup_finalfinal} and \eqref{Equ:SteadyState:WMSE_liminf_finalfinal}.\footnote{Recall that we always have $\displaystyle\liminf_{i\rightarrow\infty} \E\|\tilde{\bw}_{k,i}\|_{\Sigma}^2 \le \limsup_{i\rightarrow\infty} \E\|\tilde{\bw}_{k,i}\|_{\Sigma}^2$.} For this reason, with some slight abuse in notation, we will use the traditional limit notation for simplicity of presentation and will write instead:
	\begin{align}
		\lim_{i\rightarrow\infty} \E\|\tilde{\bw}_{k,i}\|_{\Sigma}^2
			&=	\mu_{\max}\!\cdot\!
				\Tr\left\{
					X (p^T\! \otimes\! I_M) \cdot \mc{R}_v \cdot (p\! \otimes\! I_M)
				\right\}
				\nn\\
				&\quad
				+\! 
				O\big(\mu_{\max}^{\min(3/2, 1+\kappa/2)}\big)
		\label{Equ:SteadyState:WMSE_final}
	\end{align}
\emph{Remark:} Note from \eqref{Equ:SteadyState:WMSE_final} that the steady-state MSE consists of two terms: a first-order term, and a higher-order term. We will show in Sec. \ref{Sec:LearnBehav:GlobalBehav} that the first-order term is the same as that of the centralized MSE.

\subsection{Useful Special Cases}
\label{Sec:LearnBehav:SteadyStateAnal:UsefulSpecialCases}

\begin{example} \emph{(Distributed stochastic gradient-descent:  General case)}
	\label{Example:DSGD}
	When stochastic gradients are used to define the update directions 
	$\hat{\s}_{k,i}(\cdot)$ in \eqref{Equ:ProblemFormulation:DisAlg_Comb1}--\eqref{Equ:ProblemFormulation:DisAlg_Comb2}, then we can simplify
	the mean-square-error expression \eqref{Equ:SteadyState:WMSE_final} as follows.
	We first substitute $s_k(w)=\nabla_w J_k(w)$
	into \eqref{Equ:LearnBehav:Hc_def} to obtain
		\begin{align}
			H_c		=		\sum_{k=1}^N p_k \nabla_w^2 J_k(w^o)
							\nn
		\end{align}
	Now the matrix $H_c$ is the weighted sum of the Hessian matrices of
	the individual costs $\{J_k(w)\}$ and is therefore symmetric.
	Then, the Lyapunov equation \eqref{Equ:SteadyState:ContinuousLyapunovEqu_final}
	becomes
		\begin{align}
			H_c X + X H_c = \Sigma
			\label{Equ:Example:DSGD:LyapunovEqu}
		\end{align}
	We have simple solutions to \eqref{Equ:Example:DSGD:LyapunovEqu} for the
	following two choices of $\Sigma$:
		\begin{enumerate}
			\item 
				When $\Sigma = I_M$, we have $X=\frac{1}{2}H_c^{-1}$ and 
					\begin{align}
						\lim_{i \rightarrow \infty}
						&\E\|\tilde{\bm{w}}_{k,i}\|^2
										\nn\\
							&=			\frac{\mu_{\max}}{2} \cdot
										\mathrm{Tr}
										\left\{
											H_c^{-1}
											(p^T \otimes I_M)
											\cdot
											\mc{R}_v
											\cdot
											(p \otimes I_M)
										\right\}
										\nn\\
										&\quad
										+
										O\big(\mu_{\max}^{\min(3/2, 1+\kappa/2)}\big)
						\label{Equ:Performance:MSE_I}
					\end{align}
			\item
				When $\Sigma = \frac{1}{2}H_c$, we have $X=\frac{1}{4} I_M$ and
					\begin{align}
						\lim_{i \rightarrow \infty}
						&\E\|\tilde{\bm{w}}_{k,i}\|_{\frac{H_c}{2}}^2
										\nn\\
							&=			\frac{\mu_{\max}}{4} \cdot
										\mathrm{Tr}
										\left\{
											(p^T \otimes I_M)
											\cdot
											\mc{R}_v
											\cdot
											(p \otimes I_M)
										\right\}
										\nn\\
										&\quad
										+
										O\big(\mu_{\max}^{\min(3/2, 1+\kappa/2)}\big)
						\label{Equ:Performance:MSE_R}
					\end{align}
		\end{enumerate}
%	where \eqref{Equ:Performance:MSE_R} is related to the concept of \emph{excess risk} in machine
%	learning\cite{towfic2013jmlrSubmitted}.
	\hfill\QED
\end{example}

\begin{example} \emph{(Distributed stochastic gradient descent: Uncorrelated noise)}
	In the special case that the gradient noises at the different agents are uncorrelated with each other, 
	then $\mR_v$ is block diagonal and we write it as 
		\begin{align}
			\mR_v	=	\diag\{ R_{v,1},\ldots, R_{v,N} \}
						\nn
		\end{align}
	where $R_{v,k}$ is the $M\times M$ covariance matrix of the gradient noise
	at agent $k$. Then, the MSE expression \eqref{Equ:Performance:MSE_I}
	at each agent $k$ can be written as
		\begin{align}
			&\lim_{i \rightarrow \infty}
			\E \|\tilde{\w}_{k,i} \|^2
						\nn\\
					&\quad=
						\frac{\mu_{\max}}{2} \cdot
						\mathrm{Tr}
						\left\{
							\left(
								\sum_{k=1}^N
								p_k \nabla_w^2 J_k(w^o)
							\right)^{-1}
							\cdot
							\left(
								\sum_{k=1}^N
								p_k^2 R_{v,k}
							\right)							
						\right\}
						\nn\\
						&\quad\quad
						+\!
						O\big(\mu_{\max}^{\min(3/2, 1+\kappa/2)}\big)
						\nn
		\end{align}
	and expression \eqref{Equ:Performance:MSE_R} for the weighted MSE
	becomes
		\begin{align}
			\lim_{i \rightarrow \infty}
			\E\|\tilde{\bm{w}}_{k,i}\|_{\frac{H_c}{2}}^2
					&=			\frac{\mu_{\max}}{4} \cdot
								\mathrm{Tr}
								\left\{
									\sum_{k=1}^N
									p_k^2 R_{v,k}
								\right\}
								\nn\\
								&\quad
								+
								O\big(\mu_{\max}^{\min(3/2, 1+\kappa/2)}\big)
								\nn
		\end{align}
	\hfill\QED
\end{example}

\section{Performance of Centralized Stochastic Approximation Solution}
\label{Sec:LearnBehav:GlobalBehav}

We conclude from \eqref{Equ:SteadyState:WMSE_final} that the
weighted mean-square-error at each node $k$ will be the same
across all agents in the network for small step-sizes.
This is an important ``equalization'' effect. 
Moreover, as we now verify, the performance level given by \eqref{Equ:SteadyState:WMSE_final} is close to the performance of a centralized strategy
that collects all the data from the agents and processes them using
the following recursion:
	\begin{align}
		\label{Equ:Performance:CentralizedStrategy}
		\bm{w}_{\mathrm{cent},i}	=	\bm{w}_{\mathrm{cent},i-1}
										-
										\mu_{\max}
										\sum_{k=1}^N 
										p_k 
										\hat{\bm{s}}_{k,i}
										(\bm{w}_{\mathrm{cent},i-1})
	\end{align}
To establish this fact, we first note that the performance of the above centralized strategy can be analyzed
in the same manner as the distributed strategy. Indeed, let
$\check{\bw}_{\mathrm{cent},i} \defeq \bw_{\mathrm{cent},i}-\bar{w}_{c,i}$
denote the discrepancy between the above centralized recursion
and reference recursion \eqref{Equ:LearnBehav:RefRec}. Then,
we obtain from \eqref{Equ:LearnBehav:RefRec} and \eqref{Equ:Performance:CentralizedStrategy} that 
	\begin{align}
		\check{\bw}_{\cent,i}	&=	
									T_c(\bw_{\cent,i-1}) - T_c(\bar{w}_{c,i-1})
									\nn\\
									&\quad
									-
									\mu_{\max}\cdot (p^T \otimes I_M)
									\bv_i(\bw_{\cent,i-1})
		\label{Equ:Centralized:wcentcheck_recursion}
	\end{align}
where the operator $T_c(w)$ is defined as the following mapping from $\mb{R}^M$ to $\mb{R}^M$:
	\begin{align}
		T_c(w)		\defeq
									w - \mu_{\max} \sum_{k=1}^N p_k s_k(w)
									\nn
	\end{align}
Comparing \eqref{Equ:Centralized:wcentcheck_recursion} with 
expression \eqref{P1-Equ:Lemma:ErrorDynamics:JointRec_wc_check} from Part I\cite{chen2013learningPart1} { (repeated below):
	\begin{align}
			\check{\bm{w}}_{c,i}	
							&=	T_c(\bm{w}_{c,i-1}) 
								\!-\!
								T_c(\bar{w}_{c,i-1})
								\!-\!
								\mu_{\max} \!\cdot\! (p^T \!\otimes\! I_M) 
								\left[
										\bm{z}_{i-1}
										\!+\!
										\bm{v}_i
								\right]
			\label{Equ:Lemma:ErrorDynamics:JointRec_wc_check_repeat}
	\end{align}
}%
we note that these two
recursions take similar forms except for an additional perturbation
term $\bz_{i-1}$ in \eqref{Equ:Lemma:ErrorDynamics:JointRec_wc_check_repeat}. Therefore, following the same line of transient analysis as in Part I\cite{chen2013learningPart1} and steady-state analysis as in the proof of Theorem \ref{Thm:SteadyStatePerformance} stated earlier, we can conclude that, in the small step-size regime, the transient behavior of the centralized strategy \eqref{Equ:Performance:CentralizedStrategy} is close to the reference recursion \eqref{Equ:LearnBehav:RefRec}, and the steady-state performance is again given by \eqref{Equ:SteadyState:WMSE_final}.
	\begin{theorem}[Centralized performance]
		\label{Thm:CentralizedPerformance}
		Suppose Assumptions \ref{Assumption:UpdateVectorRandomness}--\ref{Assumption:GradientNoise4thOrderMoment}  
		hold and suppose the step-size parameter $\mu_{\max}$ in the centralized 
		recursion \eqref{Equ:Performance:CentralizedStrategy} satisfies
		the following condition
			\begin{align}
				0	<	\mu_{\max}	<	\frac{
												\lambda_L
											}
											{
												\|p\|_1^2
												\cdot
												\left(
													\frac{\lambda_U^2}{2}
													+ 
													2\alpha
												\right)
											}
				\label{Equ:Thm:Centralized:stepsize}
			\end{align}
		Then, the MSE term $\E \|\tilde{\w}_{\cent,i}\|^2$  converges at the rate of
			\begin{align}
				r	
						=		
							\big[
								\rho(I_M - \mu_{\max} H_c)
							\big]^2
							+ 
							O\big( 
								(\mu_{\max} \epsilon )^{\frac{1}{2(M-1)}} 
							\big)	
				\label{Equ:Thm:Centralized:r}
			\end{align}
		where $\epsilon$ is an arbitrarily small positive number. Furthermore,
		in the small step-size regime, the steady-state MSE of 
		\eqref{Equ:Performance:CentralizedStrategy} satisfies
			\begin{align}
				\limsup_{i\rightarrow\infty} 
				\E\|\tilde{\bw}_{\cent,i}\|_{\Sigma}^2
					&\le	\mu_{\max}\!\cdot\!
						\Tr\left\{
							X 
							(p^T\! \otimes\! I_M) 
							\cdot 
							\mc{R}_v 
							\cdot 
							(p\! \otimes\! I_M)
						\right\}
						\nn\\
						&\quad
						+\! 
						O\big(\mu_{\max}^{\min(3/2, 1+\kappa/2)}\big)
				\label{Equ:Thm:Centralized:limsup_WMSE}
						\\
				\liminf_{i\rightarrow\infty} 
				\E\|\tilde{\bw}_{\cent,i}\|_{\Sigma}^2
					&\ge	\mu_{\max}\!\cdot\!
						\Tr\left\{
							X 
							(p^T\! \otimes\! I_M) 
							\cdot 
							\mc{R}_v 
							\cdot 
							(p\! \otimes\! I_M)
						\right\}
						\nn\\
						&\quad
						-\! 
						O\big(\mu_{\max}^{\min(3/2, 1+\kappa/2)}\big)
				\label{Equ:Thm:Centralized:liminf_WMSE}
			\end{align}
		\hfill\QED
	\end{theorem}	
\indent\emph{Remark:} Similar to our explanation following  \eqref{Equ:SteadyState:WMSE_limsup_finalfinal}--\eqref{Equ:SteadyState:WMSE_liminf_finalfinal}, expressions \eqref{Equ:Thm:Centralized:limsup_WMSE}--\eqref{Equ:Thm:Centralized:liminf_WMSE} also mean that, for small step-sizes, the steady-state MSE of the centralized strategy will be tightly sandwiched between two almost identical bounds. Therefore, we will again use the traditional limit notation for the centralized steady-state MSE for simplicity, and will write instead:
	\begin{align}
		\lim_{i\rightarrow\infty} 
		\E\|\tilde{\bw}_{\cent,i}\|_{\Sigma}^2
			&=	\mu_{\max}\!\cdot\!
				\Tr\left\{
					X 
					(p^T\! \otimes\! I_M) 
					\cdot 
					\mc{R}_v 
					\cdot 
					(p\! \otimes\! I_M)
				\right\}
				\nn\\
				&\quad
				+\! 
				O\big(\mu_{\max}^{\min(3/2, 1+\kappa/2)}\big)
		\label{Equ:Thm:Centralized:WMSE_final}
	\end{align}
which is the same as \eqref{Equ:SteadyState:WMSE_final} up to the first-order of $\mu_{\max}$.

%Furthermore, we note from \eqref{Equ:DistProc:r_DistProc} and
%\eqref{Equ:SteadyState:WMSE_final} that the performance of
%Transient Phase II and Steady-state Phase is mainly determined by
%the right-eigenvector $\theta$ of the matrix $A$ at eigenvalue one. 
%As we will show in Sec. \ref{Sec:Benefits}, we
%can always construct a matrix $A$ in a decentralized manner
%with arbitrary $\theta$ of positive entries as
%long as Assumption \ref{Assumption:Network} holds. 
%Therefore, an important conclusion is that in Transient Phase II and
%Steady-state, agents in the connected network exhibit \emph{global} learning behavior
%via local interactions and this behavior is insensitive to the particular topology of the network. 

\section{Benefits of Cooperation}
\label{Sec:Benefits}

In this section, we illustrate the implications of the main results of this work in the context  of distributed learning and  distributed optimization. Consider a network of $N$ connected  agents, where each agent $k$ receives a stream of data $\{\bm{x}_{k,i}\}$ arising from some underlying distribution. The networked multi-agent system would like to extract from the distributed data some useful information about the underlying process. To measure the quality of the inference task, an individual cost function $J_k(w)$ is associated with each agent $k$, where $w$ denotes an $M \times 1$ parameter vector. The agents are generally interested in minimizing some aggregate cost function of the form \eqref{Equ:FamilyDistStrategy:J_glob}:
	\begin{align}
		J^{\mathrm{glob}}(w)		=	\sum_{k=1}^N J_k(w)
		\label{Equ:Benefits:J_glob_def}
	\end{align}
Based  on whether the individual costs $\{J_k(w)\}$ share a common minimizer or not, we can classify problems of the form \eqref{Equ:Benefits:J_glob_def} into two broad categories.

%----------------------------------------------------------------------
\subsection{Category I: Distributed Learning}
\label{Sec:Benefits:DistLearn}

In this case, the data streams $\{\bm{x}_{k,i}\}$ are assumed to be
generated by (possibly different) distributions that nevertheless depend on the same parameter
vector $w^o \in \mathbb{R}^M$.
The objective is then to estimate this common parameter $w^o$ in a distributed manner. 
To do so, we first need to associate with each agent $k$
a cost function $J_k(w)$ that measures how well some arbitrary parameter $w$
approximates $w^o$.
The cost $J_k(w)$ should be such that $w^o$ is one of
its minimizers. More formally,
let $\mc{W}_k^o$ denote the set of vectors that minimize
the selected $J_k(w)$, then it is expected that
	\begin{align}
		\label{Equ:ProblemFormulation:W_k_o_mc}
		w^o	\in \mc{W}_k^o	\defeq	\left\{w: \arg\min_{w} J_k(w)\right\}
	\end{align}
for $k=1,\ldots,N$.
Since $J^{\rm glob}(w)$ is assumed to be strongly convex, 
then the intersection of the sets ${\cal W}_k^o$ should contain the 
single element $w^o$:
	\begin{align}
		\label{Equ:ProblemFormulation:W_o_mc}
		w^o		\in 	\mc{W}^o		=	\bigcap_{k=1}^N \mc{W}_k^o
	\end{align}
The main motivation
for cooperation in this case is that the data collected at 
each agent $k$ may not be sufficient to uniquely identify
$w^o$ since $w^o$ is not necessarily the unique element in $\mc{W}_k^o$; 
this happens, for example, when the individual costs $J_k(w)$ 
are not \emph{strictly} convex. 
However, once the individual costs are aggregated into \eqref{Equ:Benefits:J_glob_def}
and the aggregate function is strongly convex, then 
$w^o$ is the unique element in $\mc{W}^o$.  In this way, the 
cooperative minimization of $J^{\mathrm{glob}}(w)$ allows the agents
to estimate $w^o$.

\subsubsection{Working under Partial Observation}
Under the scenario described by \eqref{Equ:ProblemFormulation:W_o_mc}, 
the solution of \eqref{Equ:LearnBehav:FixedPointEqu} agrees with the 
unique minimizer $w^o$ for $J^{\rm glob}(w)$ given by \eqref{Equ:Benefits:J_glob_def} 
regardless of the $\{p_k\}$ and, therefore, regardless of the combination policy $A$. 
Therefore, the results from Sec. \ref{Sec:FamilyDistStrategy:ReviewPartI} (see also Sec. \ref{P1-Sec:LearningBehavior} of Part I \cite{chen2013learningPart1}) show that the iterate $\bm{w}_{k,i}$ at
each agent $k$ converges to this unique $w^o$ at a centralized rate and the results from Sec. \ref{Sec:LearnBehav:SteadyStateAnal} of this Part II show that this iterate achieves the centralized steady-state MSE performance. Note that Assumption \ref{Assumption:UpdateVectorMonot} can be
satisfied without requiring each $J_k(w)$ to be strongly convex. Instead, we only
require $J^{\mathrm{glob}}(w)$ to be strongly convex. In other words, we
do not need each agent to have complete information about $w^o$; we only need 
the network to have enough information to determine $w^o$ uniquely. 
Although the individual agents in this case have partial information 
about $w^o$, the distributed strategies \eqref{Equ:ProblemFormulation:DisAlg_Comb1}--\eqref{Equ:ProblemFormulation:DisAlg_Comb2} enable them to attain the same performance 
level as a centralized solution. The following example illustrates
the idea in the context of distributed least-mean-squares estimation over networks.
	\begin{example}
		Consider Example \ref{Ex:Rv_LMS} again.
		When the covariance matrix $R_{u,k} \defeq \E[\bm{u}_{k,i}^T\bm{u}_{k,i}]$
		is rank deficient, then
		$J_k(w)$ in \eqref{Equ:Example:J_k_LMS} would not be strongly convex and there would be 
		infinitely many minimizers to $J_k(w)$. In this case, the information
		provided to agent $k$ via \eqref{Equ:Example:LinearModel} is not
		sufficient to determine $w^o$ uniquely. However, if the global
		cost function is strongly convex, which can be verified to be
		equivalent to requiring:
			\begin{align}
				\label{Equ:Example:NetworkObservation}
				\sum_{k=1}^N p_k R_{u,k} >  \lambda_L I_M > 0
			\end{align}
		then the information collected over the entire network is rich
		enough to learn the unique $w^o$. As long as
		\eqref{Equ:Example:NetworkObservation} holds for one set of 
		positive $\{p_k\}$, it will hold for all other $\{p_k\}$. 
		A ``network	observability'' condition similar to
		\eqref{Equ:Example:NetworkObservation}
		was used in \cite{kar2011converegence} to
		characterize the sufficiency of information over the network 
		in the context of distributed estimation over linear models
		albeit with diminishing step-sizes.
		\hfill\QED
	\end{example}

\subsubsection{Optimizing the MSE Performance}
Since the distributed strategies \eqref{Equ:ProblemFormulation:DisAlg_Comb1}--\eqref{Equ:ProblemFormulation:DisAlg_Comb2} converge to the same unique minimizer $w^o$ of \eqref{Equ:Benefits:J_glob_def} for any set of $\{p_k\}$, we can then consider selecting the $\{p_k\}$ to optimize the MSE performance. Consider the case where $H_k \equiv H$ and $\mu_k \equiv \mu$ and assume the gradient noises $\bv_{k,i}(w)$ are asymptotically uncorrelated across the agents so that $\mc{R}_v$ from \eqref{Equ:Assumption:Rv:R_v_limit} is block diagonal with entries denoted by: 
	\begin{align}
		\label{Equ:Benefits:R_v_blockdiagonal}
		\mc{R}_v		=	\diag\{R_{v,1},\ldots,R_{v,N}\}
	\end{align}
Then, we have $\beta_k=1$, $p_k = \theta_k$ and 
	\begin{align}
		H_c=H = \nabla_{w}^2 J_1(w^o) = \cdots = \nabla_{w}^2 J_N(w^o)
		\label{Equ:Benefits:H_c_IdenticalHessian}
	\end{align}
in which case 
%expressions \eqref{Equ:Performane:r_StochasticGradient} and 
expression \eqref{Equ:Performance:MSE_I} becomes
	\begin{align}
%		\label{Equ:Benefits:rate}
%		r		&=			
%							1 - 2 \mu_{\max}\lambda_{\min}(H) 
%							+ 
%							O(\mu_{\max}^2) 
%							+ 
%							O\big( 
%								(\mu_{\max} \epsilon )^{\frac{1}{2(M-1)}} 
%							\big)		
%							\\
		\label{Equ:Benefits:MSE}
		\lim_{i\rightarrow\infty}
		\E\|\tilde{\bm{w}}_{k,i}\|^2
				&=			
							\frac{\mu_{\max}}{2}
							\cdot
							\sum_{k=1}^N
							\theta_k^2
							\Tr\left(
								H^{-1} R_{v,k} 
							\right)
							\nn\\
							&\quad
							+
							O\big(\mu_{\max}^{\min(3/2, 1+\kappa/2)}\big)
	\end{align}
The optimal positive coefficients $\{\theta_k\}$ that minimize \eqref{Equ:Benefits:MSE} subject to $\sum_{k=1}^N \theta_k = 1$ are given by
	\begin{align}
		\theta_{k}^o	&=	\frac{\displaystyle[{\Tr(H^{-1} R_{v,k})}]^{-1}}
							{\displaystyle\sum_{\ell=1}^N
							[{\Tr(H^{-1} R_{v,\ell})}]^{-1}},
						\quad
						k=1,\ldots, N
		\label{Equ:Benefits:theta_k_o}
	\end{align}
and, substituting into \eqref{Equ:Benefits:MSE}, the optimal MSE is then given by 
	\begin{align}
		\mathrm{MSE}^{\mathrm{opt}}
					&=			\frac{\mu_{\max} }{2}
								\! \cdot \!
								\left[
									\sum_{\ell=1}^N
									\frac{1}{\Tr( H^{-1} R_{v,\ell})}
								\right]^{-1}
%								\nn\\
%								&\quad
								\!\!\!\!\!\!\!\!
								+
								O\big(\mu_{\max}^{\min(3/2, 1+\kappa/2)}\big)
								\nn
	\end{align}
The optimal Perron-eigenvector $\theta^o = \col\{\theta_1^o,\ldots,\theta_N^o\}$ 
can be implemented
by selecting the combination policy $A$ as the following Hasting's rule
\cite{boyd2004fastest,hastings1970monte,zhao2012performance}:
	\begin{align}
		a_{\ell k}^o		&=	\left\{
								\begin{array}{cc}
									\displaystyle\!\!\!\!
									\frac{(\theta_k^o)^{-1}}
									{
										\max
										\left\{ 
												|\mc{N}_k| \!\cdot\! (\theta_k^o)^{-1}, 
												|\mc{N}_\ell| \!\cdot\! (\theta_\ell^o)^{-1}
										\right\}
									},				
											&	\!\!\! \ell \!\in\! \mc{N}_k \backslash \{k\}	\\
									&\\
									\displaystyle
									1-\sum_{m \in \mc{N}_k \backslash \{k\}} a_{mk}^o,		
											&	\!\!\! \ell = k
								\end{array}
							\right.
							\nn\\
						&\overset{(a)}{=}
							\left\{
								\begin{array}{cc}
									\displaystyle\!\!\!\!
									\frac{{\Tr(H^{-1} R_{v,k})}}
									{
										\max
										\big\{ \!
												|\mc{N}_k| \!\cdot\! {\Tr(H^{-1} \! R_{v,k})}, \;
												|\mc{N}_\ell| \!\cdot\! {\Tr(H^{-1} \! R_{v,\ell})}
										\!\big\}
									},			&\\
												&\\
									\qquad\qquad\qquad\qquad \qquad\qquad
									\ell \!\in\! \mc{N}_k \backslash \{k\}&\\
%												&\\
									\displaystyle
									1-\sum_{m \in \mc{N}_k \backslash \{k\}} a_{mk}^o,		
												&\\
												\qquad\qquad\qquad \qquad \qquad\qquad
												\ell = k \\
								\end{array}
							\right.
		\label{Equ:Benefits:HastingRule}
	\end{align}
where $|\mc{N}_k|$ denotes the cardinality of the set $\mc{N}_k$, and step (a) substitutes \eqref{Equ:Benefits:theta_k_o}.
From \eqref{Equ:Benefits:HastingRule}, we note that the above combination matrix can be constructed
in a decentralized manner, where each node only requires information
from its own neighbors. In practice, the noise covariance matrices $\{R_{v,\ell}\}$ need to be estimated from the local data and an adaptive estimation scheme is proposed in \cite{zhao2012performance} to address this issue.

\subsubsection{Matching Performance across Topologies}
	
Note that the steady-state mean-square error depends on the vector $p$, which is determined by the Perron eigenvector  $\theta$ of the matrix $A$. The above result implies that, as long as the network is strongly connected, i.e., Assumption \ref{Assumption:Network} holds, a left-stochastic matrix $A$ can always be constructed to have any desired Perron eigenvector $\theta$ with positive entries according to
\eqref{Equ:Benefits:HastingRule}. Now, starting from any collection of $N$ agents, there exists a 
finite number  of topologies that can link these agents together. For each possible topology, there are
infinitely many combination policies that can be used to train the network. One important conclusion that follows from the above results is that regardless of the topology, there always exists a choice for $A$ such that the performance of all topologies are identical to each other to first-order in $\mu_{\max}$.
In other words, no matter how the agents are connected to each other, there is always a way to select the combination weights such that the performance of the network is invariant to the topology. This also means that,
for any connected topology, there is always a way to select the combination weights such that the performance of the network matches that of the centralized solution.

\begin{figure*}
		\centering
		\centerline{
			\subfigure[]{
				\includegraphics[width=0.45\textwidth]
				{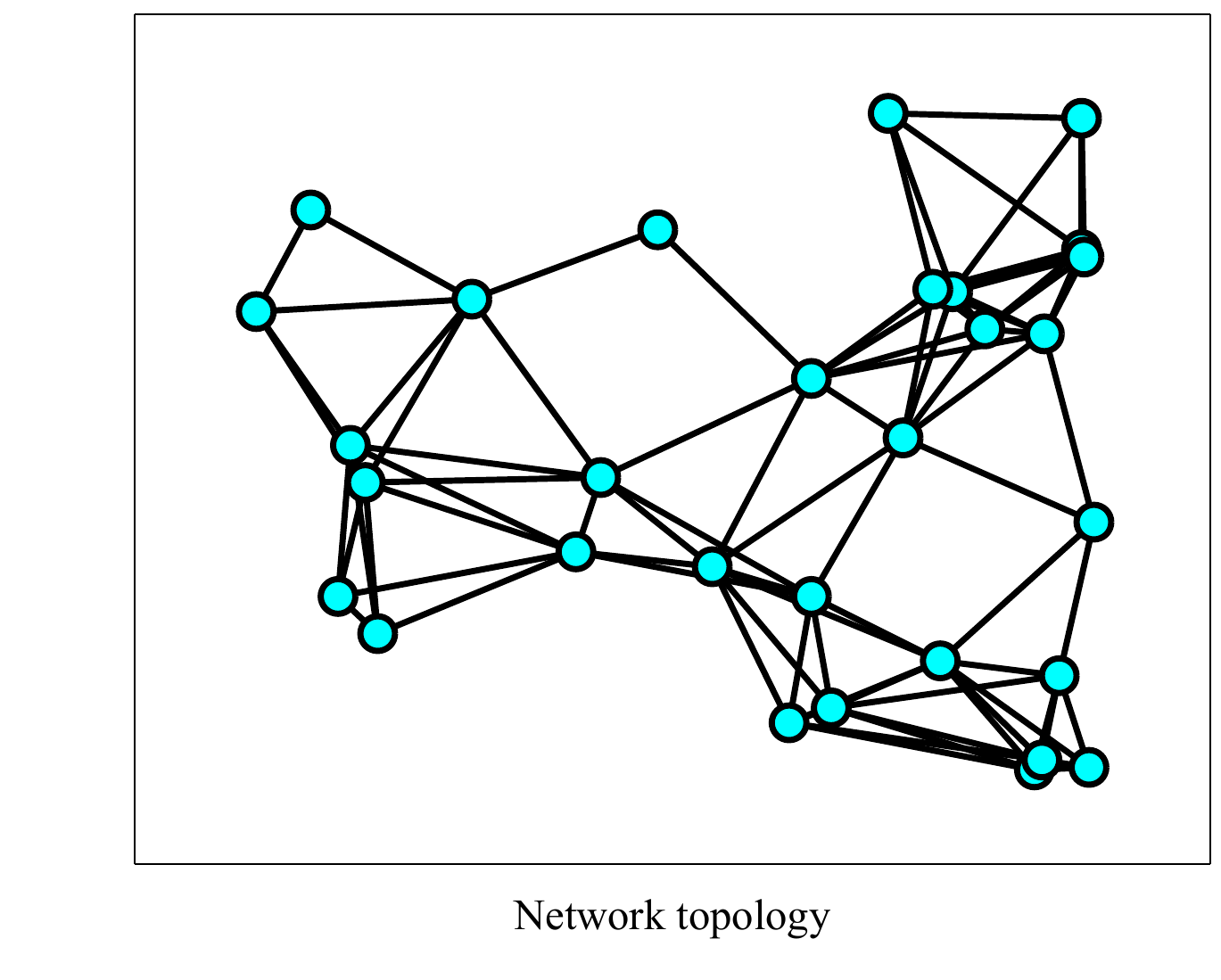}
			}	
			\subfigure[]{
				\includegraphics[width=0.45\textwidth]
				{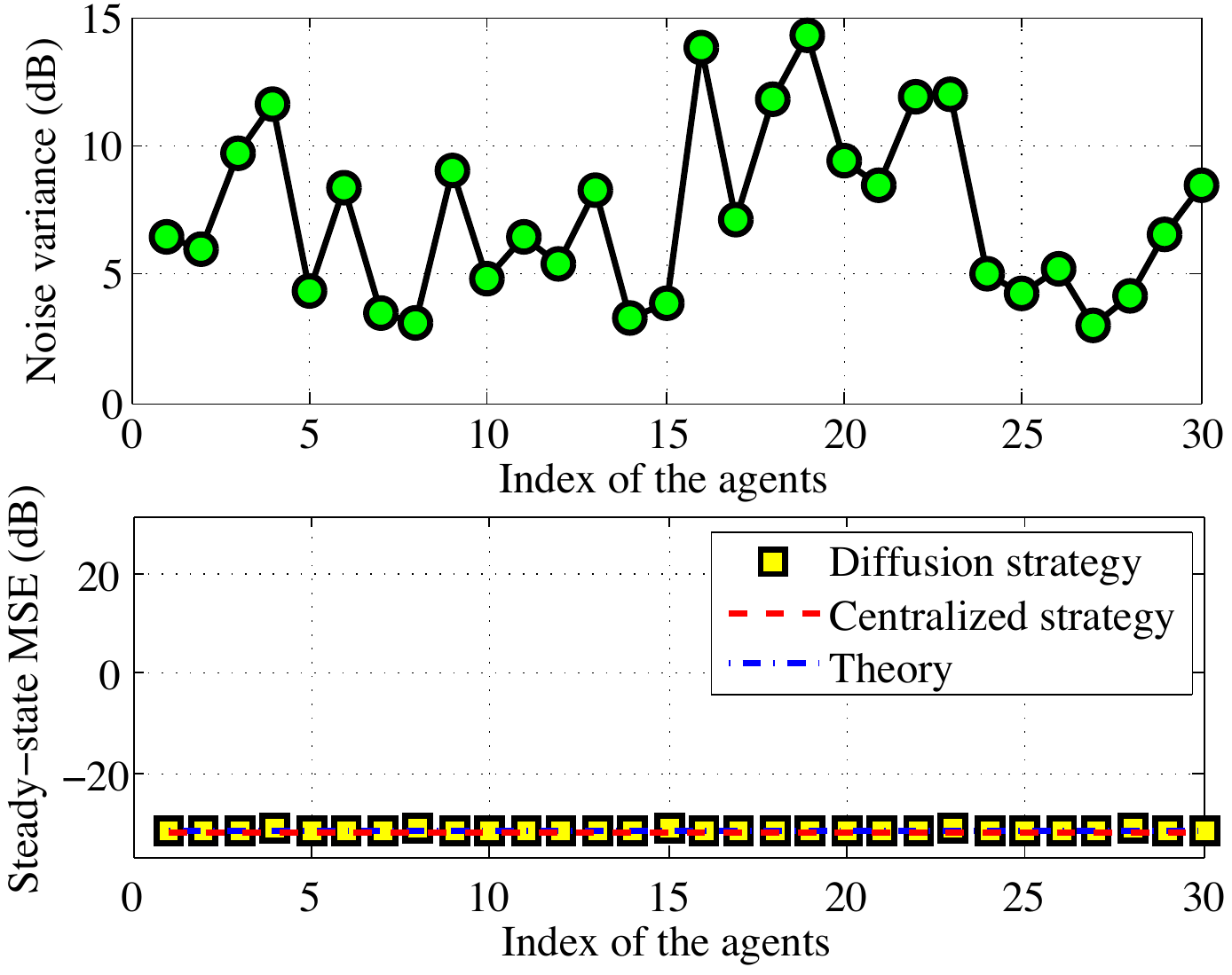}
			}
		}
		\centerline{
			\subfigure[]{
				\includegraphics[width=0.45\textwidth]
				{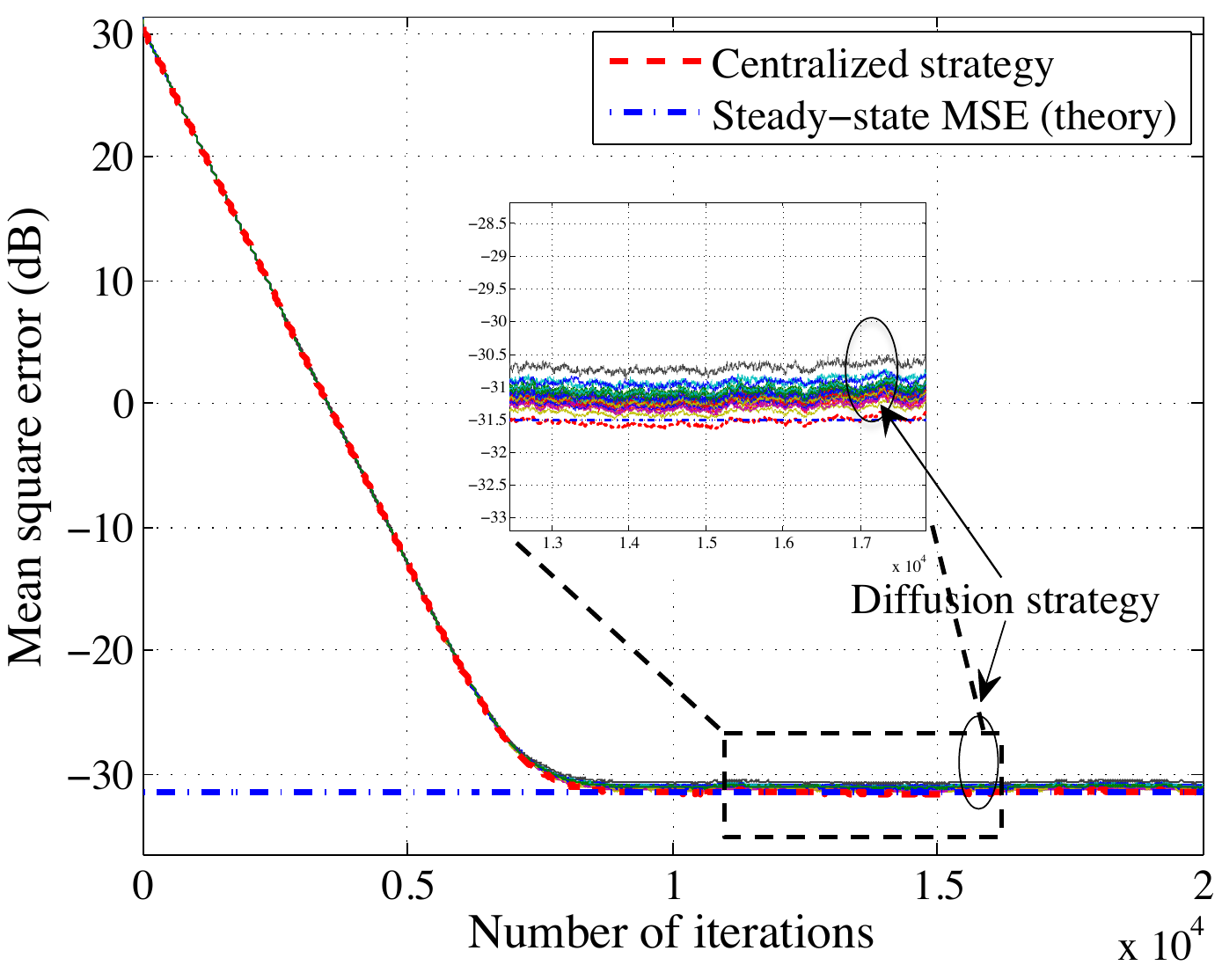}
			}	
			\subfigure[]{
				\includegraphics[width=0.45\textwidth]
				{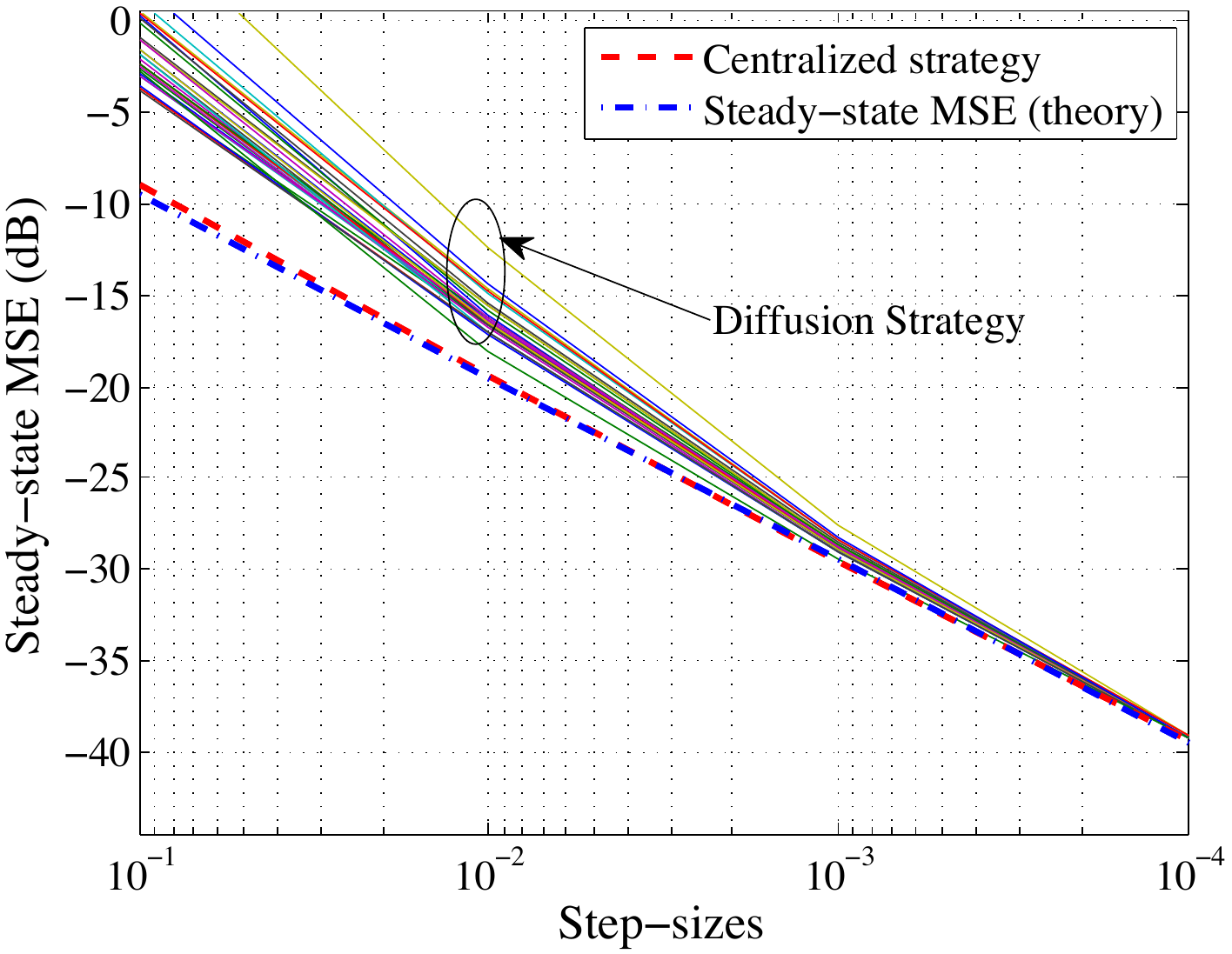}
			}
		}
		\caption{Comparing the performance of a $30$-node diffusion LMS network with that of the
		centralized strategy \eqref{Equ:Performance:CentralizedStrategy}, 
		where $M=10$, $\mu_k = 0.0005$ for all agents, and Hasting's rule 
		\eqref{Equ:Benefits:HastingRule} is used as the combination policy.
		The result is obtained by averaging over $1000$ Monte Carlo experiments.
		(a) A randomly generated topology. (b) The noise profile across the 
		network, and the steady-state MSE of diffusion LMS, centralized strategy, 
		and the theoretical value (the first-order term in \eqref{Equ:SteadyState:WMSE_final}). 
		(c) The learning curves for different agents in the diffusion
		LMS network, the centralized strategy, and the theoretical steady-state
		MSE (the first-order term in \eqref{Equ:SteadyState:WMSE_final}). 
		(d) The centralized strategy, the steady-state MSE of
		diffusion strategy at all agents,  and 
		the theoretical value (the first-order term in \eqref{Equ:SteadyState:WMSE_final})
		for different values of step-sizes.
		}
		\label{Fig:SimulatedPerformance_GeoRandGraph}
	\end{figure*}

\begin{example}

We illustrate the result using the diffusion least-mean-square estimation context discussed earlier in Example \ref{Ex:Rv_LMS}. Consider a network of $30$ agents ($N=30$), where each agent has access to a stream of data samples $\{\bu_{k,i},\bd_{k}(i)\}$ that are generated by the linear model \eqref{Equ:Example:LinearModel}. As assumed in Example \ref{Ex:Rv_LMS}, the $1 \times M$ regressors $\{\u_{k,i}\}$ are zero mean and independent over time and space with covariance matrix $R_{u,k}$, and the noise sequence $\{\bm{n}_l(j)\}$ is also zero mean, white, with variance $\sigma_{n,l}^2$, and independent of the regressors $\{\bm{u}_{k,i}\}$ for all $l,k,i,j$. In the simulation here, we consider the case where $M=10$, $R_{u,k} =  I_M$. In diffusion LMS estimation, each agent $k$ uses \eqref{Equ:Example:J_k_LMS} as its cost function $J_k(w)$ and \eqref{Equ:Example:ski_LMS} as the stochastic gradient vector $\hat{\s}_{k,i}(\cdot)$. Therefore, each agent $k$ adopts the following recursion to adaptively estimate the model parameter $w^o$, which is the minimizer of the global cost function \eqref{Equ:FamilyDistStrategy:J_glob}:
	\begin{align}
		\bm{\psi}_{k,i}		&=		
								\w_{k,i} + 2 \mu_k \bm{u}_{k,i}^T[\bm{d}_{k}(i)-\bm{u}_{k,i} \w_{k,i-1}]
								\nn\\
		\w_{k,i}			&=
								\sum_{ l \in \mN_{k} }
								a_{lk}
								\bm{\psi}_{l, i}
								\nn
	\end{align}
We randomly generate a topology as shown in Fig. \ref{Fig:SimulatedPerformance_GeoRandGraph} (a) and noise variance profile across agents as shown in Fig. \ref{Fig:SimulatedPerformance_GeoRandGraph} (b). We choose $\mu_k \equiv \mu =  0.0005$ to be the step-size for all agents and Hasting's rule \eqref{Equ:Benefits:HastingRule} as the combination policy. In the simulation, we assume the noise variances are known to the agents. Alternatively, they can also be estimated in an adaptive manner using approaches proposed in \cite{zhao2012performance}. The results are obtained by averaging over $1000$ Monte Carlo experiments. In Fig \ref{Fig:SimulatedPerformance_GeoRandGraph}(b), we also show the steady-state MSE of all agents, respectively, and compare them to the theoretical value (the first-order term in \eqref{Equ:SteadyState:WMSE_final}) and to the following centralized LMS strategy:
	\begin{align}
		\w_{\cent,i}		&=
								\w_{\cent, i-1}
								+
								2 \mu \sum_{k=1}^N 
								p_k
								\cdot
								\bm{u}_{k,i}^T[\bm{d}_{k}(i)-\bm{u}_{k,i} \w_{\cent,i-1}]
								\nn
	\end{align}
where $p_k = \theta_k^o$ is given by \eqref{Equ:Benefits:theta_k_o}. Fig. \ref{Fig:SimulatedPerformance_GeoRandGraph}(b) illustrates the equalization effect over the network; each agent in the network achieves almost the same steady-state MSE that is close to the centralized strategy although the noise variances in the data are different across the agents. Furthermore, In Fig. \ref{Fig:SimulatedPerformance_GeoRandGraph} (c), we illustrate the learning curves of all agents, and compare them to the theoretical value and the centralized LMS strategy. We observe from Fig. \ref{Fig:SimulatedPerformance_GeoRandGraph} (c) that the learning curves of all agents are close to each other and to the centralized strategy. 
{
Finally, we show in Fig. \ref{Fig:SimulatedPerformance_GeoRandGraph} the steady-state MSE of the diffusion strategy at all agents for different values of the step-sizes, and compare them to the MSE of the centralized strategy and against the theoretical values (the first-order term in \eqref{Equ:SteadyState:WMSE_final}). It is seen in the figure that, for small step-sizes, the steady-state MSE values at different agents approach that of the centralized strategy and the first-order term in \eqref{Equ:SteadyState:WMSE_final} since the higher-order term in \eqref{Equ:SteadyState:WMSE_final} decays faster than the first-order term.
}
\hfill \QED
\end{example}

%\subsubsection{Left-Stochastic $A$ Outperforms Doubly-Stochastic $A$}
%If we choose a doubly-stochastic combination policy for $A$, then 
%$\theta_k=1/N$. Note from \eqref{Equ:Benefits:rate} that the convergence
%rate would be the same for both the left-stochastic policy \eqref{Equ:Benefits:HastingRule}
%and the doubly-stochastic policy. However, the steady-state MSE is different
%and it holds that
%	\begin{align}
%		\label{Equ:Benefits:LeftStochastic_vs_DoublyStochastic}
%		\mathrm{MSE}^{\mathrm{ds}}
%						&=		\frac{\mu}{2}
%								\cdot
%								\frac{1}{N^2}
%								\sum_{k=1}^N
%								\Tr\left(
%									R_{v,k} H^{-1}
%								\right)
%						\ge 	\mathrm{MSE}^{\mathrm{opt}}
%	\end{align}
%where equality holds only when all agents have the same noise covariance
%matrices, i.e., $R_{v,\ell} \equiv R_{v}$.

\subsection{Category II: Distributed Optimization}
\label{Sec:Benefits:ParetoOpt}

In this case, we include situations where the individual costs $J_k(w)$ do not have a common minimizer, i.e., $\mc{W}^o = \emptyset$. The optimization problem should then be viewed as one of solving a multi-objective minimization problem
	\begin{align}
				\label{Equ:ProblemFormulation:Multiobjective}
				\min_w\left\{ J_1(w),\ldots,J_N(w) \right\}
			\end{align}
where $J_k(w)$ is an individual convex cost associated with each agent $k$. A vector $w^o$ is said to be a Pareto optimal solution to \eqref{Equ:ProblemFormulation:Multiobjective} if there does not exist another vector $w$ that is able to improve (i.e., reduce) any individual cost without degrading (increasing) some of the other costs. Pareto optimal solutions are not unique. The question we would like to address  now is the following. Given individual costs $\{J_k(w)\}$ and a combination policy $A$, what is the limit point of the distributed strategies \eqref{Equ:ProblemFormulation:DisAlg_Comb1}--\eqref{Equ:ProblemFormulation:DisAlg_Comb2}? From Sec. \ref{Sec:FamilyDistStrategy:ReviewPartI} (see also Theorem \ref{P1-Thm:NonAsymptoticBound} in Part I \cite{chen2013learningPart1}), the distributed strategy \eqref{Equ:ProblemFormulation:DisAlg_Comb1}--\eqref{Equ:ProblemFormulation:DisAlg_Comb2} converges to the limit point $w^o$ defined as the unique solution to  \eqref{Equ:LearnBehav:FixedPointEqu}. Substituting $s_k(w)=\nabla_w J_k(w)$ into \eqref{Equ:LearnBehav:FixedPointEqu}, we obtain
	\begin{align}
		\sum_{k=1}^N p_k \nabla_w J_k(w^o)	=	0
		\nn
	\end{align}
In other words, $w^o$ is the minimizer of the following global cost function:
	\begin{align}
		\label{Equ:ApplicationOptimization:J_glob_Pareto}
		J^{\mathrm{glob}}(w)	=	\sum_{k=1}^N p_k J_k(w)
	\end{align}
It is shown in \cite[pp.178--180]{boyd2004convex} that the minimizer of \eqref{Equ:ApplicationOptimization:J_glob_Pareto} is a Pareto-optimal solution for the multi-objective optimization problem \eqref{Equ:ProblemFormulation:Multiobjective}. And different choices for the vector $p$ lead to different 
Pareto-optimal points on the tradeoff curve.

{

Now, a useful question to consider is the reverse direction. Suppose, we are given a set of $\{p_k\}$ (instead of $A$) and we want the distributed strategy \eqref{Equ:ProblemFormulation:DisAlg_Comb1}--\eqref{Equ:ProblemFormulation:DisAlg_Comb2} to converge to the limit point $w^o$ that is the solution of:
	\begin{align}
		\sum_{k=1}^N p_k \nabla_w J_k(w^o)	=	0
		\label{Equ:Benefits:FixedPoint2}
	\end{align}
Note that once the topology of the network is given, the positions of the nonzero entries in the matrix $A$ are known and we are free to select the values of these nonzero entries. One possibility is to choose the same step-size for all agents (i.e., $\mu_k \equiv \mu$), and to select the nonzero entries of $A$ such that its Perron vector $\theta$ equals this desired $p$. This construction can be achieved by using the following Hasting's rule\cite{hastings1970monte,zhao2012performance}:
	\begin{align}
		a_{\ell k}	=	\begin{cases}
						\displaystyle
						\frac{p_k^{-1}}
						{
							\max
							\left\{ 
									|\mc{N}_k|\cdot p_k^{-1}, 
									|\mc{N}_l| \cdot p_l^{-1}
							\right\}
						},				
								&	l \in \mc{N}_k \backslash \{k\}	\\
						\displaystyle
						1-\sum_{m \in \mc{N}_k \backslash \{k\}} a_{mk},		
								&	l = k
					\end{cases}
		\label{Equ:Benefits:HastingRule2}
	\end{align}
That is, as long as we substitute the desired set of $\{p_k\}$ into \eqref{Equ:Benefits:HastingRule2} and use the obtained $\{a_{\ell k}\}$ together with $\mu_k \equiv \mu$, the distributed strategy will converge to the $w^o$ in \eqref{Equ:Benefits:FixedPoint2} with the desired $\{p_k\}$.

}

%
%
%
%%%%%%%%%%%%%%%%%%%%%%%%%%%%%%%%%%%%%%%%%%%%%%%%%%%%%%%%%%
%\section{Applications}
%\label{Sec:Benefits}
%
%\subsection{Distributed Learning}
%
%\subsubsection{Working under partial observation}
%\label{Sec:Applications:DL:PartialObservation}
%
%
%
%\subsubsection{Optimizing the MSE performance}
%
%
%
%\subsubsection{Left-stochastic $A$ outperforms doubly-Stochastic $A$}
%
%
%
%
%
%
%
%
%%-------------------------------------------------
%\subsection{Distributed Pareto Optimization}

%%%%%%%%%%%%%%%%%%%%%%%%%%%%%%%%%%%%%%%%%%%%%%%%%%%%%%%%%
\section{Conclusion}
Along with Part I \cite{chen2013learningPart1}, this work examined in some detail the mean-square performance, convergence, and stability of distributed strategies for adaptation and learning over graphs under {\em constant} step-size update rules. Keeping the step-size fixed allows the network to track drifts in the underlying data models,  their statistical distributions, and even drifts in the utility functions. Earlier work \cite{zhao2012diffusion} has shown that constant adaptation regimes endow networks with tracking abilities and derived results that quantify how the performance of adaptive networks is affected by the level of non-stationarity in the data. Similar conclusions extend to the general scenario studied in Parts I and II of the current work, which is the reason why step-sizes have been set to a constant value throughout our treatment. When this is done, the dynamics of the learning process is enriched in a nontrivial manner. This is because the effect of gradient noise does not die out anymore with time (in comparison, when diminishing step-sizes are used, gradient noise is annihilated by the decaying step-sizes). And since agents are coupled through their interactions over the network, it follows that their gradient noises will continually influence the performance of their neighbors. As a result, the network mean-square performance does not tend to zero anymore. Instead, it approaches a steady-state level. One of the main objectives of this Part II has been to quantify this level and to show explicitly how its value is affected by three parameters: the network topology, the gradient noise, and the data characteristics.  As the analysis and the detailed derivations in the appendices of the current manuscript show, this is a formidable task to pursue due to the coupling among the agents. Nevertheless, under certain conditions that are generally weaker than similar conditions used in related contexts in the literature, we were able to derive accurate expressions for the network MSE performance and its convergence rate. For example, the MSE expression we derived is accurate in the first order term of $\mu_{\max}$.  Once an MSE expression has been derived, we were then able to optimize it over the network topology (for the important case of uniform Hessian matrices across the network, as is common for example in machine learning\cite{theodoridis2008PatternRbook} and mean-square-error estimation problems\cite{Sayed08}). We were able to show that arbitrary connected topologies for the same set of agents can always be made to perform similarly. We were also able to show that arbitrary connected topologies for the same set of agents can be made to match the performance of a fully connected network. These are useful insights and they follow from the analytical results derived in this work.

%%%%%%%%%%%%%%%%%%%%%%%%%%%%%%%%%%%%%%%%%%%%%%%%%%%%%%%%%
\appendices

%----------------
\section{Proof of Theorem \ref{Thm:SteadyStatePerformance}}
\label{Appendix:Proof_Thm_SteadyStatePerf}

The argument involves several steps, labeled steps A through E, and relies also on intermediate results that are proven in this appendix. We start with step A.

\subsection{Relating the weighted MSE to the steady-state error covariance matrix $\Pi_{i}$}
\label{Appendix:Proof_Thm_SteadyStatePerf:RelatingMSEtoErrorCov}
Let $\Pi_i \defeq \E\left\{ \tilde{\bw}_{i}\tilde{\bw}_{i}^T \right\}$ denote the error covariance matrix of the global error vector
	\begin{align}
		\tilde{\bw}_i	\defeq	\col\{ \tilde{\w}_{1,i}, \ldots, \tilde{\w}_{N,i}\}
		\nn
	\end{align}
where $\tilde{\w}_{k,i} \defeq w^o - \tilde{\w}_{k,i}$. Note that if we are able to evaluate $\Pi_{i}$ as $i\rightarrow\infty$, then we can obtain the steady-state weighted mean-square-error for any individual agent via the following relation:\footnote{
More formally, the limit of \eqref{Equ:SteadyState:WMSE_Tr} may not exist. However, as we proceed to show, the $\limsup$ and the $\liminf$ of $\E\|\tilde{\bw}_{k,i}\|_{\Sigma}^2$ are equal to each other up to first-order in $\mu_{\max}$.} 
	\begin{align}
		\E\|\tilde{\bw}_{k,i}\|_{\Sigma}^2
				&=
					\E \Big\{
						\col\{\tilde{\w}_{1,i}, \ldots, \tilde{\w}_{N,i}\}^T
						\nn\\
						&\qquad
						\cdot
						\diag\{0, \ldots, \Sigma, \ldots, 0\}
						\nn\\
						&\qquad
						\cdot
						\col\{\tilde{\w}_{1,i}, \ldots, \tilde{\w}_{N,i}\}
					\Big\}
					\nn\\
				&=
					\E \left\{
						\tilde{\w}_{i}^T (E_{kk} \otimes \Sigma) \tilde{\w}_{i}
					\right\}
					\nn\\
				&=
					\E \left( \Tr\left\{
						\tilde{\w}_{i} \tilde{\w}_{i}^T (E_{kk} \otimes \Sigma) 
					\right\}\right)
					\nn\\
				&=
					\Tr\left\{
						\E\left[ \tilde{\w}_{i} \tilde{\w}_{i}^T \right]
						(E_{kk} \otimes \Sigma) 
					\right\}
					\nn\\
				&=
					\Tr\left\{
						\Pi_i
						(E_{kk} \otimes \Sigma) 
					\right\}
		\label{Equ:SteadyState:WMSE_Tr}
	\end{align}
where $E_{kk}$ is an $M \times M$ matrix with $(k,k)$-entry equal to one and all other entries equal to zero. 
We could proceed with the analysis by deriving a recursion of $\tilde{\w}_i$ from \eqref{Equ:ProblemFormulation:DisAlg_Comb1}--\eqref{Equ:ProblemFormulation:DisAlg_Comb2} and examining the corresponding error covariance matrix, $\Pi_i$. However, we will take an alternative approach here by calling upon the following decomposition of the error quantity $\tilde{\w}_{k,i}$ from Part I\cite{chen2013learningPart1} (see Eq. \eqref{P1-Equ:DistProc:wki_tilde_decomposition_final} therein):
	\begin{align}
		\tilde{\w}_{k,i}		&=
							\tilde{w}_{c,i} - \check{\w}_{c,i} - (u_{L,k} \otimes I_M) \w_{e,i}
		\label{Equ:DistProc:wki_tilde_decomposition_final}
	\end{align}	
where $\tilde{w}_{c,i} \defeq w^o - w_{c,i}$ denotes the error of the reference recursion \eqref{Equ:LearnBehav:RefRec} relative to $w^o$, the vectors $\check{\w}_{c,i}$ and $\w_{e,i}$ are the two transformed quantities introduced in Eqs. \eqref{P1-Equ:DistProc:wicheck_blockstructure} and \eqref{P1-Equ:DistProc:w_i_prime_w_ci_w_ei} in Part I\cite{chen2013learningPart1}, and $u_{L,k}$ is the $k$th row of the matrix $U_L$ which is a sub-matrix of the transform matrix introduced in Eq. \eqref{P1-Equ:DistProc:D_U_Uinv} in Part I \cite{chen2013learningPart1}. In particular, $\check{\w}_{c,i}$ represents the error of the centroid of the iterates $\{\w_{k,i}\}$ relative to the reference recursion:
	\begin{align}
		\check{\w}_{c,i}	&\defeq		\w_{c,i} - \bar{w}_{c,i}
								\nn
	\end{align}
where the centroid $\w_{c,i}$ is defined as
	\begin{align}
		\w_{c,i}				&\defeq		\sum_{k=1}^N \theta_k \w_{k,i}
										\nn
	\end{align}
and $(u_{L,k} \otimes I_M) \w_{e,i}$ represents the error of the iterate $\w_{k,i}$ at agent $k$ relative to the centroid $\w_{c,i}$. The details and derivation of the decomposition \eqref{Equ:DistProc:wki_tilde_decomposition_final} appear in Sec. \ref{P1-Sec:Performance:ErrorRecursion} of Part I \cite{chen2013learningPart1}. Relation \eqref{Equ:DistProc:wki_tilde_decomposition_final} can also be written in the following equivalent global form:
	\begin{align}
		\tilde{\w}_{i}		&=
							\one \otimes \tilde{w}_{c,i} 
							- 
							\one \otimes \check{\w}_{c,i} 
							- 
							(U_L \otimes I_M) \w_{e,i}
		\label{Equ:DistProc:wi_tilde_decomposition_globalform_final}
	\end{align}
The major motivation to use \eqref{Equ:DistProc:wi_tilde_decomposition_globalform_final} in our steady-state analysis is that the convergence results and non-asymptotic MSE bounds already derived in Part I\cite{chen2013learningPart1} for each term in \eqref{Equ:DistProc:wi_tilde_decomposition_globalform_final} will reveal that some quantities will either disappear or become higher order terms in steady-state for small step-sizes. In particular, we are going to show that the mean-square-error of $\tilde{\w}_{i}$ is dominated by the mean-square-error of $\check{\w}_{c,i}$. Therefore, it will suffice to examine the mean-square-error of $\check{\w}_{c,i}$. We start by recalling the related non-asymptotic and asymptotic bounds from Part I\cite{chen2013learningPart1}. We derived in expression \eqref{P1-Equ:Lemma:ErrorDynamics:JointRec_wc_check} from Part I \cite{chen2013learningPart1} the following relation for $\check{\w}_{c,i}$: 
	\begin{align}
		\check{\bm{w}}_{c,i}	
							&=	T_c(\bm{w}_{c,i-1}) 
								-
								T_c(\bar{w}_{c,i-1})
								\nn\\
								&\quad
								- 
								\mu_{\max}\cdot (p^T \otimes I_M) 
								\left[
										\bm{z}_{i-1}
										+
										\bm{v}_i
								\right]
			\label{Equ:Lemma:ErrorDynamics:JointRec_wc_check}
	\end{align}
where
	\begin{align}
		T_c(x)	
				&\defeq		x - \mu_{\max}\sum_{k=1}^N p_k s_k(x)
		\label{Equ:Def:Tc2}
						\\
		\bv_i
				&\defeq	
						\hat{\bs}_{i}
						\left(
							 \bm{\phi}_{i-1}
						\right)
						\!-\!
						s\left(
							\bm{\phi}_{i-1}
						\right)
		\label{Equ:Def:v_i}
						\\
		\bz_{i-1}
				&\defeq
						s\left(
							\bm{\phi}_{i-1}
						\right)
						\!-\!
						s(\one \otimes \bw_{c,i-1})
		\label{Equ:Def:z_i1}
	\end{align}
The two perturbation terms $\bv_i(\bm{\phi}_{i-1})$ and $\bz_{i-1}$ were further shown to satisfy the following bounds in Appendix \ref{P1-Appendix:Proof_BoundsPerturbation} in Part I\cite{chen2013learningPart1}.
\begin{align}
	\label{Equ:Lemma:BoundsPerturbation:P_z}
	P[\bm{z}_{i\!-\!1}]	&\preceq		\lambda_U^2
								\cdot
								\left\|
									\bP_1[A_1^T U_L]
								\right\|_{\infty}^2
								\!\cdot\!
								\mathds{1}\mathds{1}^T
								\!\cdot\!
								P[\bm{w}_{e,i\!-\!1}]
								\\
	\label{Equ:Lemma:BoundsPerturbaton:P_s}
	P[s(\mathds{1}\!\otimes\! \bm{w}_{c,i\!-\!1})]
					&\preceq		3\lambda_U^2
								\!\cdot\!
								P[\check{\bm{w}}_{c,i - 1}]
								\!\cdot\!
								\mathds{1}
								\!+\!
								3\lambda_U^2 \|\tilde{w}_{c,0}\|^2 \cdot \mathds{1}
								\!+\!
								3g^o
								\\
	\E \{P[\bm{v}_i] | \mF_{i-1} \}	
					&\preceq		
								4\alpha \cdot \one
								\cdot
								P[ \check{\w}_{c,i-1} ]
								\nn\\
								&\quad
								+
								4 \alpha
								\cdot
								\| \bP[ \mA_1^T \mU_L ] \|_{\infty}^2
								\cdot
								\one \one^T
								P[ \w_{e,i-1} ]
								\nn\\
								&\quad
								+\!
								\left[
									4\alpha
									\cdot
									( \|\tilde{w}_{c,0}\|^2 \!+\! \|w^o\|^2 )
									\!+\!
									\sigma_v^2
								\right]
								\cdot
								\one
	\label{Equ:Lemma:BoundsPerturbation:P_v_E_Fiminus1}
								\\
	\E P[\bm{v}_i]
					&\preceq		
								4\alpha \cdot \one
								\cdot
								\Expt P[ \check{\w}_{c,i-1} ]
								\nn\\
								&\quad
								+
								4 \alpha
								\cdot
								\| \bP[ \mA_1^T \mU_L ] \|_{\infty}^2
								\cdot
								\one \one^T
								\Expt P[ \w_{e,i-1} ]
								\nn\\
								&\quad
								+\!
								\left[
									4\alpha
									\cdot
									( \|\tilde{w}_{c,0}\|^2 \!+\! \|w^o\|^2 )
									\!+\!
									\sigma_v^2
								\right]
								\cdot
								\one
	\label{Equ:Lemma:BoundsPerturbation:P_v}				
\end{align}
where $P[\check{\w}_{c,i-1}]=\|\check{\w}_{c,i-1}\|^2$, and $g^o \defeq P[s(\mathds{1} \otimes w^o)]$.
We further showed in Eqs. \eqref{P1-Equ:Cor:EPwcicheck_asymptotic_bound} and \eqref{P1-Equ:Cor:EPwei_asymptotic_bound} from Part I \cite{chen2013learningPart1} that 
		\begin{align}
			\limsup_{i\rightarrow \infty} \E \|\check{\bw}_{c,i}\|^2
								&\le		O(\mu_{\max})
			\label{Equ:Cor:EPwcicheck_asymptotic_bound}
											\\
			\limsup_{i\rightarrow \infty} \E \|\bw_{e,i}\|^2
								&\le 
											O(\mu_{\max}^2)
			\label{Equ:Cor:EPwei_asymptotic_bound}
		\end{align}

\subsection{Approximation of $\Pi_{i}$ by $\one\one^T \otimes \check{\Pi}_{c,i}$}
\label{Appendix:Proof_Thm_SteadyStatePerf:ApproxPiByPic}

In order to examine $\Pi_i$, which is needed for the limiting value of \eqref{Equ:SteadyState:WMSE_Tr}, we first establish the result \eqref{Equ:SteadyState:limsup_Pii_Picicheck} further ahead using \eqref{Equ:DistProc:wi_tilde_decomposition_globalform_final}: in steady-state, the error covariance matrix of $\tilde{\bw}_{i}$ (i.e., $\Pi_{i}$) is equal to the error covaraince matrix of the component $\one \otimes \check{\w}_{c,i}$ to the first order in $\mu_{\max}$. Indeed, let $\check{\Pi}_{c,i}$ denote the covariance matrix of 
$\check{\bw}_{c,i}$, i.e., $\check{\Pi}_{c,i} \defeq \E\{\check{\bw}_{c,i}\check{\bw}_{c,i}^T\}$. By \eqref{Equ:DistProc:wi_tilde_decomposition_globalform_final}, we have
	\begin{align}
		\Pi_i			&=
							\E \left\{
								\tilde{\w}_i \tilde{\w}_i^T
							\right\}
							\nn\\
					&=
							\one\one^T 
							\otimes
							[ \tilde{w}_{c,i} \tilde{w}_{c,i}^T ]
							+
							\one\one^T 
							\otimes
							\check{\Pi}_{c,i}
							\nn\\
							&\quad
							+
							\E
							\left\{
								[ (U_L \otimes I_M) \w_{e,i} ]
								[ (U_L \otimes I_M) \w_{e,i} ]^T
							\right\}
							\nn\\
							&\quad
							-
							(\one \otimes \tilde{w}_{c,i})
							\big(
								\one \otimes \E \check{\w}_{c,i} + (U_L \otimes I_M) \E \w_{e,i}
							\big)^T
							\nn\\
							&\quad
							-
							\big(
								\one \otimes \E \check{\w}_{c,i} + (U_L \otimes I_M) \E \w_{e,i}
							\big)
							(\one \otimes \tilde{w}_{c,i})^T
							\nn\\
							&\quad
							+
							\E \left\{
								(\one\otimes \check{\w}_{c,i})
								[ (U_L \otimes I_M) \w_{e,i} ]^T
							\right\}
							\nn\\
							&\quad
							+
							\E \left\{
								[ (U_L \otimes I_M) \w_{e,i} ]
								(\one\otimes \check{\w}_{c,i})^T
							\right\}
							\nn
	\end{align}
so that
	\begin{align}
		&\left\| 
			\Pi_i - \one\one^T \otimes \check{\Pi}_{c,i}
		\right\|				
							\nn\\
				&\overset{(a)}{\le }
							\left\|
								\one\one^T 
								\otimes
								[ \tilde{w}_{c,i} \tilde{w}_{c,i}^T ]
							\right\|
							\nn\\
							&\quad
							+
							\left\|
								\E
								\left\{
									[ (U_L \otimes I_M) \w_{e,i} ]
									[ (U_L \otimes I_M) \w_{e,i} ]^T
								\right\}
							\right\|
							\nn\\
							&\quad
							+
							2\left\|\one \otimes \tilde{w}_{c,i}\right\|
							\cdot
							\big\|
								\one \otimes \E \check{\w}_{c,i} + (U_L \otimes I_M) \E \w_{e,i}
							\big\|
							\nn\\
							&\quad
							+
							2\left\|
							\E \left\{
								(\one\otimes \check{\w}_{c,i})
								[ (U_L \otimes I_M) \w_{e,i} ]^T
							\right\}
							\right\|
							\nn\\
				&\overset{(b)}{\le }
							\left\|
								\one\one^T 
								\otimes
								[ \tilde{w}_{c,i} \tilde{w}_{c,i}^T ]
							\right\|
							+
							\E
							\left\|
									(U_L \otimes I_M) \w_{e,i} ]
							\right\|^2
							\nn\\
							&\quad
							+
							2\left\|\one \otimes \tilde{w}_{c,i}\right\|
							\cdot
							\big\|
								\one \otimes \E \check{\w}_{c,i} + (U_L \otimes I_M) \E \w_{e,i}
							\big\|
							\nn\\
							&\quad
							+
							2\left\|
							\E \left\{
								(\one\otimes \check{\w}_{c,i})
								[ (U_L \otimes I_M) \w_{e,i} ]^T
							\right\}
							\right\|
							\nn
	\end{align}
where step (a) uses triangular inequality, and step (b) applies Jensen's inequality $\| \E [ \cdot ] \| \le \E \| \cdot \|$ to the convex function $\| \cdot \|$ and the inequality $\| x y^T \| \le \| x \| \cdot \| y\|$. Taking $\limsup$ of both sides as $i \rightarrow \infty$, we obtain
	\begin{align}
		&\limsup_{i\rightarrow\infty}
		\left\| 
			\Pi_i - \one\one^T \otimes \check{\Pi}_{c,i}
		\right\|
							\nn\\
				&\le 
							\limsup_{i\rightarrow\infty}
							\E
							\left\|
									(U_L \otimes I_M) \w_{e,i} ]
							\right\|^2
							\nn\\
							&\quad
							+
							\limsup_{i\rightarrow\infty}
							\left\{
								2\left\|
								\E \left\{
									(\one\otimes \check{\w}_{c,i})
									[ (U_L \otimes I_M) \w_{e,i} ]^T
								\right\}
								\right\|
							\right\}
		\label{Equ:Appendix:Pi_Piccheck_gapbound_interm1}
	\end{align}
since $\tilde{w}_{c,i} \rightarrow 0$  as $i\rightarrow\infty$ according to Theorem \ref{P1-Thm:ConvergenceRefRec:DeterministcCent} in Part I\cite{chen2013learningPart1}. We now bound the two terms on the right-hand side of \eqref{Equ:Appendix:Pi_Piccheck_gapbound_interm1} using \eqref{Equ:Cor:EPwcicheck_asymptotic_bound}--\eqref{Equ:Cor:EPwei_asymptotic_bound} and show that they are higher order terms of $\mu_{\max}$. By \eqref{Equ:Cor:EPwei_asymptotic_bound}, the first term on the right-hand side of \eqref{Equ:Appendix:Pi_Piccheck_gapbound_interm1} is $O(\mu_{\max}^2)$ because
	\begin{align}
		&\limsup_{i\rightarrow\infty}
		\E
		\left\|
				(U_L \otimes I_M) \w_{e,i} ]
		\right\|^2
						\nn\\
				&\quad\le 		
						\limsup_{i\rightarrow\infty}
						\left\|
							U_L \otimes I_M
						\right\|^2
						\cdot
						\E\|\bw_{e,i}\|^2
				\le		O(\mu_{\max}^2)
		\label{Equ:Appendix:Pi_Piccheck_gapbound_interm1_term1}
	\end{align} 
Moreover, for any random variables $\bx$ and $\by$, we have 
$|\E\{\bx\by\}|^2 \le \E\{\bx^2\} \cdot \E\{\by^2\}$. Applying this result to the last term in \eqref{Equ:Appendix:Pi_Piccheck_gapbound_interm1} we have 
	\begin{align}
		\big\|&
			\E \left[
				(\one \otimes \check{\bw}_{c,i})
				[(U_L \otimes I_M) \bw_{e,i}]^T
			\right]
		\big\|
					\nn\\
				&\le 	\sqrt{
							\E\|\one\otimes\check{\bw}_{c,i}\|^2
							\cdot
							\E\|(U_L \otimes I_M) \bw_{e,i}\|^2
						}
		\label{Equ:Appendix:Pi_Piccheck_gapbound_interm1_term2}
	\end{align}
Using \eqref{Equ:Cor:EPwcicheck_asymptotic_bound} and \eqref{Equ:Cor:EPwei_asymptotic_bound}, we conclude that
	\begin{align}
		\limsup_{i \rightarrow \infty}
		\big\| 
			\E \left[ \!
				(\one \otimes \check{\bw}_{c,i})
				[(U_L \otimes I_M) \bw_{e,i}]^T
				\!
			\right] \!
		\big\|
				\leq O(\mu_{\max}^{3/2})  \!\!
		\label{Equ:Appendix:Piinf_Pici_2ndTerm_limsupBound}
	\end{align}
Therefore, substituting \eqref{Equ:Appendix:Pi_Piccheck_gapbound_interm1_term1} and \eqref{Equ:Appendix:Piinf_Pici_2ndTerm_limsupBound} into \eqref{Equ:Appendix:Pi_Piccheck_gapbound_interm1}, we conclude that
	\begin{align}
		\limsup_{i\rightarrow\infty}
		\left\| 
			\Pi_i - \one\one^T \otimes \check{\Pi}_{c,i}
		\right\|
				&\le 
						O(\mu_{\max}^{3/2})
		\label{Equ:SteadyState:limsup_Pii_Picicheck}
	\end{align}
%so that, as claimed earlier, in steady-state ($i \rightarrow \infty$):
%	\begin{align}
%		\Pi_{\infty}
%				=	\one\one^T \otimes \check{\Pi}_{c,\infty}
%							+ O(\mu_{\max}^{3/2})
%		\label{Equ:SteadyState:Pi_inf_Piccheck_inf}
%	\end{align}
%{\color{blue}
%Note that \eqref{Equ:SteadyState:Pi_inf_Piccheck_inf} is an intuitive interpretation of the more formal result \eqref{Equ:SteadyState:limsup_Pii_Picicheck} as the limits of $\Pi_i$ and $\check{\Pi}_{c,i}$ may not exist.
%}

\subsection{Approximation of $\check{\Pi}_{c,i}$ by $\check{\Pi}_{a,i}$}
\label{Appendix:Proof_Thm_SteadyStatePerf:ApproxPicByPia}
Now we examine the expression for $\check{\Pi}_{c,i}$ at steady-state ($i \rightarrow \infty$) to arrive at further expression \eqref{Equ:SteadyState:Pici_Pai_gap_limsupBound}. To do this, we rewrite expressions \eqref{Equ:Lemma:ErrorDynamics:JointRec_wc_check}--\eqref{Equ:Def:z_i1} for $\check{\bw}_{c,i}$ as
	\begin{align}
		\check{\bm{w}}_{c,i}	
						&=	\check{\bm{w}}_{c,i-1}
							\!-\!
							\mu_{\max}
							\sum_{k=1}^N p_k 
							\big[
								s_k(\bw_{c,i-1})
								\!-\!
								s_k(\bar{w}_{c,i-1})
							\big]
							\nn\\
							&\quad
							- 
							\mu_{\max}\cdot (p^T \otimes I_M) 
							\left[
									\bm{z}_{i-1}
									+
									\bm{v}_i
							\right]
							\nn\\
						&=	
							\left[ 
									I_M 
									- 
									\mu_{\max}
									H_c
							\right]
							\check{\bw}_{c,i-1}
							-
							\mu_{\max}\cdot (p^T \otimes I_M) 
							\bm{v}_i	
							\nn\\
							&\quad							
							- 
							\mu_{\max}\cdot (p^T \otimes I_M) \bm{z}_{i-1}
							\nn\\
							&\quad
							- \mu_{\max} 
							\big(
								s_c( \w_{c,i-1} )
								-
								s_c(\bar{w}_{c,i-1})
								-
								H_c 
								\check{\w}_{c,i-1}
							\big)
		\label{Equ:SteadyState:wccheck_recursion_interm1}
	\end{align}
where
	\begin{align}
		H_c					&\defeq	\sum_{k=1}^N p_k \nabla_{w^T} s_k(w^o)
									\label{Equ:SteadyState:H_c_def}
									\\
		s_c(w)				&\defeq	\sum_{k=1}^N p_k  s_k(w)
	\end{align}
Next, we show that the mean-square-error between $\check{\bw}_{c,i}$ generated by \eqref{Equ:SteadyState:wccheck_recursion_interm1} and the $\check{\bw}_{a,i}$ generated by the following auxiliary recursion is small for small step-sizes:
	\begin{align}
		\check{\bw}_{a,i}	=		\left[ 
											I_M 
											- 
											\mu_{\max}
											H_c
									\right]
									\check{\bw}_{a,i-1}
									-
									\mu_{\max}\cdot (p^T \otimes I_M) 
									\bm{v}_i
		\label{Equ:SteadyState:wacheck_recursion}
	\end{align}
Indeed, subtracting \eqref{Equ:SteadyState:wacheck_recursion} from \eqref{Equ:SteadyState:wccheck_recursion_interm1} leads to
	\begin{align}
		\check{\bw}_{c,i} - \check{\bw}_{a,i}
							&=		\left[ 
											I_M 
											- 
											\mu_{\max}
											H_c
									\right]
									(
										\check{\bw}_{c,i-1}
										-
										\check{\bw}_{a,i-1}
									)
									\nn\\
									&\quad
									-
									\mu_{\max} 
									\big(
										s_c( \w_{c,i-1} )
										-
										s_c(\bar{w}_{c,i-1})
										-
										H_c 
										\check{\w}_{c,i-1}
									\big)
									\nn\\
									&\quad
									-
									\mu_{\max}
									\cdot
									(p^T \otimes I_M) \z_{i-1}	
		\label{Equ:Appendix:wcicheck_waicheck_gap_interm1}							
	\end{align}
We recall the definition of the scalar factor $\gamma_c$ from Eq. \eqref{P1-Equ:VarPropt:gamma_c} in Part I \cite{chen2013learningPart1}:
	\begin{align}
           	\label{Equ:VarPropt:gamma_c}
           	\gamma_c		\defeq	1 - \mu_{\max}\lambda_L 
           						+ \frac{1}{2}\mu_{\max}^2\|p\|_1^2 \lambda_U^2
           \end{align}
Now evaluating the squared Euclidean norm of both sides of \eqref{Equ:Appendix:wcicheck_waicheck_gap_interm1}, we get
	\begin{align}
		\|&\check{\bw}_{c,i} - \check{\bw}_{a,i}\|^2
									\nn\\
							&=		
									\big\|
										\gamma_c
										\cdot
										\frac{1}{\gamma_c}
										\left[ 
												I_M 
												- 
												\mu_{\max}
												H_c
										\right]
										(
											\check{\bw}_{c,i-1}
											-
											\check{\bw}_{a,i-1}
										)
										\nn\\
										&\quad
										+
										\frac{1-\gamma_c}{2}
										\cdot
										\frac{-2\mu_{\max}}{1-\gamma_c}
										\cdot	
										\big(
											s_c( \w_{c,i-1} )
											\!-\!
											s_c(\bar{w}_{c,i-1})
											\!-\!
											H_c 
											\check{\w}_{c,i-1}
										\big)
										\nn\\
										&\quad
										+
										\frac{1-\gamma_c}{2}
										\cdot
										\frac{-2\mu_{\max}}{1-\gamma_c}
										\cdot
										(p^T \!\otimes\! I_M) \z_{i-1}
									\big\|^2
									\nn\\
							&\overset{(a)}{\le}
									\gamma_c
									\cdot
									\big\|										
										\frac{1}{\gamma_c}
										\left[ 
												I_M 
												\!-\! 
												\mu_{\max}
												H_c
										\right]
										(
											\check{\bw}_{c,i-1}
											\!-\!
											\check{\bw}_{a,i-1}
										)
									\big\|^2
									\nn\\
									&\quad
									+\!
									\frac{1 \!-\! \gamma_c}{2}
									\!\cdot\!
									\big\|
										\frac{-2\mu_{\max}}{1-\gamma_c}
										\!\cdot\!	
										\big(
											s_c( \w_{c,i-1} )
											\!-\!
											s_c(\bar{w}_{c,i-1})
											\!-\!
											H_c 
											\check{\w}_{c,i-1}
										\big)
									\big\|^2
										\nn\\
										&\quad
										+
									\frac{1-\gamma_c}{2}
									\cdot
									\big\|
										\frac{-2\mu_{\max}}{1-\gamma_c}
										\cdot
										(p^T \!\otimes\! I_M) \z_{i-1}
									\big\|^2
									\nn\\							
							&\le 	\frac{1}{\gamma_c}
									\cdot
									\|	
												I_M 
												- 
												\mu_{\max}
												H_c
									\|^2
									\cdot
									\|
											\check{\bw}_{c,i-1}
											-
											\check{\bw}_{a,i-1}
									\|^2
									\nn\\
									&\quad
										+
									\frac{2\mu_{\max}^2}{1-\gamma_c}
									\cdot
									\big\|
										s_c( \w_{c,i-1} )
										-
										s_c(\bar{w}_{c,i-1})
										-
										H_c 
										\check{\w}_{c,i-1}
									\big\|^2
										\nn\\
										&\quad
										+
									\frac{2\mu_{\max}^2}{1-\gamma_c}
									\cdot
									\|
										(p^T \!\otimes\! I_M)
									\|^2
									\cdot
									\|
										\z_{i-1}
									\|^2
									\nn\\
							&\overset{(b)}{=}
									\frac{1}{\gamma_c}
									\cdot
									\|	
												B_c
									\|^2
									\cdot
									\|
											\check{\bw}_{c,i-1}
											-
											\check{\bw}_{a,i-1}
									\|^2
									\nn\\
									&\quad
										+
									\frac{2\mu_{\max}}
									{\lambda_L \!-\! \frac{1}{2}\mu_{\max} \|p\|_1^2 \lambda_U^2}
									\!\cdot\!
									\big\|
										s_c( \w_{c,i-1} )
										\!-\!
										s_c(\bar{w}_{c,i-1})
										\!-\!
										H_c 
										\check{\w}_{c,i-1}
									\big\|^2
										\nn\\
										&\quad
										+
									\frac{2\mu_{\max}}
									{\lambda_L \!- \!\frac{1}{2}\mu_{\max} \|p\|_1^2 \lambda_U^2}
									\!\cdot\!
									\|
										(p^T \!\otimes\! I_M)
									\|^2
									\!\cdot\!
									\one^T P[\z_{i-1}] 
		\label{Equ:SteadyState:Recursion_wcwa_interm1}
	\end{align}
where in step (a) we used the convexity of the squared norm $\|\cdot\|^2$, and in step (b) we introduced $B_c \defeq I_M-\mu_{\max}H_c$. We now proceed to bound the three terms on the right-hand side of the above inequality. First note that
	\begin{align}
		B_c^T B_c	&=	(I - \mu_{\max} H_c)^T (I - \mu_{\max} H_c)	
						\nonumber\\
					&=	I - \mu_{\max} (H_c + H_c^T) + \mu_{\max}^2 H_c^T H_c
		\label{Equ:SteadyState:Bc_UB}
	\end{align}
Under Assumption \ref{Assumption:JacobianUpdatVectorLipschitz}, conditions \eqref{Equ:Lemma:EquivCondUpdateVec:LipschitzUpdate_HessianUB} and \eqref{Equ:Lemma:EquivCondUpdateVec:StrongMono_HessianLB} hold in the ball $||\delta w\|\leq r_H$ around $w^o$. Recall from \eqref{Equ:SteadyState:H_c_def} that $H_c$ is evaluated at $w^o$. Therefore, from \eqref{Equ:Lemma:EquivCondUpdateVec:StrongMono_HessianLB} we have 
	\begin{align}
		H_c + H_c^T	&\ge	2\lambda_L \cdot I_M
		\label{Equ:SteadyState:Hc_LB}
	\end{align}
and by \eqref{Equ:Lemma:EquivCondUpdateVec:LipschitzUpdate_HessianUB}, we have
	\begin{align}
		\| H_c \|		&=	\left\|
							\sum_{k=1}^N p_k \nabla_{w^T} s_k(w^o)
						\right\|
						\nn\\
					&\le 
						\sum_{k=1}^N p_k \| \nabla_{w^T} s_k(w^o)\|
						\nn\\
					&\le 	
						\sum_{k=1}^N p_k \cdot \lambda_U
					=	\|p\|_1 \cdot \lambda_U	
						\nn
	\end{align}
Note further that $\| H_c \|^2 \equiv \lambda_{\max}(H_c^T H_c)$, where $\lambda_{\max}(\cdot)$ denotes the largest eigenvalue of the matrix argument. This implies that
	\begin{align}
		0	\le 	H_c^T H_c	\le 	\|p\|_1^2  \lambda_U^2 \cdot I_M
		\label{Equ:SteadyState:Hc_UB}
	\end{align}
Substituting \eqref{Equ:SteadyState:Hc_LB} and \eqref{Equ:SteadyState:Hc_UB} into \eqref{Equ:SteadyState:Bc_UB}, we obtain
	\begin{align}
		B_c^T B_c
					&\le 
						\left(
								1 
								- 
								2\mu_{\max} \lambda_L 
								+ 
								\mu_{\max}^2 \|p\|_1^2 \lambda_U^2
						\right) \cdot
						I
						\nn
	\end{align}
so that
	\begin{align}
		\|B_c\|^2	&\le 
						1 
						- 
						2\mu_{\max} \lambda_L 
						+ 
						\mu_{\max}^2 \|p\|_1^2 \lambda_U^2
						\nn\\
					&\le 
						\left(1 
						- 
						\mu_{\max} \lambda_L 
						+ 
						\frac{1}{2}\mu_{\max}^2 \|p\|_1^2 \lambda_U^2 
						\right)^2
					=
						\gamma_c^2
		\label{Equ:Appendix:Bc_norm_UB}
	\end{align}
where in the last inequality we used $(1-x) \le (1-\frac{1}{2}x)^2$. Next, we bound the second term on the right-hand side of \eqref{Equ:SteadyState:Recursion_wcwa_interm1}. To do this, we need to bound it in two separate cases:
	\begin{enumerate}
		\item
		{\bf Case 1:} $\| \tilde{w}_{c,i-1} \| + \| \check{\w}_{c,i-1} \| \le r_H$\\
		This condition implies that, for any $0 \le t \le 1$, the vector $\bar{w}_{c,i-1} + t \check{\w}_{c,i-1}$ is inside
		a ball that is centered at $w^o$ with radius $r_H$ since:
			\begin{align}
				\| (\bar{w}_{c,i-1} + t \check{\w}_{c,i-1}) - w^o \|
							&=		\| -\tilde{w}_{c,i-1} + t\check{\w}_{c,i-1} \|
									\nn\\
							&\le 		\| \tilde{w}_{c,i-1} \| + t \| \check{\w}_{c,i-1} \|
									\nn\\
							&\le 		\| \tilde{w}_{c,i-1} \| +  \| \check{\w}_{c,i-1} \|
									\nn\\
							&\le 		r_H
									\nn
			\end{align}
		By Assumption \ref{Assumption:JacobianUpdatVectorLipschitz}, the function $s_k(w)$ is differentiable at 
		$\bar{w}_{c,i-1} + t \check{\w}_{c,i-1}$ so that using the following mean-value theorem
		\cite[p.6]{poliak1987introduction}:
				\begin{align}
					&s_k(\bw_{c,i-1})
							=	
									s_k(\bar{w}_{c,i-1})
									\nn\\
									&\quad
									+
									\left(
										\int_{0}^{1}
											\nabla_{w^T}
											s_k(\bar{w}_{c,i-1}+t\check{\bw}_{c,i-1})
											dt
									\right)
									\cdot
									\check{\bw}_{c,i-1}	
					\label{Equ:SteadyState:MeanValThm}			
				\end{align}
		Then, we have
			\begin{align}
				\big\|&
					s_c( \w_{c,i-1} )
					-
					s_c(\bar{w}_{c,i-1})
					-
					H_c 
					\check{\w}_{c,i-1}
				\big\|^2				\nn\\
							&=		
									\big\|
										\sum_{k=1}^N
										p_k
										[ s_k( \w_{c,i-1} ) - s_k( \bar{w}_{c,i-1} ) ]
										- 
										H_c 
										\check{\w}_{c,i-1}
									\big\|^2
									\nn\\
							&=		\big\|
										\sum_{k=1}^N
										p_k
										\int_{0}^{1}
											\nabla_{w^T}
											s_k(\bar{w}_{c,i-1}+t\check{\bw}_{c,i-1})
											dt
										\cdot
										\check{\bw}_{c,i-1}	
										\nn\\
										&\quad
										-
										\sum_{k=1}^N 
										p_k
										\nabla_{w^T} s_k(w^o)
										\cdot
										\check{\w}_{c,i-1}
									\big\|^2
									\nn\\
							&=		\Big\|
										\sum_{k=1}^N
										p_k \cdot
											\int_{0}^{1}
											\big[
												\nabla_{w^T}
												s_k(\bar{w}_{c,i-1}+t\check{\bw}_{c,i-1})
												\nn\\
												&\qquad
												-
												\nabla_{w^T} s_k(w^o)
											\big]
											dt
										\cdot
										\check{\w}_{c,i-1}
									\Big\|^2
									\nn\\
							&\le 		
									\Big\{
										\sum_{k=1}^N
										p_k \cdot
											\int_{0}^{1}
											\big\|
												\nabla_{w^T} 
												s_k(\bar{w}_{c,i-1}+t\check{\bw}_{c,i-1})
												\nn\\
												&\qquad
												-
												\nabla_{w^T} 
												s_k(w^o)
											\big\|
											dt
										\cdot
										\|\check{\w}_{c,i-1}\|
									\Big\}^2
									\nn\\
							&\overset{(a)}{\le} 
									\Big\{
										\sum_{k=1}^N
										p_k \cdot
											\int_{0}^{1}
											\lambda_H
											\cdot
											\|
												(\bar{w}_{c,i-1}+t\check{\bw}_{c,i-1})
												-
												w^o
											\|
											dt
										\nn\\
										&\qquad
										\cdot
										\|\check{\w}_{c,i-1}\|
									\Big\}^2
									\nn\\
							&\le		
									\Big\{
										\sum_{k=1}^N
										p_k \cdot
											\int_{0}^{1}
											\lambda_H
											\cdot
											(
												\| \bar{w}_{c,i-1} - w^o \|
												+
												t
												\|\check{\w}_{c,i-1}\|
											)
											dt
										\nn\\
										&\qquad
										\cdot
										\|\check{\w}_{c,i-1}\|
									\Big\}^2
									\nn\\
							&\le 		\Big\{
										\sum_{k=1}^N
										p_k \cdot
											\int_{0}^{1}
											\lambda_H
											\cdot
											(
												\| \bar{w}_{c,i-1} - w^o \|
												+
												\|\check{\w}_{c,i-1}\|
											)
											dt
										\nn\\
										&\qquad
										\cdot
										\|\check{\w}_{c,i-1}\|
									\Big\}^2
									\nn\\
							&= 		\Big\{
										\sum_{k=1}^N
										p_k \cdot
											\lambda_H
											\cdot
											(
												\| \tilde{w}_{c,i-1} \|
												+
												\|\check{\w}_{c,i-1}\|
											)
										\cdot
										\|\check{\w}_{c,i-1}\|
									\Big\}^2
									\nn\\
							&=		\Big\{
										\|p\|_1 
										\cdot
											\lambda_H
											\cdot
											(
												\| \tilde{w}_{c,i-1} \|
												+
												\|\check{\w}_{c,i-1}\|
											)
										\cdot
										\|\check{\w}_{c,i-1}\|
									\Big\}^2
									\nn\\
							&=		\|p\|_1^2
									\cdot
									\lambda_H^2
									\cdot
										(
											\| \tilde{w}_{c,i-1} \|
											+
											\|\check{\w}_{c,i-1}\|
										)^2
									\cdot
									\|\check{\w}_{c,i-1}\|^2
									\nn\\
							&\le 		2\|p\|_1^2
									\cdot
									\lambda_H^2
									\cdot
										(
											\| \tilde{w}_{c,i-1} \|^2
											+
											\|\check{\w}_{c,i-1}\|^2
										)
									\cdot
									\|\check{\w}_{c,i-1}\|^2
				\label{Equ:SteadyState:HessianPerturb_2nd_Case1}
			\end{align}
		where step (a) uses
		Assumption \ref{Assumption:JacobianUpdatVectorLipschitz} and the last inequality uses 
		$(x+y)^2 \le 2x^2 + 2y^2$.

		\item
		{\bf Case 2:} $\| \tilde{w}_{c,i-1} \| + \| \check{\w}_{c,i-1} \| > r_H$\\
		It holds that
			\begin{align}
					\big\|&
						s_c( \w_{c,i-1} )
						-
						s_c(\bar{w}_{c,i-1})
						-
						H_c 
						\check{\w}_{c,i-1}
					\big\|^2				\nn\\
								&=		
										\big\|
											\sum_{k=1}^N
											p_k
											\big[ 
											s_k( \w_{c,i-1} ) - s_k( \bar{w}_{c,i-1} ) 
											\nn\\
											&\qquad
											- 
											\nabla_{w^T} s_k(w^o)
											\cdot
											\check{\w}_{c,i-1}
											\big]
										\big\|^2
										\nn\\
								&\le 
										\Big\{
											\sum_{k=1}^N
											p_k
											\big\|
												s_k( \w_{c,i-1} ) - s_k( \bar{w}_{c,i-1} ) 
												\nn\\
												&\qquad
												- 
												\nabla_{w^T} s_k(w^o)
												\cdot
												\check{\w}_{c,i-1}
											\big\|
										\Big\}^2
										\nn\\
								&\le 	
										\Big\{
											\sum_{k=1}^N
											p_k		
											\big(								
												\|s_k( \w_{c,i-1} ) - s_k( \bar{w}_{c,i-1} ) \|
												\nn\\
												&\qquad
												+
												\|
												\nabla_{w^T} s_k(w^o)
												\|
												\cdot
												\|
												\check{\w}_{c,i-1}
												\|
											\big)
										\Big\}^2
										\nn\\
								&\overset{(a)}{\le}
										\Big\{
											\sum_{k=1}^N
											p_k	
											\cdot
											(
												\lambda_U
												\cdot
												\|\check{\w}_{c,i-1} \|
												+
												\lambda_U
												\cdot
												\|
													\check{\w}_{c,i-1}
												\|
											)
										\Big\}^2
										\nn\\
								&=		
										\Big\{
											2\|p\|_1
											\cdot								
											\lambda_U
											\cdot
											\|\check{\w}_{c,i-1}\|
										\Big\}^2
										\nn\\
								&\le 
										4\|p\|_1^2
										\cdot								
										\lambda_U^2
										\cdot
										\|\check{\w}_{c,i-1}\|^2
										\nn\\
								&\overset{(b)}{\le}
										4\|p\|_1^2
										\cdot								
										\lambda_U^2
										\cdot
										\left(
											\frac{\|\tilde{w}_{c,i-1}\| + \|\check{\w}_{c,i-1}\|}
											{r_H}
										\right)^2
										\cdot
										\|\check{\w}_{c,i-1}\|^2
										\nn\\
								&\overset{(c)}{\le}
										2\|p\|_1^2
										\cdot
										\frac{4\lambda_U^2}{r_H^2}
										\cdot
											(
												\| \tilde{w}_{c,i-1} \|^2
												+
												\|\check{\w}_{c,i-1}\|^2
											)
										\cdot
										\|\check{\w}_{c,i-1}\|^2
				\label{Equ:SteadyState:HessianPerturb_2nd_Case2}
			\end{align}			
		where in step (a) we used \eqref{Equ:Assumption:Lipschitz} and
		\eqref{Equ:Lemma:EquivCondUpdateVec:LipschitzUpdate_HessianUB}, in step (b) we used
		the fact that  $\| \tilde{w}_{c,i-1} \| + \| \check{\w}_{c,i-1} \| > r_H$ in the current
		case, and in step (c) we used the relation $\| x + y \|^2 \le 2 \| x \|^2 + 2 \| y \|^2$.
	\end{enumerate}
Based on \eqref{Equ:SteadyState:HessianPerturb_2nd_Case1} and \eqref{Equ:SteadyState:HessianPerturb_2nd_Case2} from both cases, we have
	\begin{align}
			\big\|&
				s_c( \w_{c,i-1} )
				-
				s_c(\bar{w}_{c,i-1})
				-
				H_c 
				\check{\w}_{c,i-1}
			\big\|^2
										\nn\\
							&\le 		
										2\|p\|_1^2
										\!\cdot\!
										\lambda_{HU}^2
										\!\cdot\!
											(
												\| \tilde{w}_{c,i-1} \|^2
												\!+\!
												\|\check{\w}_{c,i-1}\|^2
											)
										\cdot
										\|\check{\w}_{c,i-1}\|^2 \!\!
			\label{Equ:SteadyState:HessianPerturb_2nd_final}
	\end{align}		
where
	\begin{align}
		\lambda_{HU}^2		\defeq		\max\left\{
											\lambda_H^2, 
											\frac{4\lambda_U^2}{r_H^2}
										\right\}
										\nn
	\end{align}	
The third term on the right-hand side of \eqref{Equ:SteadyState:Recursion_wcwa_interm1} can be bounded by \eqref{Equ:Lemma:BoundsPerturbation:P_z}. Therefore, substituting \eqref{Equ:Appendix:Bc_norm_UB}, \eqref{Equ:SteadyState:HessianPerturb_2nd_final} and \eqref{Equ:Lemma:BoundsPerturbation:P_z} into \eqref{Equ:SteadyState:Recursion_wcwa_interm1} and applying the expectation operator, we get
	\begin{align}
		\E&\|\check{\w}_{c,i} - \check{\w}_{a,i}\|^2
						\nn\\
				&\le 	
						\gamma_c 
						\cdot
						\E\|\check{\w}_{c,i-1} - \check{\w}_{a,i-1}\|^2
						\nn\\
						&\quad
						+
						\frac{4\mu_{\max} \|p\|_1^2 \lambda_{HU}^2}
						{\lambda_L - \frac{1}{2} \mu_{\max} \|p\|_1^2 \lambda_U^2}
						\nn\\
						&\qquad
						\cdot
						\big(
							\E\|\check{\w}_{c,i-1}\|^2 \cdot \|\tilde{w}_{c,i-1}\|^2
							+
							\E\|\check{\w}_{c,i-1}\|^4
						\big)
						\nn\\
						&\quad
						+\!\!
						\frac{2N\mu_{\max} \|p^T \!\otimes\! I_M\|^2 }
						{\lambda_L \!-\! \frac{1}{2} \mu_{\max} \|p\|_1^2 \lambda_U^2}
						\!\cdot\!
						\lambda_U^2
						\!\cdot\!
						\| \bP_1[A_1^T U_1] \|_{\infty}^2
						\!\cdot\!
						\E \|\w_{e,i-1}\|^2
		\label{Equ:SteadyState:Recursion_wcwa_interm2}
	\end{align}
where in the last term on the right-hand side of \eqref{Equ:SteadyState:Recursion_wcwa_interm2} we used  $\one^T P[x] = \|x\|^2$ from property \eqref{P1-Equ:Properties:PX_EuclNorm} in Part I\cite{chen2013learningPart1}. Recall from Theorem \ref{P1-Thm:ConvergenceRefRec:DeterministcCent} in Part I\cite{chen2013learningPart1} that $\tilde{w}_{c,i-1} \rightarrow 0$, and from \eqref{Equ:Cor:EPwcicheck_asymptotic_bound}--\eqref{Equ:Cor:EPwei_asymptotic_bound} that $\E\|\check{\bw}_{c,i-1}\|^2 \le O(\mu_{\max})$ and $\E \|\w_{e,i-1}\|^2 \le O(\mu_{\max}^2)$ in steady-state. Moreover, we also have the following result regarding $\E\|\check{\w}_{c,i-1}\|^4$ in steady-state.
	\begin{lemma}[Asymptotic bound on the 4th order moment]
		\label{Lemma:AsymptoticBound_4thOrderMoments}
		Using Assumptions \ref{Assumption:Network}--%
		\ref{Assumption:GradientNoise4thOrderMoment},
		it holds that
			\begin{align}
				\limsup_{i \rightarrow \infty}
								\E \| \check{\w}_{c,i} \|^4
						&\le 	O(\mu_{\max}^2)
				\label{Equ:Lemma:AsympBound_4thOrderMoments_wc}
								\\
				\limsup_{i \rightarrow \infty} \E \| \w_{e,i} \|^4
						&\le 	O(\mu_{\max}^4)
				\label{Equ:Lemma:AsympBound_4thOrderMoments_we}
			\end{align}
	\end{lemma}
	\begin{proof}
		See Appendix \ref{Appendix_SketchProof_4thOrderMoment}.
	\end{proof}

Therefore, taking $\limsup$ of both sides of inequality recursion \eqref{Equ:SteadyState:Recursion_wcwa_interm2}, we obtain
	\begin{align}
		\limsup_{i\rightarrow \infty} &\E\|\check{\w}_{c,i} \!-\! \check{\w}_{a,i}\|^2
						\nn\\
				&\le 	\gamma_c 
						\cdot
						\limsup_{i \rightarrow \infty}
						\E\|\check{\w}_{c,i-1} \!-\! \check{\w}_{a,i-1}\|^2
						+
						O(\mu_{\max}^3)
						\nn\\
				&=
					\gamma_c 
						\cdot
						\limsup_{i \rightarrow \infty}
						\E\|\check{\w}_{c,i} \!-\! \check{\w}_{a,i}\|^2
						+
						O(\mu_{\max}^3)
		\label{Equ:SteadyState:Recursion_wcwa_final}
	\end{align}
As long as $\gamma_c<1$, which is guaranteed by the stability condition \eqref{P1-Equ:Thm_NonAsympBound:StepSize} from Part I \cite{chen2013learningPart1}, inequality \eqref{Equ:SteadyState:Recursion_wcwa_final} leads to
	\begin{align}
		&\limsup_{i \rightarrow \infty} \E\|\check{\w}_{c,i} - \check{\w}_{a,i}\|^2
						\nn\\
				&\le 	
						\frac{1}{1-\gamma_c} 
						\cdot
						O(\mu_{\max}^3)
						\nn\\
				&=
						\frac{1}{\mu_{\max} - \frac{1}{2} \mu_{\max}^2\|p\|_1^2 \lambda_U^2}
						\cdot
						O(\mu_{\max}^3)
				=		O(\mu_{\max}^2)
		\label{Equ:SteadyState:limsup_wawc_mu2}
	\end{align}
Based on \eqref{Equ:SteadyState:limsup_wawc_mu2}, we can now show that the steady-state covariance matrix of $\check{\w}_{c,i}$ is equal to the covariance matrix of $\check{\w}_{a,i}$ plus a high order perturbation term. First, we have
	\begin{align}
		\check{\Pi}_{c,i}	
				&=		\E[ \check{\w}_{c,i} \check{\w}_{c,i}^T ]
						\nn\\
				&=		\E[ (\check{\w}_{a,i} + \check{\w}_{c,i} - \check{\w}_{a,i})
							(\check{\w}_{a,i} + \check{\w}_{c,i} - \check{\w}_{a,i})^T ]
						\nn\\
				&=		\E[ \check{\w}_{a,i} \check{\w}_{a,i}^T ]
						+
						\E[ \check{\w}_{a,i} (\check{\w}_{c,i}-\check{\w}_{a,i})^T ]
						\nn\\
						&\quad
						+\!
						\E[ (\check{\w}_{c,i} \!-\! \check{\w}_{a,i}) \check{\w}_{a,i}^T ]
						\!+\!
						\E[ (\check{\w}_{c,i} \!-\! \check{\w}_{a,i})
							(\check{\w}_{c,i} \!-\! \check{\w}_{a,i})^T ]
						\nn\\
				&=		\check{\Pi}_{a,i}
						+
						\E[ \check{\w}_{c,i} (\check{\w}_{c,i}-\check{\w}_{a,i})^T ]
						+
						\E[ (\check{\w}_{c,i}-\check{\w}_{a,i}) \check{\w}_{c,i}^T ]
						\nn\\
						&\quad
						-
						\E[ (\check{\w}_{c,i}-\check{\w}_{a,i})
							(\check{\w}_{c,i}-\check{\w}_{a,i})^T ]
		\label{Equ:SteadyState:Pici_check_Approx_interm1}
	\end{align}
The second to the fourth terms in \eqref{Equ:SteadyState:Pici_check_Approx_interm1} are asymptotically high order terms of $\mu_{\max}$. Indeed, for the second term, we have as $i\rightarrow\infty$:
	\begin{align}
		\limsup_{i \rightarrow \infty}&
		\left\|
			\E[ \check{\w}_{c,i} (\check{\w}_{c,i}-\check{\w}_{a,i})^T ]
		\right\|
						\nn\\
				&\le 	
						\limsup_{i \rightarrow \infty}
						\E\left\|\check{\w}_{c,i} (\check{\w}_{c,i}-\check{\w}_{a,i})^T \right\|
						\nn\\
				&\le 	
						\limsup_{i \rightarrow \infty}
						\E[ \|\check{\w}_{c,i}\| \cdot \|\check{\w}_{c,i}-\check{\w}_{a,i}\| ]
						\nn\\
				&\le 	
						\limsup_{i \rightarrow \infty}
						\sqrt{
							\E\|\check{\w}_{c,i}\|^2 
							\cdot 
							\E\|\check{\w}_{c,i}-\check{\w}_{a,i}\|^2
						}
						\nn\\
				&\le		O(\mu_{\max}^{3/2})
		\label{Equ:SteadyState:Pici_check_2ndPerturb}
	\end{align}
Likewise, the third term in \eqref{Equ:SteadyState:Pici_check_Approx_interm1} is asymptotically $O(\mu_{\max}^{3/2})$. For the fourth term in \eqref{Equ:SteadyState:Pici_check_Approx_interm1}, we have as $i\rightarrow\infty$:
	\begin{align}
		\limsup_{i \rightarrow \infty}&
		\left\|
			\E[ (\check{\w}_{c,i}-\check{\w}_{a,i})
							(\check{\w}_{c,i}-\check{\w}_{a,i})^T ]
		\right\|
						\nn\\
				&\overset{(a)}{\le}
						\limsup_{i \rightarrow \infty}
						\E \left\|
							(\check{\w}_{c,i}-\check{\w}_{a,i})
							(\check{\w}_{c,i}-\check{\w}_{a,i})^T 
						\right\|
						\nn\\
				&\overset{(b)}{\le}
						\limsup_{i \rightarrow \infty}
						\E \left\|
							\check{\w}_{c,i}-\check{\w}_{a,i}
						\right\|^2
						\nn\\
				&\overset{(c)}{\le}
						O(\mu_{\max}^2)
		\label{Equ:SteadyState:Pici_check_4thPerturb}
	\end{align}
where step (a) applies Jensen's inequality to the convex function $\|\cdot\|$, step (b) uses the relation $\| x y^T \| \le \|x\|\cdot \|y\|$, and step (c) uses \eqref{Equ:SteadyState:limsup_wawc_mu2}.
Substituting \eqref{Equ:SteadyState:Pici_check_2ndPerturb}--\eqref{Equ:SteadyState:Pici_check_4thPerturb} into \eqref{Equ:SteadyState:Pici_check_Approx_interm1}, we get, 
	\begin{align}
		\limsup_{i \rightarrow \infty }
		\left\|
			\check{\Pi}_{c,i}	 - \check{\Pi}_{a,i}
		\right\|
				&\le 
						O(\mu_{\max}^{3/2})
		\label{Equ:SteadyState:Pici_Pai_gap_limsupBound}
	\end{align}
Combining \eqref{Equ:SteadyState:Pici_Pai_gap_limsupBound} with \eqref{Equ:SteadyState:limsup_Pii_Picicheck} we therefore find that
	\begin{align}
		&\limsup_{i \rightarrow \infty} \| \Pi_i - \one \one^T \otimes \check{\Pi}_{a,i} \|
					\nn\\
				&=
					\limsup_{i \rightarrow \infty} 
					\| 
						\Pi_i 
						- 
						\one \one^T \otimes \check{\Pi}_{c,i}
						+
						\one \one^T \otimes 
						(\check{\Pi}_{c,i} - \check{\Pi}_{a,i})
					\|	
					\nn\\
				&\le
					\limsup_{i \rightarrow \infty} 
					\| 
						\Pi_i 
						\!-\!
						\one \one^T \otimes \check{\Pi}_{c,i}
					\|
						\!+\!
					\limsup_{i \rightarrow \infty}
					\|
						\one \one^T \otimes
						(\check{\Pi}_{c,i} \!-\! \check{\Pi}_{a,i})
					\|	
					\nn\\
				&\overset{(a)}{=}
					\limsup_{i \rightarrow \infty} 
					\| 
						\Pi_i 
						\!-\!
						\one \one^T \otimes \check{\Pi}_{c,i}
					\|
						\!+\!
					\limsup_{i \rightarrow \infty}
					\|\one \one^T\| 
					\!\cdot\!
					\|
						\check{\Pi}_{c,i} \!-\! \check{\Pi}_{a,i}
					\|	
					\nn\\
				&\le 
					O(\mu_{\max}^{3/2})
		\label{Equ:SteadyState:limsup_Pii_Piaicheck}
	\end{align}
where step (a) uses the fact that the $2-$induced matrix norm is the largest singular value and that the singular values of $X \otimes Y$ are equal to the products of the respective singular values of $X$ and $Y$.

\subsection{Evaluation of $\check{\Pi}_{a,\infty}$}
\label{Appendix:Proof_Thm_SteadyStatePerf:PiaExpr}
We now proceed to evaluate $\check{\Pi}_{a,i}$ from recursion \eqref{Equ:SteadyState:wacheck_recursion}:
	\begin{align}
		\check{\Pi}_{a,i}	&=	B_c \check{\Pi}_{a,i-1} B_c^T
								+
								\mu_{\max}^2 
								(p^T\otimes I_M)\E\mc{R}_{v,i}(\bm{\phi}_{i-1}) (p \otimes I_M)
								\nonumber\\
							&=	B_c \check{\Pi}_{a,i-1} B_c^T
								\!+\!
								\mu_{\max}^2 
								(p^T\otimes I_M)\E\mc{R}_{v,i}(\one \!\otimes\! w^o) (p \!\otimes\! I_M)
								\nonumber\\
								&\quad+\!
								\mu_{\max}^2 (p^T \!\otimes\! I_M)
								\E\big[
									\mc{R}_{v,i}(\bm{\phi}_{i-1}) 
									\!-\!
									\mc{R}_{v,i}(\one \!\otimes\! w^o) 
								\big]
								(p\! \otimes \! I_M \!)
		\label{Equ:Pi_a_LyapunotvDiscreteTime_interm1}
	\end{align}
We will verify that the last perturbation term in \eqref{Equ:Pi_a_LyapunotvDiscreteTime_interm1}
is also a high-order term in $\mu_{\max}$.  First note that 
	\begin{align}
		\big\| &\mu_{\max}^2 (p^T\otimes I_M)
								\E\big[
									\mc{R}_{v,i}(\bm{\phi}_{i-1}) 
									-
									\mc{R}_{v,i}(\one \otimes w^o) 
								\big]
								(p \otimes I_M)
		\big\|
								\nn\\
				&\le 				\mu_{\max}^2 
								\cdot
								\| p \|^2
								\cdot
									\E\big\|
										\mc{R}_{v,i}(\bm{\phi}_{i-1}) 
										-
										\mc{R}_{v,i}(\one \otimes w^o) 
									\big\|\!\!
		\label{Equ:SteadyState:NoiseCovPerturb_interm1}
	\end{align}
Next, we bound the rightmost term inside the expectation of \eqref{Equ:SteadyState:NoiseCovPerturb_interm1}. We also need to bound it in two separate cases before arriving at a universal bound:
	\begin{enumerate}
		\item
		{\bf {Case 1:}}  $\|\tilde{\bm{\phi}}_{i-1}\| \le r_V$ \\
		By \eqref{Equ:Assumption:RvLipschitz} in Assumption \ref{Assumption:Rv}, we have
			\begin{align}
				\big\|&
					\mc{R}_{v,i}(\bm{\phi}_{i-1}) 
					-
					\mc{R}_{v,i}(\one \otimes w^o) 
				\big\|	
								\nn\\			
					&\le 			\lambda_v \cdot
								\|
									\bm{\phi}_{i-1}
									-
									\one \otimes w^o
								\|^\kappa
					=			\lambda_v \cdot
								\| \tilde{\bm{\phi}}_{i-1} \|^\kappa
				\label{Equ:SteadyState:NoiseCovPerturb_Case1_final}
			\end{align}
		
		\item
		{\bf {Case 2:}} $\|\tilde{\bm{\phi}}_{i-1}\| > r_V$ \\
		In this case, we have
			\begin{align}
				\big\|&
					\mR_{v,i}(\bm{\phi}_{i-1})
					-
					\mR_{v,i}(\one \otimes w^o)
				\big\|
							\nn\\
					&\le 		\| \mR_{v,i}(\bm{\phi}_{i-1}) \|
							+
							\| \mR_{v,i}(\one \otimes w^o) \|
				\label{Equ:Appendix:Rvi_diff_bound_interm1}
			\end{align}
		To proceed, we first bound $\| \mR_{v,i}(\w)\|$ as follows, where $\w \defeq \col\{ \w_1, \ldots, \w_N \}$. From the definition of $\mR_{v,i}(\w)$ in \eqref{Equ:Assumption:Rv:Rvi_def}, we have
			\begin{align}
				\| \mR_{v,i}(\w) \|
					&\overset{(a)}{\le}
					 			\Tr[ \mR_{v,i}(\w) ]
								\nn\\
					&=			\Tr\big[
									\E\{
										\bv_i(\w) \bv_i^T(\w)
										|
										\mF_{i-1}
									\}
								\big]
								\nn\\
					&=			\E\{
									\Tr[\bv_i(\w) \bv_i^T(\w)]
									|
									\mF_{i-1}
								\}
								\nn\\
					&=			\E\{
									\|\bv_i(\w)\|^2
									|
									\mF_{i-1}
								\}
								\nn\\
					&\overset{(b)}{=}			
								\sum_{k=1}^N
								\E\{
									\|\bv_{k,i}(\w_k)\|^2
									|
									\mF_{i-1}
								\}
								\nn\\
					&\overset{(c)}{\le}
								\sum_{k=1}^N
								\{
									\alpha 
									\cdot
									\| \w_k\|^2
									+
									\sigma_v^2
									|
									\mF_{i-1}
								\}
								\nn\\
					&=			
								\sum_{k=1}^N
								\{
									\alpha 
									\cdot
									\| \w_k - w^o + w^o \|^2
									+
									\sigma_v^2
									|
									\mF_{i-1}
								\}
								\nn\\
					&\le		
								\sum_{k=1}^N
								\{
									2\alpha \| \w_k \!-\! w^o \|^2 
									\!+ \!
									2\alpha \| w^o \|^2
									\!+\!
									\sigma_v^2
									|
									\mF_{i-1}
								\}
								\nn\\
					&=
								2\alpha
								\cdot
								\| \w \!-\! \one \otimes w^o \|^2
								\!+\!
								2\alpha N
								\| w^o \|^2
								\!+\!
								N \sigma_v^2	
				\label{Equ:SteadyState:NoiseCovPerturb_interm2}						
			\end{align}
		where in step (a) we used $\|X\| \le \Tr(X)$ for any symmetric positive semi-definite matrix $X$,
		in step (b) we used the definition of $\bv_i(\w)$ in \eqref{Equ:Assumption:Rv:vi_def}, and
		in step (c) we used \eqref{Equ:Assumption:Randomness:RelAbsNoise}. Using
		\eqref{Equ:SteadyState:NoiseCovPerturb_interm2} with $\w=\bm{\phi}_{i-1}$ 
		and $\w=\one\otimes w^o$, respectively, for the two terms on the right-hand side of
		\eqref{Equ:Appendix:Rvi_diff_bound_interm1}, we get
			\begin{align}
				\big\|&
					\mR_{v,i}(\bm{\phi}_{i-1})
					-
					\mR_{v,i}(\one \otimes w^o)
				\big\|
							\nn\\
					&\le 		
							2\alpha
							\cdot
							\| \tilde{\bm{\phi}}_{i-1} \|^2
							+
							4\alpha N
							\| w^o \|^2
							+
							2N \sigma_v^2	
							\nn\\
					&\overset{(a)}{\le}
							2\alpha
							\cdot
							\| \tilde{\bm{\phi}}_{i-1} \|^2
							+
							\big(
								4\alpha N
								\| w^o \|^2
								+
								2N \sigma_v^2
							\big)
							\cdot
							\frac{\| \tilde{\bm{\phi}}_{i-1} \|^2 }{r_V^2}
							\nn\\
					&=		\big(
								2\alpha						
								+
								\frac{
									4\alpha N
									\| w^o \|^2
									+
									2N \sigma_v^2
								}{r_V^2}
							\big)	
							\cdot
							\| \tilde{\bm{\phi}}_{i-1} \|^2						
				\label{Equ:SteadyState:NoiseCovPerturb_Case2_final}
			\end{align}
		where in step (a) we used the fact that $\|\tilde{\bm{\phi}}_{i-1}\| > r_V$ in the current case.
	\end{enumerate}
In summary, from \eqref{Equ:SteadyState:NoiseCovPerturb_Case1_final} and \eqref{Equ:SteadyState:NoiseCovPerturb_Case2_final}, we obtain the following bound that holds in general:
	\begin{align}
		\big\|&
			\mR_{v,i}(\bm{\phi}_{i-1})
			-
			\mR_{v,i}(\one \otimes w^o)
		\big\|
								\nn\\
				&\le 
								\max \left\{ \!
									\lambda_v \!\cdot\! \|\tilde{\bm{\phi}}_{i-1} \|^\kappa,
									\left(
										2\alpha						
										\!+\!
										\frac{
											4\alpha N
											\| w^o \|^2
											\!+\!
											2N \sigma_v^2
										}{r_V^2}
									\right)	
									\!\cdot\!
									\| \tilde{\bm{\phi}}_{i-1} \|^2
									\!
								\right\}
								\nn\\
				&\le				\lambda_{VU}
								\cdot
								\max \left\{
									\| \tilde{\bm{\phi}}_{i-1} \|^2,
									\| \tilde{\bm{\phi}}_{i-1} \|^\kappa
								\right\}
								\nn\\
				&\le 
								\lambda_{VU}
								\cdot
								\left\{
									\| \tilde{\bm{\phi}}_{i-1} \|^2
									+
									\| \tilde{\bm{\phi}}_{i-1} \|^\kappa
								\right\}
		\label{Equ:SteadyState:NoiseCovPerturb_universal_final}
	\end{align}
where 
	\begin{align}
		\lambda_{VU}		\defeq	\max\big\{
									\lambda_v,
									\;
									2\alpha								
									+
									\frac{
										4\alpha N
										\| w^o \|^2
										+
										2N \sigma_v^2
									}{r_V^2}
								\big\}
								\nn
	\end{align}
Substituting \eqref{Equ:SteadyState:NoiseCovPerturb_universal_final} into \eqref{Equ:SteadyState:NoiseCovPerturb_interm1}, we arrive at
	\begin{align}
		&\limsup_{i \rightarrow \infty}
		\Big\| \mu_{\max}^2  (p^T\otimes I_M)
								\E\big[
									\mc{R}_{v,i}(\bm{\phi}_{i-1}) 
									\nn\\
									&\quad
									-
									\mc{R}_{v,i}(\one \otimes w^o) 
								\big]
								(p \otimes I_M)
		\Big\|
								\nn\\
				&\le				
								\limsup_{i \rightarrow \infty}
								\mu_{\max}^2 
								\cdot
								\| p \|^2
								\cdot
								\lambda_{VU}
								\cdot
								\left[
									\E \| \tilde{\bm{\phi}}_{i-1} \|^2
									+
									\E \| \tilde{\bm{\phi}}_{i-1} \|^\kappa
								\right]
								\nn\\
				&\overset{(a)}{=}
								\limsup_{i \rightarrow \infty}
								\mu_{\max}^2 
								\cdot
								\| p \|^2
								\cdot
								\lambda_{VU}
								\cdot
								\big[
									\E \| \mA_1^T \tilde{\w}_{i-1} \|^2 
									\nn\\
									&\quad
									+ 
									\E \| \mA_1^T \tilde{\w}_{i-1} \|^\kappa
								\big]
								\nn\\
				&\le 				
								\limsup_{i \rightarrow \infty}
								\mu_{\max}^2 
								\cdot
								\| p \|^2
								\cdot
								\lambda_{VU}
								\cdot
								\big[
									\| \mA_1^T \|^2
									\cdot
									\E \|  \tilde{\w}_{i-1} \|^2
									\nn\\
									&\quad
									+
									\| \mA_1^T \|^\kappa
									\cdot
									\E \|  \tilde{\w}_{i-1} \|^\kappa
								\big]
								\nn\\
				&=
								\limsup_{i \rightarrow \infty}
								\mu_{\max}^2 
								\cdot
								\| p \|^2
								\cdot
								\lambda_{VU}
								\cdot
								\Big[
									\| \mA_1^T \|^2
									\cdot
									\E \|  \tilde{\w}_{i-1} \|^2
									\nn\\
									&\quad
									+
									\| \mA_1^T \|^\kappa
									\cdot
									\E \big\{  (\|  \tilde{\w}_{i-1} \|^4)^{\kappa/4} \big\}
								\Big]
								\nn\\
				&\overset{(b)}{\le}
								\limsup_{i \rightarrow \infty}
								\mu_{\max}^2 
								\cdot
								\| p \|^2
								\cdot
								\lambda_{VU}
								\cdot
								\Big[
									\| \mA_1^T \|^2
									\cdot
									\E \|  \tilde{\w}_{i-1} \|^2
									\nn\\
									&\quad
									+
									\| \mA_1^T \|^\kappa
									\cdot
									( \E \|  \tilde{\w}_{i-1} \|^4)^{\kappa/4}
								\Big]			
								\nn\\		
				&\overset{(c)}{\le}
				 				\mu_{\max}^2 \cdot [ O(\mu_{\max}) + O(\mu_{\max}^{\kappa/2}) ]	
								\nn\\
				&=				O(\mu_{\max}^3) + O(\mu_{\max}^{\kappa/2+2})
		\label{Equ:SteadyState:NoiseCovPerturb_finalfinal}
	\end{align}
where in step (a) we used the relation $\tilde{\bm{\phi}}_{i-1} = \mA_1^T \tilde{\w}_{i-1}$ from \eqref{P1-Equ:DistProc:relation_phi_w_wprime} in Part I\cite{chen2013learningPart1}, in step (b) we applied Jensen's inequality $\E(\bm{x}^{\kappa/4}) \le (\E \bm{x})^{\kappa/4}$ since $x^{\kappa/4}$ is a concave function when $0<\kappa \le 4$, and in step (c) we used the fact that the $\limsup$ of $\E\|\tilde{\bm{w}}_{i-1}\|^2$ is on the order of $O(\mu_{\max})$\footnote{This can be derived by using \eqref{Equ:DistProc:wi_tilde_decomposition_globalform_final}, \eqref{Equ:Cor:EPwcicheck_asymptotic_bound}, \eqref{Equ:Cor:EPwei_asymptotic_bound} along with the fact that $\tilde{w}_{c,i} \rightarrow 0$ (Thm. \ref{P1-Thm:ConvergenceRefRec:DeterministcCent} in Part I\cite{chen2013learningPart1}) and $\|x+y+z\|^2 \le 3(\|x\|^2+\|y\|^2+\|z\|^2)$.} and that the $\limsup$ of $\E\|\tilde{\bm{w}}_{i-1}\|^4$ is on the order of $O(\mu_{\max}^2)$\footnote{This can be derived by using \eqref{Equ:DistProc:wi_tilde_decomposition_globalform_final}, \eqref{Equ:Lemma:AsympBound_4thOrderMoments_wc}, \eqref{Equ:Lemma:AsympBound_4thOrderMoments_we} along with the fact that $\tilde{w}_{c,i} \rightarrow 0$ and $\|x+y+z\|^4 \le 27(\|x\|^4+\|y\|^4+\|z\|^4)$.}. The bound \eqref{Equ:SteadyState:NoiseCovPerturb_finalfinal} implies that recursion \eqref{Equ:Pi_a_LyapunotvDiscreteTime_interm1} is a perturbed version of the following recursion
	\begin{align}
		\check{\Pi}_{a,i}^o	&=	B_c \check{\Pi}_{a,i-1}^o B_c^T
								\!+\!
								\mu_{\max}^2 
								(p^T \!\otimes\! I_M)
								\E\mc{R}_{v,i}(\one \!\otimes\! w^o) 
								(p \!\otimes\! I_M)
		\label{Equ:Pi_a_LyapunotvDiscreteTime}
	\end{align}
We now show that the covariance matrices obtained from these two
recursions are close to each other in the sense that
	\begin{align}
		\limsup_{i \rightarrow \infty} 
		\|\check{\Pi}_{a,i}	-	\check{\Pi}_{a,i}^o\| \le O\big(\mu_{\max}^{\min(2, 1+\kappa/2)} \big)
		\label{Equ:SteadyState:Pi_ai_perturbationRelation}
	\end{align}
Subtracting \eqref{Equ:Pi_a_LyapunotvDiscreteTime} from \eqref{Equ:Pi_a_LyapunotvDiscreteTime_interm1},
we get
	\begin{align}
		\check{\Pi}_{a,i}-&\check{\Pi}_{a,i}^o	
							=	B_c (\check{\Pi}_{a,i-1} - \check{\Pi}_{a,i-1}^o) B_c^T
								\nn\\
								&
								+
								\mu_{\max}^2 (p^T \!\otimes\! I_M)
								\E\big[
									\mc{R}_{v,i}(\bm{\phi}_{i-1}) 
									\!-\!
									\mc{R}_{v,i}(\one \!\otimes\! w^o) 
								\big]
								(p\! \otimes \! I_M \!)
								\nn
	\end{align}
Taking the $2$-induced norm of both sides, we get
	\begin{align}
		\big\|&\check{\Pi}_{a,i}-\check{\Pi}_{a,i}^o	\big\|
									\nn\\
							&\le 	
									\|B_c\|^2 
									\cdot 
									\left\|
										\check{\Pi}_{a,i-1} - \check{\Pi}_{a,i-1}^o
									\right\|
									\nn\\
									&\quad
									+
									\left\|
										\mu_{\max}^2 (p^T \!\otimes\! I_M)
										\E\big[
											\mc{R}_{v,i}(\bm{\phi}_{i-1}) 
											\!-\!
											\mc{R}_{v,i}(\one \!\otimes\! w^o) 
										\big]
										(p\! \otimes \! I_M \!)
									\right\|
									\nn\\
							&\overset{(a)}{\le}
									\gamma_c^2 
									\cdot 
									\left\|
										\check{\Pi}_{a,i-1} - \check{\Pi}_{a,i-1}^o
									\right\|
									\nn\\
									&\quad
									+
									\left\|
										\mu_{\max}^2 (p^T \!\otimes\! I_M)
										\E\big[
											\mc{R}_{v,i}(\bm{\phi}_{i-1}) 
											\!-\!
											\mc{R}_{v,i}(\one \!\otimes\! w^o) 
										\big]
										(p\! \otimes \! I_M \!)
									\right\|
									\nn
	\end{align}
where in step (a) we are using \eqref{Equ:Appendix:Bc_norm_UB}. Taking $\limsup$ of both sides the above inequality, we obtain
	\begin{align}
		&\limsup_{i \rightarrow \infty} \big\|\check{\Pi}_{a,i}-\check{\Pi}_{a,i}^o\big\|
									\nn\\
							&\le
									\gamma_c^2 
									\cdot 
									\limsup_{i \rightarrow \infty}
									\left\|
										\check{\Pi}_{a,i-1} - \check{\Pi}_{a,i-1}^o
									\right\|
									\nn\\
									&\quad
									+
									\limsup_{i \rightarrow \infty}
									\big\|
										\mu_{\max}^2 (p^T \!\otimes\! I_M)
										\E\big[
											\mc{R}_{v,i}(\bm{\phi}_{i-1}) 
											-
											\mc{R}_{v,i}(\one \!\otimes\! w^o) 
										\big]
											\nn\\
											&\qquad\qquad\qquad
										\cdot (p\! \otimes \! I_M \!)
									\big\|
									\nn\\
							&\overset{(a)}{\le}
								 	\gamma_c^2 
									\cdot 
									\limsup_{i \rightarrow \infty}
									\left\|
										\check{\Pi}_{a,i-1} - \check{\Pi}_{a,i-1}^o
									\right\|
									+
									O(\mu_{\max}^3) + O(\mu_{\max}^{\kappa/2+2})
									\nn\\
							&=
									\gamma_c^2 
									\cdot 
									\limsup_{i \rightarrow \infty}
									\left\|
										\check{\Pi}_{a,i} \!-\! \check{\Pi}_{a,i}^o
									\right\|
									\!+\!
									O(\mu_{\max}^3) \!+\! O(\mu_{\max}^{\kappa/2+2})
		\label{Equ:SteadyState:Pi_ai_UB_interm0}
	\end{align}
where step (a) uses \eqref{Equ:SteadyState:NoiseCovPerturb_finalfinal}. Recalling that $\gamma_c < 1$, which is already guaranteed by choosing $\mu_{\max}$ according to the stability condition \eqref{P1-Equ:Thm_NonAsympBound:StepSize} in Part I\cite{chen2013learningPart1}, we can move the first term on the right-hand side of \eqref{Equ:SteadyState:Pi_ai_UB_interm0} to the left, divide both sides by $1-\gamma_c^2$ and get
	\begin{align}
		\limsup_{i \rightarrow \infty} \big\|\check{\Pi}_{a,i}-\check{\Pi}_{a,i}^o\big\|
							&\le 									
									\frac{O(\mu_{\max}^3) + O(\mu_{\max}^{\kappa/2+2})}
									{1-\gamma_c^2}
									\nn\\
							&\le 
									\frac{O(\mu_{\max}^3) + O(\mu_{\max}^{\kappa/2+2})}
									{1-\gamma_c}
									\nn\\
							&\overset{(a)}{=}
									\frac{O(\mu_{\max}^3) + O(\mu_{\max}^{\kappa/2+2})}
									{\mu_{\max} \lambda_L 
									- \frac{1}{2}\mu_{\max}^2 \|p\|_1 \lambda_U^2}
									\nonumber\\
							&=		
									\frac{O(\mu_{\max}^2) + O(\mu_{\max}^{\kappa/2+1})}
									{\lambda_L - \frac{1}{2}\mu_{\max} \|p\|_1 \lambda_U^2}
									\nn\\
							&=
									\frac{O\big(\mu_{\max}^{\min(2, 1+\kappa/2)}\big)}
									{\lambda_L - \frac{1}{2}\mu_{\max} \|p\|_1 \lambda_U^2}
		\label{Equ:SteadyState:Pi_ai_UB_interm1} 
	\end{align}
where in step (a) we are substituting \eqref{Equ:VarPropt:gamma_c}.
%where in step (a) we are using \eqref{Equ:Appendix:Bc_norm_UB} and in step (b) we are using the fact that $\gamma_c < 1$, which is already guaranteed by choosing $\mu_{\max}$ according to the stability condition \eqref{P1-Equ:Thm_NonAsympBound:StepSize} in Part I\cite{chen2013learningPart1}. Taking $\limsup$ of both sides of  \eqref{Equ:SteadyState:Pi_ai_UB_interm1}, we arrive at
%\eqref{Equ:SteadyState:Pi_ai_perturbationRelation}. 

\subsection{Final expression for $\Pi_{\infty}$}
\label{Appendix:Proof_Thm_SteadyStatePerf:PiFinalExpr}
Therefore, by \eqref{Equ:SteadyState:limsup_Pii_Piaicheck} and \eqref{Equ:SteadyState:Pi_ai_perturbationRelation}, we have
	\begin{align}
		&\limsup_{i \rightarrow \infty}
		\| \Pi_i - \one\one^T \otimes \check{\Pi}_{a,i}^o \|
					\nn\\
			&=
					\limsup_{i \rightarrow \infty}
					\| 
						\Pi_i - \one\one^T \otimes \check{\Pi}_{a,i} 
						+ 
						\one\one^T
						\otimes
						(\check{\Pi}_{a,i} - \check{\Pi}_{a,i}^o) 
					\|
					\nn\\
			&\le
					\limsup_{i \rightarrow \infty}
					\| \Pi_i - \one\one^T \otimes \check{\Pi}_{a,i}  \|
					\nn\\
					&\quad
					+ 
					\limsup_{i \rightarrow \infty}
					\|
						\one\one^T
						\otimes
						(\check{\Pi}_{a,i} - \check{\Pi}_{a,i}^o) 
					\|
					\nn\\
			&=
					\limsup_{i \rightarrow \infty}
					\| \Pi_i - \one\one^T \otimes \check{\Pi}_{a,i}  \|
					\nn\\
					&\quad
					+ 
					\limsup_{i \rightarrow \infty}
					\|
						\one\one^T
					\|
					\cdot
					\|
						\check{\Pi}_{a,i} - \check{\Pi}_{a,i}^o
					\|
					\nn\\
			&\le
					O(\mu_{\max}^{3/2}) + O(\mu_{\max}^{\min(2, 1+\kappa/2)})
					\nn\\
			&=
					O(\mu_{\max}^{\min(3/2, 1+\kappa/2)})
		\label{Equ:SteadyState:limsup_Picicheck_Phiaiocheck}
	\end{align}
%which further implies
%	\begin{align}
%		&\limsup_{i \rightarrow \infty}
%		\| \Pi_i - \one\one^T \otimes \check{\Pi}_{a,\infty}^o \|
%					\nn\\
%			&=
%					\limsup_{i \rightarrow \infty}
%					\| 
%						\Pi_i - \one\one^T \otimes \check{\Pi}_{a,i}^o
%						+ 
%						\one\one^T
%						\otimes
%						(\check{\Pi}_{a,i}^o - \check{\Pi}_{a,\infty}^o) 
%					\|
%					\nn\\
%			&\le
%					\limsup_{i \rightarrow \infty}
%					\| \Pi_i - \one\one^T \otimes \check{\Pi}_{a,i}^o  \|
%					\nn\\
%					&\quad
%					+ 
%					\limsup_{i \rightarrow \infty}
%					\|
%						\one\one^T
%						\otimes
%						(\check{\Pi}_{a,i}^o - \check{\Pi}_{a,\infty}^o) 
%					\|
%					\nn\\
%			&=
%					\limsup_{i \rightarrow \infty}
%					\| \Pi_i - \one\one^T \otimes \check{\Pi}_{a,i}^o  \|
%					\nn\\
%					&\quad
%					+ 
%					\limsup_{i \rightarrow \infty}
%					\|
%						\one\one^T
%					\|
%					\cdot
%					\|
%						\check{\Pi}_{a,i}^o - \check{\Pi}_{a,\infty}^o
%					\|
%					\nn\\
%			&\overset{(a)}{\le}
%					O(\mu_{\max}^{\min(3/2, 1+\kappa/2)}) + 0
%					\nn\\
%			&=
%					O(\mu_{\max}^{\min(3/2, 1+\kappa/2)})
%		\label{Equ:SteadyState:limsup_Picicheck_Phiaiocheck}
%	\end{align}
%where step (a) uses the fact that $\check{\Pi}_{a,i}^o$ converges to $\check{\Pi}_{a,\infty}^o$ since
As $i\rightarrow\infty$, the unperturbed recursion \eqref{Equ:Pi_a_LyapunotvDiscreteTime} 
converges to a unique solution $\check{\Pi}_{a,\infty}^o$ that satisfies
the following discrete Lyapunov equation:
	\begin{align}
		\check{\Pi}_{a,\infty}^o	=	B_c \check{\Pi}_{a,\infty}^o B_c^T
									+
									\mu_{\max}^2 (p^T\otimes I_M) \cdot \mc{R}_v \cdot (p \otimes I_M)
		\label{Equ:SteadyState:Pi_c_LyapunovDiscreteTime_final}
	\end{align}	 
where we used \eqref{Equ:Assumption:Rv:R_v_limit} from Assumption \ref{Assumption:Rv}.\footnote{The almost sure convergence in \eqref{Equ:Assumption:Rv:R_v_limit} implies $\E \mc{R}_{v,i}(\one \otimes w^o) \rightarrow \mc{R}_{v}$ in \eqref{Equ:Pi_a_LyapunotvDiscreteTime} and \eqref{Equ:SteadyState:Pi_c_LyapunovDiscreteTime_final} because of the dominated convergence theorem\cite[p.44]{chung2001course}. The condition of dominated convergence theorem can be verified by showing that $\|\mc{R}_{v,i}(\one \otimes w^o)\|$ is upper bounded almost surely by a deterministic constant $2\alpha N \| w^o \|^2 + N \sigma_v^2$, which can be proved by following a similar line of argument in \eqref{Equ:SteadyState:NoiseCovPerturb_interm2} using \eqref{Equ:Assumption:Randomness:RelAbsNoise}, \eqref{Equ:Assumption:Rv:vi_def}, and \eqref{Equ:Assumption:Rv:Rvi_def}.} In other words, as $i \rightarrow \infty$, $\check{\Pi}_{a,i}^o$ converges to $\check{\Pi}_{a,\infty}^o$ so that
	\begin{align}
		&\limsup_{i \rightarrow \infty}
		\| \Pi_i - \one\one^T \otimes \check{\Pi}_{a,\infty}^o \|
					\nn\\
			&=
					\limsup_{i \rightarrow \infty}
					\| 
						\Pi_i - \one\one^T \otimes \check{\Pi}_{a,i}^o
						+ 
						\one\one^T
						\otimes
						(\check{\Pi}_{a,i}^o - \check{\Pi}_{a,\infty}^o) 
					\|
					\nn\\
			&\le
					\limsup_{i \rightarrow \infty}
					\| \Pi_i - \one\one^T \otimes \check{\Pi}_{a,i}^o  \|
					\nn\\
					&\quad
					+ 
					\limsup_{i \rightarrow \infty}
					\|
						\one\one^T
						\otimes
						(\check{\Pi}_{a,i}^o - \check{\Pi}_{a,\infty}^o) 
					\|
					\nn\\
%			&=
%					\limsup_{i \rightarrow \infty}
%					\| \Pi_i - \one\one^T \otimes \check{\Pi}_{a,i}^o  \|
%					\nn\\
%					&\quad
%					+ 
%					\limsup_{i \rightarrow \infty}
%					\|
%						\one\one^T
%					\|
%					\cdot
%					\|
%						\check{\Pi}_{a,i}^o - \check{\Pi}_{a,\infty}^o
%					\|
%					\nn\\
			&\le
					O(\mu_{\max}^{\min(3/2, 1+\kappa/2)})
		\label{Equ:SteadyState:limsup_Picicheck_Phiaiocheck}
	\end{align}
Furthermore, using \eqref{Equ:SteadyState:WMSE_Tr} and \eqref{Equ:SteadyState:limsup_Picicheck_Phiaiocheck}, we also have
	\begin{align}
		&\limsup_{i \rightarrow \infty}
		\left|
			\E \|\tilde{\w}_{k,i}\|_{\Sigma}^2 
			- 
			\Tr\big\{
				(\one\one^T \otimes \check{\Pi}_{a,\infty}^o)
				(E_{kk} \otimes \Sigma)
			\big\}
		\right|
			\nn\\
			&=	
					\limsup_{i \rightarrow \infty}
					\big|
						\Tr\big\{
							\Pi_i
							(E_{kk} \otimes \Sigma)
						\big\} 
						\nn\\
						&\qquad\qquad
						- 
						\Tr\big\{
							(\one\one^T \otimes \check{\Pi}_{a,\infty}^o)
							(E_{kk} \otimes \Sigma)
						\big\}
					\big|
					\nn\\
			&=
					\limsup_{i \rightarrow \infty}
					\big|
						\Tr\big\{
							(\Pi_i - \one\one^T \otimes \check{\Pi}_{a,\infty}^o)
							(E_{kk} \otimes \Sigma)
						\big\}
					\big|
					\nn\\
			&=
					\limsup_{i \rightarrow \infty}
					\big|
						[\mathrm{vec}(\Pi_i - \one\one^T \otimes \check{\Pi}_{a,\infty}^o)]^T
						\mathrm{vec}(E_{kk} \otimes \Sigma)
					\big|
					\nn\\
			&\overset{(a)}{\le}
					\limsup_{i \rightarrow \infty}
					\big\|
						\mathrm{vec}(\Pi_i - \one\one^T \otimes \check{\Pi}_{a,\infty}^o)
					\big\|
					\cdot
					\big\|
						\mathrm{vec}(E_{kk} \otimes \Sigma)
					\big\|
					\nn\\
			&=
					\limsup_{i \rightarrow \infty}
					\big\|
						\Pi_i - \one\one^T \otimes \check{\Pi}_{a,\infty}^o
					\big\|_F
					\cdot
					\big\|
						E_{kk} \otimes \Sigma
					\big\|_F
					\nn\\
			&\overset{(b)}{\le}
					C \cdot 
					\limsup_{i \rightarrow \infty}					
					\big\|
						\Pi_i - \one\one^T \otimes \check{\Pi}_{a,\infty}^o
					\big\|
					\nn\\
			&\le
					O(\mu_{\max}^{\min(3/2, 1+\kappa/2)})
		\label{Equ:SteadyState:limsup_WMSE_Piainfocheck}
	\end{align}
where step (a) uses Cauchy-Schwartz inequality and step (b) uses the equivalence of matrix norms. The bound \eqref{Equ:SteadyState:limsup_WMSE_Piainfocheck} is useful in that it has the following implications about the $\limsup$ and $\liminf$ of the weighted mean-square-error $\E\|\tilde{\w}_{k,i}\|_{\Sigma}^2$:
	\begin{align}
		&\limsup_{i \rightarrow \infty} \E \|\tilde{\w}_{k,i}\|_{\Sigma}^2
					\nn\\
			&\overset{(a)}{=}
					\limsup_{i \rightarrow \infty}
					\Big\{
						\Tr\big\{
							(\one\one^T \otimes \check{\Pi}_{a,\infty}^o)
							(E_{kk} \otimes \Sigma)
						\big\}
						\nn\\
						&\qquad
						+
						\E \|\tilde{\w}_{k,i}\|_{\Sigma}^2 
						- 
						\Tr\big\{
							(\one\one^T \otimes \check{\Pi}_{a,\infty}^o)
							(E_{kk} \otimes \Sigma)
						\big\}
					\Big\}
					\nn\\
			&=
					\Tr\big\{
						(\one\one^T \otimes \check{\Pi}_{a,\infty}^o)
						(E_{kk} \otimes \Sigma)
					\big\}
					\nn\\
					&\quad
					+
					\limsup_{i \rightarrow \infty}
					\Big[
						\E \|\tilde{\w}_{k,i}\|_{\Sigma}^2 
						- 
						\Tr\big\{
							(\one\one^T \otimes \check{\Pi}_{a,\infty}^o)
							(E_{kk} \otimes \Sigma)
						\big\}
					\Big]
					\nn\\
			&\overset{(b)}{\le}
					\Tr\big\{
						(\one\one^T \otimes \check{\Pi}_{a,\infty}^o)
						(E_{kk} \otimes \Sigma)
					\big\}
					\nn\\
					&\quad
					+
					\limsup_{i \rightarrow \infty}
					\Big|
						\E \|\tilde{\w}_{k,i}\|_{\Sigma}^2 
						- 
						\Tr\big\{
							(\one\one^T \otimes \check{\Pi}_{a,\infty}^o)
							(E_{kk} \otimes \Sigma)
						\big\}
					\Big|
					\nn\\
			&\overset{(c)}{\le}
					\Tr\big\{
						(\one\one^T \otimes \check{\Pi}_{a,\infty}^o)
						(E_{kk} \otimes \Sigma)
					\big\}
					+
					O(\mu_{\max}^{\min(3/2, 1+\kappa/2)})
					\nn\\
			&=
					\Tr\big\{
						(\one\one^T E_{kk})
						\otimes 
						(\check{\Pi}_{a,\infty}^o \Sigma)
					\big\}
					+
					O(\mu_{\max}^{\min(3/2, 1+\kappa/2)})
					\nn\\
			&\overset{(d)}{=}
					\Tr(\one\one^T E_{kk})
					\cdot
					\Tr(\check{\Pi}_{a,\infty}^o \Sigma)
					+
					O(\mu_{\max}^{\min(3/2, 1+\kappa/2)})
					\nn\\
			&=
					\Tr(\check{\Pi}_{a,\infty}^o \Sigma)
					+
					O(\mu_{\max}^{\min(3/2, 1+\kappa/2)})
					\nn\\
			&=
					[\mathrm{vec}(\check{\Pi}_{a,\infty}^o)]^T
					\mathrm{vec}(\Sigma)
					+
					O(\mu_{\max}^{\min(3/2, 1+\kappa/2)})
		\label{Equ:SteadyState:limsup_WMSE_intermediate1}
	\end{align}
where step (a) adds and subtracts the same term, step (b) uses $x \le |x|$, step (c) uses \eqref{Equ:SteadyState:limsup_WMSE_Piainfocheck}, and step (d) uses the property $\Tr(X \otimes Y) = \Tr(X) \Tr(Y)$ for Kronecker products \cite[p.142]{laub2005matrix}. Likewise, the $\liminf$ of the weighted MSE can be derived as
	\begin{align}
		&\liminf_{i \rightarrow \infty} \E \|\tilde{\w}_{k,i}\|_{\Sigma}^2
					\nn\\
			&\overset{(a)}{=}
					\liminf_{i \rightarrow \infty}
					\Big\{
						\Tr\big\{
							(\one\one^T \otimes \check{\Pi}_{a,\infty}^o)
							(E_{kk} \otimes \Sigma)
						\big\}
						\nn\\
						&\qquad
						+
						\E \|\tilde{\w}_{k,i}\|_{\Sigma}^2 
						- 
						\Tr\big\{
							(\one\one^T \otimes \check{\Pi}_{a,\infty}^o)
							(E_{kk} \otimes \Sigma)
						\big\}
					\Big\}
					\nn\\
			&=
					\Tr\big\{
						(\one\one^T \otimes \check{\Pi}_{a,\infty}^o)
						(E_{kk} \otimes \Sigma)
					\big\}
					\nn\\
					&\quad
					+
					\liminf_{i \rightarrow \infty}
					\Big[
						\E \|\tilde{\w}_{k,i}\|_{\Sigma}^2 
						- 
						\Tr\big\{
							(\one\one^T \otimes \check{\Pi}_{a,\infty}^o)
							(E_{kk} \otimes \Sigma)
						\big\}
					\Big]
					\nn\\
			&\overset{(b)}{\ge}
					\Tr\big\{
						(\one\one^T \otimes \check{\Pi}_{a,\infty}^o)
						(E_{kk} \otimes \Sigma)
					\big\}
					\nn\\
					&\quad
					+
					\liminf_{i \rightarrow \infty}
					\Big\{
					-
					\Big|
						\E \|\tilde{\w}_{k,i}\|_{\Sigma}^2 
						- 
						\Tr\big\{
							(\one\one^T \otimes \check{\Pi}_{a,\infty}^o)
							(E_{kk} \otimes \Sigma)
						\big\}
					\Big|
					\Big\}
					\nn\\
			&\ge
					\Tr\big\{
						(\one\one^T \otimes \check{\Pi}_{a,\infty}^o)
						(E_{kk} \otimes \Sigma)
					\big\}
					\nn\\
					&\quad
					-
					\limsup_{i \rightarrow \infty}
					\Big|
						\E \|\tilde{\w}_{k,i}\|_{\Sigma}^2 
						- 
						\Tr\big\{
							(\one\one^T \otimes \check{\Pi}_{a,\infty}^o)
							(E_{kk} \otimes \Sigma)
						\big\}
					\Big|
					\nn\\
			&\overset{(c)}{\ge}
					\Tr\big\{
						(\one\one^T \otimes \check{\Pi}_{a,\infty}^o)
						(E_{kk} \otimes \Sigma)
					\big\}
					-
					O(\mu_{\max}^{\min(3/2, 1+\kappa/2)})
					\nn\\
			&=
					\Tr\big\{
						(\one\one^T E_{kk})
						\otimes 
						(\check{\Pi}_{a,\infty}^o \Sigma)
					\big\}
					-
					O(\mu_{\max}^{\min(3/2, 1+\kappa/2)})
					\nn\\
			&\overset{(d)}{=}
					\Tr(\one\one^T E_{kk})
					\cdot
					\Tr(\check{\Pi}_{a,\infty}^o \Sigma)
					-
					O(\mu_{\max}^{\min(3/2, 1+\kappa/2)})
					\nn\\
			&=
					\Tr(\check{\Pi}_{a,\infty}^o \Sigma)
					-
					O(\mu_{\max}^{\min(3/2, 1+\kappa/2)})
					\nn\\
			&=
					[\mathrm{vec}(\check{\Pi}_{a,\infty}^o)]^T
					\mathrm{vec}(\Sigma)
					-
					O(\mu_{\max}^{\min(3/2, 1+\kappa/2)})
		\label{Equ:SteadyState:liminf_WMSE_intermediate1}
	\end{align}
where step (a) adds and subtracts the same term, step (b) uses $x \ge -|x|$, step (c) uses\eqref{Equ:SteadyState:limsup_WMSE_Piainfocheck}, and step (d) uses the property $\Tr(X \otimes Y) = \Tr(X) \Tr(Y)$ for Kronecker products. Note from \eqref{Equ:SteadyState:limsup_WMSE_intermediate1} and \eqref{Equ:SteadyState:liminf_WMSE_intermediate1} that the first terms in the $\limsup$ and $\liminf$ bounds are the same, and the second terms are high-order terms of $\mu_{\max}$. Therefore, once we find the expression for $\check{\Pi}_{a,\infty}^o$, we will have a complete characterization of the steady-state MSE.

Now we proceed to derive the expression for $\check{\Pi}_{a,\infty}^o$. Vectorizing both sides of \eqref{Equ:SteadyState:Pi_c_LyapunovDiscreteTime_final}, we obtain 
	\begin{align}
		\mathrm{vec}&(\check{\Pi}_{a,\infty}^o)
								\nn\\
					&=			\mu_{\max}^2 \cdot(I_{M^2} - B_c \otimes B_c)^{-1}
								\nn\\
								&\quad
								\times
								\mathrm{vec}
								\left\{
									(p^T\otimes I_M)\cdot \mc{R}_v \cdot (p \otimes I_M)
								\right\}
								\nonumber\\
					&=			\mu_{\max} \cdot
								(
									I_M \otimes H_c \!+\! H_c \otimes I_M
									\!-\!
									\mu_{\max} H_c \otimes H_c
								)^{-1}
								\nn\\
								&\quad
								\times
								\mathrm{vec}
								\left\{
									(p^T\otimes I_M) \cdot \mc{R}_v \cdot (p \otimes I_M)
								\right\}
								\nonumber\\
					&\overset{(a)}{=}
								\mu_{\max} \cdot
								(
									I_M \otimes H_c + H_c \otimes I_M
								)^{-1}
								\nn\\
								&\quad
								\times\!\!
								\big[
									I_{M^2}
									\!-\!
									\mu_{\max} 
									(H_c \!\otimes\! H_c)
									(I_M \!\otimes\! H_c \!+\! H_c \!\otimes\! I_M)^{-1}
								\big]^{-1}
								\nonumber\\
								&\quad\times
								\mathrm{vec}
								\left\{
									(p^T\otimes I_M) \cdot \mc{R}_v \cdot (p \otimes I_M)
								\right\}
		\label{Equ:SteadyState:Pi_ci_o_interm111}
	\end{align}
where step (a) uses the fact that $(X + Y)^{-1} = X^{-1} ( I + Y X^{-1} )^{-1}$ given $X$ is invertible. Note that the existence of the inverse of $I_M \otimes H_c + H_c \otimes I_M$ is guaranteed by \eqref{Equ:Lemma:EquivCondUpdateVec:StrongMono_HessianLB} for the following reason. First, condition \eqref{Equ:Lemma:EquivCondUpdateVec:StrongMono_HessianLB} ensures that all the eigenvalues of $H_c$ have positive real parts. To see this, let $\lambda(H_c)$ and $x_0$ ($x_0 \neq 0$) denote an eigenvalue of $H_c$ and the corresponding eigenvector\footnote{Note that the matrix $H_c$ need not be symmetric and hence its eigenvalues and eigenvectors need not be real.}. Then,
	\begin{align}
		H_c x_0 = \lambda(H_c) \cdot x_0
		&\Rightarrow
		x_0^* H_c x_0 		= 	\lambda(H_c) \cdot \| x_0 \|^2
		\label{Equ:Appendix:xHcx_interm1}
								\\
		&\Rightarrow
		(x_0^* H_c x_0)^* 	= 	\lambda^*(H_c) \cdot \| x_0 \|^2
								\nn\\
		&\Rightarrow
		x_0^* H_c^* x_0 	= 	\lambda^*(H_c) \cdot \| x_0 \|^2
								\nn\\
		&\Rightarrow
		x_0^* H_c^T x_0 	= 	\lambda^*(H_c) \cdot \| x_0 \|^2
		\label{Equ:Appendix:xHcx_interm2}
	\end{align}
where $(\cdot)^{*}$ denotes the conjugate transpose operator, and the last step uses the fact that $H_c$ is real so that $H_c^* = H_c^T$. Summing \eqref{Equ:Appendix:xHcx_interm1} and \eqref{Equ:Appendix:xHcx_interm2} leads to
	\begin{align}
		&x_0^* (H_c+H_c)^T x_0 =	2\mathrm{Re}\{\lambda(H_c)\} \cdot \| x_0 \|^2
		\nn\\
		&\quad\Rightarrow\quad
		\mathrm{Re}\{\lambda(H_c)\}
							=	
								\frac{x_0^* (H_c+H_c)^T x_0}{2\|x_0\|^2}
							\ge 
								\lambda_L >0
								\nn
	\end{align}
where the last step uses \eqref{Equ:Lemma:EquivCondUpdateVec:StrongMono_HessianLB}. Furthermore, the $M^2$ eigenvalues of $I_M \otimes H_c + H_c \otimes I_M$ are $\lambda_{m_1}(H_c) + \lambda_{m_2}(H_c)$ for $m_1, m_2 = 1,\ldots,M$, where $\lambda_m(\cdot)$ denotes the $m$th eigenvalue of a matrix\cite[p.143]{laub2005matrix}. Therefore, the real parts of the eigenvalues of $I_M \otimes H_c + H_c \otimes I_M$ are $\mathrm{Re}\left\{\lambda_{m_1}(H_c)\right\} + \mathrm{Re}\left\{\lambda_{m_2}(H_c)\right\} > 0$ so that the matrix $I_M \otimes H_c + H_c \otimes I_M$ is not singular and is invertible. Observing that for any matrix $X$ where the necessary inverse holds, we have 
	\begin{align}
		&(I - \mu_{\max} X)^{-1} (I - \mu_{\max} X) = I 
		\nn\\
		\Leftrightarrow
		&(I - \mu_{\max} X)^{-1}
		- 
		\mu_{\max} 
		(I - \mu_{\max} X)^{-1} X = I
		\nonumber\\
		\Leftrightarrow
		&(I - \mu_{\max} X)^{-1}
		=
		I + 
		\mu_{\max} 
		(I - \mu_{\max} X)^{-1} X
		\nn
	\end{align}
and, hence, 
	\begin{align}
		&\big[
			I_{M^2}
			-
			\mu_{\max} 
			(H_c \otimes H_c)
			(I_M \otimes H_c + H_c \otimes I_M)^{-1}
		\big]^{-1}
		\nonumber\\
		&=
			I_{M^2}
			\!+\!
			\mu_{\max} 
			\big[
				I \!-\! \mu_{\max} (H_c \!\otimes\! H_c)
				(I_M	 \!\otimes\! H_c \!+\! H_c \!\otimes\! I_M)^{-1} 
			\big]^{-1} 
			\nn\\
			&\quad
			\times
			(H_c \otimes H_c)
			(I_M \otimes H_c + H_c \otimes I_M)^{-1}
			\nonumber\\
		&\overset{(a)}{=}
			I_{M^2} + O(\mu_{\max})
		\label{Equ:Appendix:Inverse_Asymp_interm1}
	\end{align}
where step (a) is because 
	\begin{align}
		&\lim_{ \mu_{\max} \rightarrow 0 }
		\frac{1}{\mu_{\max}}
		\nn\\
		&
		\times
		\Big\{
			\mu_{\max} 
			\big[
				I \!-\! \mu_{\max} (H_c \!\otimes\! H_c)
				(I_M	 \!\otimes\! H_c \!+\! H_c \!\otimes\! I_M)^{-1} 
			\big]^{-1} 
			\nn\\
			&\qquad
			(H_c \!\otimes\! H_c)
			(I_M \!\otimes\! H_c \!+\! H_c \!\otimes\! I_M)^{-1}
		\Big\}
								\nn\\
					&=		
								\lim_{ \mu_{\max} \rightarrow 0 }
								\big[
									I - \mu_{\max} (H_c \otimes H_c)
									(I_M	 \otimes H_c + H_c \otimes I_M)^{-1} 
								\big]^{-1} 
								\nn\\
								&\qquad\quad
								\times
								(H_c \otimes H_c)
								(I_M \otimes H_c + H_c \otimes I_M)^{-1}
								\nn\\
					&=
								(H_c \otimes H_c)
								(I_M \otimes H_c + H_c \otimes I_M)^{-1}
								\nn\\
					&=			\mathrm{constant}
								\nn
	\end{align}
Therefore, substituting \eqref{Equ:Appendix:Inverse_Asymp_interm1} into \eqref{Equ:SteadyState:Pi_ci_o_interm111} leads to
	\begin{align}
		\mathrm{vec}(\check{\Pi}_{a,\infty}^o)
					&=			\mu_{\max} \cdot
								\left[
									(
										I_M \otimes H_c + H_c \otimes I_M
									)^{-1}
									+
									O(\mu_{\max})
								\right]
								\nn\\
								&\quad
								\times
								\mathrm{vec}
								\left\{
									(p^T\otimes I_M)\mc{R}_v (p \otimes I_M)
								\right\}
								\nonumber\\
					&=			\mu_{\max} \cdot
								(
									I_M \otimes H_c + H_c \otimes I_M
								)^{-1}
								\nn\\
								&\quad
								\times
								\mathrm{vec}
								\left\{
									(p^T\otimes I_M)\mc{R}_v (p \otimes I_M)
								\right\}
								+ O(\mu_{\max}^2)
		\label{Equ:SteadyState:Pi_ci_o_final}
	\end{align}
Substituting \eqref{Equ:SteadyState:Pi_ci_o_final} into \eqref{Equ:SteadyState:limsup_WMSE_intermediate1} and\eqref{Equ:SteadyState:liminf_WMSE_intermediate1}, we obtain
	\begin{align}
		\limsup_{i\rightarrow\infty}&\; \E\|\tilde{\bw}_{k,i}\|_{\Sigma}^2
					\nonumber\\
			&\le		\mu_{\max}\cdot
					\left(
						\mathrm{vec}
						\left\{
							(p^T\otimes I_M)\cdot \mc{R}_v \cdot (p \otimes I_M)
						\right\}
					\right)^T
					\nn\\
					&\quad
					\times
					(
						I_M \otimes H_c^T \!+\! H_c^T \otimes I_M
					)^{-1}\mathrm{vec}(\Sigma) 
					\nn\\
					&\quad+ 
					O(\mu_{\max}^{2})
					+
					O\big(\mu_{\max}^{\min(3/2, 1+\kappa/2)}\big)
					\nn\\
			&=
					\mu_{\max}\cdot
					\left(
						\mathrm{vec}
						\left\{
							(p^T\otimes I_M)\cdot \mc{R}_v \cdot (p \otimes I_M)
						\right\}
					\right)^T
					\nn\\
					&\quad
					\times
					(
						I_M \otimes H_c^T \!+\! H_c^T \otimes I_M
					)^{-1}\mathrm{vec}(\Sigma) 
					\nn\\
					&\quad+
					O\big(\mu_{\max}^{\min(3/2, 1+\kappa/2)}\big)
	\label{Equ:SteadyState:limsup_WMSE_final}
					\\
		\liminf_{i\rightarrow\infty}& \;\E\|\tilde{\bw}_{k,i}\|_{\Sigma}^2
					\nonumber\\
			&\ge		\mu_{\max}\cdot
					\left(
						\mathrm{vec}
						\left\{
							(p^T\otimes I_M)\cdot \mc{R}_v \cdot (p \otimes I_M)
						\right\}
					\right)^T
					\nn\\
					&\quad
					\times
					(
						I_M \otimes H_c^T \!+\! H_c^T \otimes I_M
					)^{-1}\mathrm{vec}(\Sigma) 
					\nn\\
					&\quad+ 
					O(\mu_{\max}^{2})
					-
					O\big(\mu_{\max}^{\min(3/2, 1+\kappa/2)}\big)
					\nn\\
			&=
					\mu_{\max}\cdot
					\left(
						\mathrm{vec}
						\left\{
							(p^T\otimes I_M)\cdot \mc{R}_v \cdot (p \otimes I_M)
						\right\}
					\right)^T
					\nn\\
					&\quad
					\times
					(
						I_M \otimes H_c^T \!+\! H_c^T \otimes I_M
					)^{-1}\mathrm{vec}(\Sigma) 
					\nn\\
					&\quad-
					O\big(\mu_{\max}^{\min(3/2, 1+\kappa/2)}\big)
		\label{Equ:SteadyState:liminf_WMSE_final}
	\end{align}
Note that  the term 
$(
	I_M \otimes H_c^T \!+\! H_c^T \otimes I_M
)^{-1}\mathrm{vec}(\Sigma)$ in \eqref{Equ:SteadyState:limsup_WMSE_final} and \eqref{Equ:SteadyState:liminf_WMSE_final} is in fact the vectorized
version of the solution matrix $X$ to the Lyapunov equation 
\eqref{Equ:SteadyState:ContinuousLyapunovEqu_final}
for any given positive semi-definite weighting matrix $\Sigma$.
Using again the relation $\Tr(X Y) = \left(\mathrm{vec}(X^T)\right)^T \mathrm{vec}(Y) = \left( \left(\mathrm{vec}(X^T)\right)^T \mathrm{vec}(Y) \right)^T  = \mathrm{vec}(Y)^T \mathrm{vec}(X^T)$,
the $\limsup$ and $\liminf$ expressions \eqref{Equ:SteadyState:limsup_WMSE_final}--\eqref{Equ:SteadyState:liminf_WMSE_final} for the weighted MSE become
	\begin{align}
		\limsup_{i\rightarrow\infty} \E\|\tilde{\bw}_{k,i}\|_{\Sigma}^2
			&\le		
					\mu_{\max}\cdot
					\Tr\left\{
							X (p^T\otimes I_M)\cdot \mc{R}_v \cdot (p \otimes I_M)
					\right\}
					\nn\\
					&\quad
					+
					O\big(\mu_{\max}^{\min(3/2, 1+\kappa/2)}\big)
	\label{Equ:SteadyState:limsup_WMSE_finalfinal}
					\\
		\liminf_{i\rightarrow\infty} \E\|\tilde{\bw}_{k,i}\|_{\Sigma}^2
			&\ge		
					\mu_{\max}\cdot
					\Tr\left\{
							X (p^T\otimes I_M)\cdot \mc{R}_v \cdot (p \otimes I_M)
					\right\}
					\nn\\
					&\quad
					-
					O\big(\mu_{\max}^{\min(3/2, 1+\kappa/2)}\big)
		\label{Equ:SteadyState:liminf_WMSE_finalfinal}
	\end{align}
As a final remark, since condition \eqref{Equ:Lemma:EquivCondUpdateVec:StrongMono_HessianLB} ensures that all the eigenvalues of $H_c$ have positive real parts, i.e., the matrix $-H_c$ is asymptotically stable, the following Lyapunov equation, which is equivalent to \eqref{Equ:SteadyState:ContinuousLyapunovEqu_final}, 
	\begin{align}
		(-H_c^T) X + X (-H_c) = -\Sigma
		\nn
	\end{align}
will have a unique solution given by \eqref{Equ:SteadyState:ContinuousLyapunovEqu_Solution} \cite[pp.145-146]{laub2005matrix} and is positive semi-definite (strictly positive definite) if $\Sigma$ is symmetric and positive semi-definite (strictly positive definite) ( see \cite[p.39]{poliak1987introduction} and \cite[p.769]{kailath2000linear}).

\section{Proof of Lemma \ref{Lemma:AsymptoticBound_4thOrderMoments}}
\label{Appendix_SketchProof_4thOrderMoment}

The arguments in the previous appendix relied on results \eqref{Equ:Lemma:AsympBound_4thOrderMoments_wc} and \eqref{Equ:Lemma:AsympBound_4thOrderMoments_we} from Lemma \ref{Lemma:AsymptoticBound_4thOrderMoments}. To establish these results, we first need to introduce a fourth-order version of the energy operator we dealt with in Appendices \ref{P1-Appendix:Proof_BasicProperties} and \ref{P1-Appendix:Proof_VarianceRelations} in Part I \cite{chen2013learningPart1}, and establish some of its properties.  
	\begin{definition}[Fourth order moment operator]
		Let $x = \col\{ x_1,\ldots, x_N \}$ with sub-vectors of size $M\times 1$ each. We define $\Pm[x]$ to be an operator that maps from $\mb{R}^{MN}$ 
		to $\mb{R}^N$:
			\begin{align}
				\Pm[x]	\defeq	\col\{
									\|x_1\|^4, \|x_2\|^4, \ldots, \|x_N\|^4
								\}
								\nn
			\end{align}
		\hfill\QED
	\end{definition}
By following the same line of reasoning as the one used for the energy operator $P[\cdot]$ in Appendices \ref{P1-Appendix:Proof_BasicProperties} and \ref{P1-Appendix:Proof_VarianceRelations} in Part I \cite{chen2013learningPart1}, we can establish the following properties for $\Pm[\cdot]$.
	
\begin{lemma}[Properties of the 4th order moment operator]
	\label{Lemma:VarMomentPropt}
		The operator $\Pm[\cdot]$ satisfies the following properties:
			\begin{enumerate}
				\item {\bf (\emph{Nonnegativity}): }
					$\Pm[x] \succeq 0$
					
				\item {\bf (\emph{Scaling}): }
					$\Pm[a x] = |a|^4 \cdot \Pm[x]$
					
				\item {\bf (\emph{Convexity}): }
					Suppose $x^{(1)},\ldots, x^{(K)}$ are $N\times 1$ block vectors formed in the same manner 
					as $x$, and let $a_1, \ldots, a_K$ be non-negative real scalars that add up to one. Then,
						\begin{align}
							\Pm\big[ &a_1 x^{(1)} + \cdots + a_K x^{(K)} \big]
												\nn\\
									&\preceq		a_1 \Pm\big[ x^{(1)}\big] + \cdots + a_K \Pm\big[ x^{(K)} \big]
							\label{Equ:VarMomentPropt:Convexity}
						\end{align}
						
				\item {\bf (\emph{Super-additivity}): }
					\begin{align}
						\Pm[x+y] \preceq 8 \cdot \Pm[x] + 8 \cdot \Pm[y]
						\label{Equ:VarMomentPropt:SuperAdd}
					\end{align}

				\item {\bf (\emph{Linear transformation}): }
					\begin{align}
						\Pm[ Q x ]	&\preceq		\| \bP[Q] \|_{\infty}^3 \cdot \bP[Q] \; \Pm[x]
						\label{Equ:VarMomentPropt:Linear}
												\\
									&\preceq		\| \bP[Q] \|_{\infty}^4 \cdot \one \one^T \cdot \Pm[x]
						\label{Equ:VarMomentPropt:Linear_ub}
					\end{align}
					
				\item {\bf (\emph{Update operation}): }
					The global update vector $s(x) \defeq \col\{ s_1(x_1), \ldots, s_N(x_N)\}$ satisfies the
					following relation on $\Pm[\cdot]$:
						\begin{align}
							\Pm[ s(x) - s(y) ]	\preceq	\lambda_U^4 \cdot \Pm[x-y]
							\label{Equ:VarMomentPropt:Update}
						\end{align}
					
				\item {\bf (\emph{Centralizd operation}): }
					\begin{align}
						\Pm[T_c(x) - T_c(y)]		\preceq		\gamma_c^4 \cdot \Pm[x-y]
						\label{Equ:VarMomentPropt:Centralized}
					\end{align}
					with the same factor
						\begin{align}
							\gamma_c		\defeq		1 - \mu_{\max} \lambda_L 
														+ 
														\frac{1}{2}
														\mu_{\max}^2
														\|p\|_1^2
														\lambda_U^2
						\end{align}
					
%				\item
%	           			{\bf (\emph{Stable Jordan operation}):}
%	           			Suppose $D_L$ is an $L \times L$ stable
%	           			Jordan matrix of the same form as
%	           			\eqref{P1-Equ:Propert:DL_def}--\eqref{P1-Equ:Propert:DLn_def} in Part I
%	           			\cite{chen2013learningPart1}.
%	           			Then, for any $L \times 1$ vectors $x'$  and $y'$,
%	           			we have 
%	           				\begin{align}
%	           					\label{Equ:VarMomentPropert:StableJordanOperator}
%	           					\Pm_1[ D_L x' + y']	\preceq 
%	           										\Gamma_{e,4} \cdot \Pm_1[x'] 
%	           										+ 
%	           										\frac{8}{(1-|d_2|)^3}
%	           										\cdot
%	           										\Pm_1[y']
%	           				\end{align}
%	           			where $\Gamma_{e,4}$ is now the $L \times L$ matrix defined as
%	           				\begin{align}
%	           					\label{Equ:VarMomentPropert:Gamma_e}
%	           					\Gamma_{e,4}	&\defeq	\begin{bmatrix}
%			           									|d_2| & \frac{8}{(1-|d_2|)^3} & &
%			           														  	\\
%			           									    & \ddots & \ddots &	\\
%			           									    & & \ddots & 
%			           									    	\frac{8}{(1-|d_2|)^3}	\\
%			           									    & & & |d_2|
%			           								\end{bmatrix}
%	           				\end{align}
	           		\item
	           			{\bf (\emph{Stable Kronecker Jordan operator}):}
	           			Suppose $\mc{D}_L = D_L \otimes I_M$, where $D_L$ is 
	           			the $L \times L$ Jordan matrix defined
	           			by \eqref{P1-Equ:Propert:DL_def}--%
	           			\eqref{P1-Equ:Propert:DLn_def} in Part I\cite{chen2013learningPart1}. 
	           			Then, for any
	           			$LM \times 1$ vectors $x_e$ and $y_e$, we have
	           				\begin{align}
	           					\label{Equ:VarMomentPropert:StableKronJordanOperator}
	           					\Pm[\mc{D}_L x_e + y_e]
	           						\preceq	\Gamma_{e,4} \!\cdot \!\Pm[x_e] \!+ \!
	           						\frac{8}{(1 \!-\! |d_2|)^3} \!\cdot\! \Pm[y_e]
	           				\end{align}
					where $\Gamma_{e,4}$ is the $L \times L$ matrix defined as
	           				\begin{align}
	           					\label{Equ:VarMomentPropert:Gamma_e}
	           					\Gamma_{e,4}	&\defeq	\begin{bmatrix}
			           									|d_2| & \frac{8}{(1-|d_2|)^3} & &
			           														  	\\
			           									    & \ddots & \ddots &	\\
			           									    & & \ddots & 
			           									    	\frac{8}{(1-|d_2|)^3}	\\
			           									    & & & |d_2|
			           								\end{bmatrix}
	           				\end{align}
					
			\end{enumerate}
		\hfill\QED
	\end{lemma}

\noindent 
To proceed, we recall the transformed recursions \eqref{P1-Equ:Lemma:ErrorDynamics:JointRec_wc_check}--\eqref{P1-Equ:Lemma:ErrorDynamics:JointRec_TF2_we} from Part I\cite{chen2013learningPart1}, namely, 
	\begin{align}
			\check{\bm{w}}_{c,i}	
							&=	T_c(\bm{w}_{c,i-1}) 
								-
								T_c(\bar{w}_{c,i-1})
								\nn\\
								&\quad
								- 
								\mu_{\max}\cdot (p^T \otimes I_M) 
								\left[
										\bm{z}_{i-1}
										+
										\bm{v}_i
								\right]
				\label{Equ:Appendix:wci_check_recursion_Part2}
							\\
				\bm{w}_{e,i}	&=	\mD_{N-1} \bm{w}_{e,i-1}
								\!-\!
								\mU_R\mA_2^T \mM
								\left[
										s(\mathds{1} \otimes \bm{w}_{c,i-1})
										\!+\!
										\bm{z}_{i-1}
										\!+\!
										\bm{v}_i
								\right]
				\label{Equ:Appendix:wei_check_recursion_Part2}
		\end{align}
If we now apply the operator $\Pm[\cdot]$ to recursions \eqref{Equ:Appendix:wci_check_recursion_Part2}--\eqref{Equ:Appendix:wei_check_recursion_Part2}, and follow arguments similar to the those employed in Appendices \ref{P1-Appendix:Proof_Lemma_W_check_prime_recursion} and \ref{P1-Appendix:Proof_BoundsPerturbation} from Part I \cite{chen2013learningPart1}, we arrive at the following result. The statement extends Lemma \ref{P1-Lemma:IneqRecur_W_check_prime} in Part I \cite{chen2013learningPart1} to  $4$th order moments. 
	\begin{lemma}[Recursion for the 4th order moments]
		\label{Lemma:Recursion4thOrderMoments}
		The fourth order moments satisfy the following inequality recursion
			\begin{align}
				\check{\mW}_{4,i}'	&\preceq	F_4 \; \check{\mW}_{4,i-1}' 
												+ 
												H_4 \check{\mW}_{i-1}' 
												+ 
												\mu_{\max}^4 
												\cdot 
												b_{v,4}
				\label{Equ:StabBound4thOrder:IneqRec_4thOrderMoment_final}
			\end{align}
		where
			\begin{align}
				\check{\mW}_{4,i}'	&\defeq		\col\big\{
													\Expt \Pm [ \check{\w}_{c,i} ], \;
													\Expt \Pm [ \w_{e,i} ]			
												\big\}
				\label{Equ:StabBound4thOrder:mW_4i_prime_def}
												\\
				\check{\mW}_{i}'	&\defeq		\col\big\{
													\Expt P[ \check{\w}_{c,i} ], \;
													\Expt P[ \w_{e,i} ]
												\big\}
												\\
				F_4					&\defeq		\begin{bmatrix}
													f_{cc}(\mu_{\max})		
														&	f_{ce}(\mu_{\max}) \cdot \one^T	
														\\
													f_{ec}(\mu_{\max}) \cdot \one	&	F_{ee}(\mu_{\max})
												\end{bmatrix}
				\label{Equ:StabBound4thOrder:F4_def}
												\\
				H_4					&\defeq		\begin{bmatrix}
													h_{cc}(\mu_{\max})	
														&	h_{ce}(\mu_{\max}) \cdot \one^T		\\
													0	&	0
												\end{bmatrix}
				\label{Equ:StabBound4thOrder:H4_def}
												\\
				b_{v,4}				&\defeq		\col\big\{
													b_{v_4, c}, \; b_{v_4, e} \cdot \one
											\big\}
			\end{align}
		where $\gamma_c$ is defined in \eqref{Equ:VarPropt:gamma_c}, and $\Gamma_{e,4}$ is
		defined in \eqref{Equ:VarMomentPropert:Gamma_e}, Moreover, the entries in \eqref{Equ:StabBound4thOrder:F4_def}--\eqref{Equ:StabBound4thOrder:H4_def} are given by:
			\begin{align}
			f_{cc}(\mu_{\max})	
								&\defeq
											\gamma_c
											+
											\mu_{\max}^4
											\cdot
											432\alpha_4 \|p\|_1^4
											\nn\\
											&\quad
											+
											\mu_{\max}^2
											\cdot
											20
											\alpha \|p\|_1^2
											\nn\\
											&\qquad
											\cdot
											\Big(
												2
												+
												\| \bP[ \mA_1^T \mU_L ] \|_{\infty}^2
												\cdot
												\frac{
													\lambda_L+\frac{1}{2}\mu_{\max} \|p\|_1^2 \lambda_U^2
												}
												{
													\lambda_L-\frac{1}{2}\mu_{\max} \|p\|_1^2 \lambda_U^2
												}
											\Big)
											\nn\\
									&=
											\gamma_c + O(\mu_{\max}^2)
			\label{Equ:StabBound4thOrder:fcc_def}
											\\
			f_{ce}(\mu_{\max})
								&\defeq
											\mu_{\max} \cdot
											\frac{
												\|p\|_1^4 \cdot \lambda_U^4
												\cdot
												\left\|
													\bP[\mA_1^T \mU_L ]
												\right\|_{\infty}^4
											}
											{
												(\lambda_L - \frac{1}{2}\mu_{\max} \|p\|_1^2 \lambda_U^2)^3
											}	
											\nn\\
											&\quad
											+
											432 \mu_{\max}^4
											\alpha_4 \|p\|_1^4
											\cdot
											\| \bP[ \mA_1^T \mU_L ] \|_{\infty}^4
											\nn\\
											&\quad
											+
											20\mu_{\max}^2
											\alpha \|p\|_1^2
											\cdot
											\| \bP[ \mA_1^T \mU_L ] \|_{\infty}^2
											\nn\\
											&\quad
											\cdot
											\Big(
												\mu_{\max} \cdot
												\frac{
														2 \|p\|_1^2 \cdot \lambda_U^2 
														\cdot
														\| \bP[ \mA_1^T \mU_L ] \|_{\infty}^2
												}
												{
													\lambda_L-\frac{1}{2}\mu_{\max} \|p\|_1^2 \lambda_U^2
												}	
												\nn\\
												&\qquad
												+
												\frac{\lambda_L + \frac{1}{2}\mu_{\max} \|p\|_1^2 \lambda_U^2}
												{\lambda_L-\frac{1}{2}\mu_{\max} \|p\|_1^2 \lambda_U^2}
											\Big)
											\nn\\
							&=
											O(\mu_{\max})
			\label{Equ:StabBound4thOrder:fce_def}
											\\
			h_{cc}(\mu_{\max})	
							&\defeq
											10\mu_{\max}^2
											\cdot
											\Big(
													4 \alpha \|p\|_1^2 
													\cdot
													\|\tilde{w}_{c,0}\|^2
													\nn\\
													&\quad
													+
													4 \alpha \|p\|_1^2
													\cdot
													\|w^o\|^2
													+
													\sigma_{v}^2
													\cdot
													\|p\|_1^2
											\Big)
											\nn\\
							&=
											O(\mu_{\max}^2)
			\label{Equ:StabBound4thOrder:hcc_def}
											\\
			h_{ce}(\mu_{\max})
							&\defeq	
											\frac{10\|p\|_1^4 \lambda_U^2 \cdot \mu_{\max}^3}
											{\lambda_L \!-\! \frac{1}{2}\mu_{\max} \|p\|_1^2 \lambda_U^2}	
											\!\cdot\!
											\big\|
												\bP[\mA_1^T \mU_L ]
											\big\|_{\infty}^2
											\nn\\
											&\quad
											\cdot\!
											\Big(\!
												4 \alpha
												\!\cdot\!
												\|\tilde{w}_{c,0}\|^2
												\!+\!
												4 \alpha
												\!\cdot\!
												\|w^o\|^2
												\!+\!
												\sigma_{v}^2
												\!\cdot\!
											\Big)
											\nn\\
							&=
											O(\mu_{\max}^3)
			\label{Equ:StabBound4thOrder:hce_Order}
											\\
			b_{v_4, c}			
							&\defeq
											2
											\|p\|_1^4
											\cdot
											\big(
												27 \alpha_4 
												\cdot 
												(
													\|\tilde{w}_{c,0}\|^4
													+
													\|w^o\|^4
												)
												+
												\sigma_{v4}^4
											\big)
											\nn\\
							&=			
											\mathrm{constant}
			\label{Equ:StabBound4thOrder:xic_def}
											\\
			F_{ee}(\mu_{\max})			
					&\defeq			\Gamma_{e,4}
									+
									\mu_{\max}^4
									\cdot
									\frac{216N \cdot (\lambda_U^4 + 216 \alpha_4)}
									{(1-|\lambda_2(A)|)^3}
									\nn\\
									&\qquad
									\times
									\| \bP[\mA_1^T \mU_L] \|_{\infty}^4
									\cdot
									\| \bP[ \mU_R \mA_2^T ] \|_{\infty}^4
									\cdot
									\one \one^T
									\nn\\
					&=
									\Gamma_{e,4} 
									+
									O(\mu_{\max}^4)
			\label{Equ:StabBound4thOrder:Fee_def}
									\\
			f_{ec}(\mu_{\max})	
						&\defeq
										\mu_{\max}^4
										\cdot
										\frac{5832N \cdot (\lambda_U^4+8\alpha_4) }
										{(1-|\lambda_2(A)|)^3}
										\| \bP[ \mU_R \mA_2^T ] \|_{\infty}^4
										\nn\\
						&=
										O(\mu_{\max}^4)
			\label{Equ:StabBound4thOrder:fec_def}
										\\
			b_{v_4, e}
						&\defeq			\frac{
											216N\cdot
											\| 
												\bP[ \mU_R \mA_2^T ] 
											\|_{\infty}^4
										}
										{
											(1-|\lambda_2(A)|)^3
										}
										\cdot
										\Big\{
											27
											\big[
												(\lambda_U^4 + \alpha_4)
												\cdot
												\|\tilde{w}_{c,0}\|^4
												\nn\\
												&\quad
												+
												\|g_4^o\|_{\infty}
												+
												\alpha_4
												\cdot
												\|w^o\|^4
											\big]
											+
											\sigma_{v4}^4
										\Big\}
										\nn\\
						&=			
										\mathrm{constant}
			\label{Equ:StabBound4thOrder:xie_def}
	\end{align}
%		but just list their order of magnitude relative to $\mu_{\max}$ below, which is relevant to our later derivations:
%			\begin{align}
%				f_{cc}(\mu_{\max})	
%									&=			\gamma_c + O(\mu_{\max}^2)
%				\label{Equ:StabBound4thOrder:fcc_def}
%												\\
%				f_{ce}(\mu_{\max})
%									&=			O(\mu_{\max})
%				\label{Equ:StabBound4thOrder:fce_def}
%												\\
%				f_{ec}(\mu_{\max})		&=				O(\mu_{\max}^4)
%				\label{Equ:StabBound4thOrder:fec_def}
%														\\
%				F_{ee}(\mu_{\max})		&=				\Gamma_{e,4} + O(\mu_{\max}^4)
%				\label{Equ:StabBound4thOrder:Fee_def}
%														\\
%				h_{cc}(\mu_{\max})	
%								&=
%												O(\mu_{\max}^2)
%				\label{Equ:StabBound4thOrder:hcc_Order}
%												\\
%				h_{ce}(\mu_{\max})
%								&=	
%												O(\mu_{\max}^2)
%				\label{Equ:StabBound4thOrder:hce_Order}
%												\\
%				b_{v_4, c}			
%									&=
%												\mathrm{costant}
%												\\
%				b_{v_4, e}					&=				\mathrm{constant}
%			\end{align}
%		where $\gamma_c$ is defined in \eqref{Equ:VarPropt:gamma_c}, and $\Gamma_{e,4}$ is defined in 
%		\eqref{Equ:VarMomentPropert:Gamma_e}.
	\end{lemma}
	\begin{proof}
		See Appendix \ref{Appendix:Proof_Lemma_IneqRec4thOrderMoments}.
	\end{proof}

Observe from \eqref{Equ:StabBound4thOrder:IneqRec_4thOrderMoment_final} that the recursion of the fourth order moments are coupled with the second order moments contained in $\check{\mW}_{i-1}'$. Therefore, we will augment recursion \eqref{Equ:StabBound4thOrder:IneqRec_4thOrderMoment_final} together with the following recursion for the second-order moment developed in  \eqref{P1-Equ:FirstOrderAnal:W_i_prime_ineq_Rec1} of Part I \cite{chen2013learningPart1}:
	\begin{align}
		\label{Equ:FirstOrderAnal:W_i_prime_ineq_Rec1}
		\check{\mc{W}}_i'	\preceq		\Gamma \check{\mc{W}}_{i-1}' + \mu_{\max}^2 b_v
	\end{align}
to form the following joint recursion:
	\begin{align}
		\begin{bmatrix}
			\check{\mW}_{i}'	\\
			\check{\mW}_{4,i}'
		\end{bmatrix}
							&\preceq	
										\begin{bmatrix}
											\Gamma		&		0	\\
											H_4			&		F_4
										\end{bmatrix}
										\begin{bmatrix}
											\check{\mW}_{i-1}'	\\
											\check{\mW}_{4,i-1}'
										\end{bmatrix}
										+
										\begin{bmatrix}
											\mu_{\max}^2 \cdot b_v		\\
											\mu_{\max}^4 \cdot b_{v,4}
										\end{bmatrix}
		\label{Equ:StabBound4thOrder:IneqRec_Joint4th2ndOrderMoment_final}
	\end{align}
The stability of the above recursion is guaranteed by the stability of the matrices $\Gamma$ and $F_4$, i.e., 
	\begin{align}
		\rho(\Gamma) < 1 
		\quad\mathrm{and}\quad
		\rho(F_4) < 1
		\nn
	\end{align}
The stability of $\Gamma$ has already been established in Appendix \ref{P1-Appendix:Proof_Thm_NonAsymptotiBound} of Part I \cite{chen2013learningPart1}. Now, we discuss the stability of $F_4$. Using \eqref{Equ:StabBound4thOrder:fcc_def}--\eqref{Equ:StabBound4thOrder:Fee_def} and the definition of $\gamma_c$ in \eqref{Equ:Appendix:Bc_norm_UB}, we can express $F_4$ as
	\begin{align}
		F_4				&=				\begin{bmatrix}
											\gamma_c + O(\mu_{\max}^2)
													&	O(\mu_{\max}) \cdot \one^T	\\
											O(\mu_{\max}^4)
													&	\Gamma_{e,4} + O(\mu_{\max}^4)
										\end{bmatrix}
		\label{Equ:StabBound4thOrder:F4_highOrderForm1}
										\\
						&=				\begin{bmatrix}
											1-\mu_{\max} \lambda_L
													&	O(\mu_{\max})	\\
											0		&	\Gamma_{e,4}
										\end{bmatrix}
										+
										O(\mu_{\max}^2)
		\label{Equ:StabBound4thOrder:F4_highOrderForm2}
	\end{align}
which has a similar structure to $\Gamma$ --- see expressions \eqref{P1-Equ:FirstOrderAnal:Gamma_def}--\eqref{P1-Equ:FirstOrderAnal:Gamma0_def} in Part I\cite{chen2013learningPart1}, and where in the last step we absorb the factor $\one^T$ in the $(1,2)$-th block into $O(\mu_{\max})$. Therefore, following the same line of argument from \eqref{P1-Equ:Appendix:epsilon_def} to \eqref{P1-Equ:Appendix:rhoGama_stepsizecond3} in Appendix \ref{P1-Appendix:Proof_Thm_NonAsymptotiBound} of Part I \cite{chen2013learningPart1}, we can show that $F_4$ is also stable when the step-size parameter $\mu_{\max}$ is sufficiently small. Iterating \eqref{Equ:StabBound4thOrder:IneqRec_Joint4th2ndOrderMoment_final}, we get
	\begin{align}
		\begin{bmatrix}
			\check{\mW}_{i}'	\\
			\check{\mW}_{4,i}'
		\end{bmatrix}
							&\preceq	
										\begin{bmatrix}
											\Gamma		&		0	\\
											H_4			&		F_4
										\end{bmatrix}^i
										\!
										\begin{bmatrix}
											\check{\mW}_{0}'	\\
											\check{\mW}_{4,0}'
										\end{bmatrix}
										\!+\!
										\sum_{j=0}^{i-1}
										\begin{bmatrix}
											\Gamma		&		0	\\
											H_4			&		F_4
										\end{bmatrix}^j
										\!\cdot\!
										\begin{bmatrix}
											\mu_{\max}^2 \!\cdot\! b_v		\\
											\mu_{\max}^4 \!\cdot\! b_{v,4}
										\end{bmatrix}
		\label{Equ:StabBound4thOrder:IneqRec_Joint4th2ndOrderMoment_final_expanded}
	\end{align}
When both $\Gamma$ and $F_4$ are stable, we have
	\begin{align}
		&\limsup_{i \rightarrow \infty}
		\begin{bmatrix}
			\check{\mW}_{i}'	\\
			\check{\mW}_{4,i}'
		\end{bmatrix}	
										\nn\\
							&\preceq
										\left(
											I
											-
											\begin{bmatrix}
												\Gamma		&		0	\\
												H_4			&		F_4
											\end{bmatrix}
										\right)^{-1}
										\cdot
										\begin{bmatrix}
											\mu_{\max}^2 \cdot b_v		\\
											\mu_{\max}^4 \cdot b_{v,4}
										\end{bmatrix}
										\nn\\
							&=			\begin{bmatrix}
											\mu_{\max}^2 \!\cdot\! (I \!-\! \Gamma)^{-1} b_v	\\
											(I \!-\! F_4)^{-1}
											H_4 \!\cdot \!
											\mu_{\max}^2 
											\cdot (I \!-\! \Gamma)^{-1} b_v
											\!+\!
											\mu_{\max}^4 \!\cdot\!
											(I \!-\! F_4)^{-1} b_{v,4}			
										\end{bmatrix}
										\nn
	\end{align}
which implies that, for the fourth-order moment, we get
	\begin{align}
		\limsup_{i \rightarrow \infty}
		\check{\mW}_{4,i}'
							&\preceq	(I-F_4)^{-1}
										H_4 \cdot 
										\mu_{\max}^2 
										\cdot (I-\Gamma)^{-1} b_v
										\nn\\
										&\quad
										+
										\mu_{\max}^4 \cdot
										(I-F_4)^{-1} b_{v,4}	
		\label{Equ:StabBound4thOrder:AsymptoticBound_4ndOrder}
	\end{align}
To evaluate the right-hand side of the above expression, we derive expressions for $(I-F_4)^{-1}$ and $(I-\Gamma)^{-1}$ using the following formula for inverting a $2\times 2$ block matrix\cite[p.48]{laub2005matrix},\cite[p.16]{Sayed08}:
	\begin{align}
		\begin{bmatrix}
			A	&	B	\\
			C	&	D
		\end{bmatrix}^{-1}
							&=			\begin{bmatrix}
											A^{-1} + A^{-1} B E C A^{-1} 	& 	-A^{-1} B E	\\
											-E C A^{-1}						&	E				
										\end{bmatrix}
		\label{Equ:Appendix:Inverse_2by2Matrix}
	\end{align}
where $E = (D-CA^{-1}B)^{-1}$. By \eqref{Equ:StabBound4thOrder:F4_highOrderForm2}, we have the following expression for $(I-F_4)^{-1}$:
	\begin{align}
		&(I-F_4)^{-1}					\nn\\
						&=			\left(
											I 
											\!-\!
											\begin{bmatrix}
												1 \! - \! \mu_{\max} \lambda_L
														&	O(\mu_{\max})	\\
												0		&	\Gamma_{e,4}
											\end{bmatrix}
											\!-\!
											O(\mu_{\max}^2) 
										\right)^{-1}
										\nn\\
							&=			\begin{bmatrix}
											\mu_{\max} \lambda_L - O(\mu_{\max}^2)
													&	- O(\mu_{\max}) \!-\! O(\mu_{\max}^2)
														\\
											-O(\mu_{\max}^2)		
													&I\!-\! \Gamma_{e,4} \!-\! O(\mu_{\max}^2)
										\end{bmatrix}^{-1}
										\nn\\
							&=			\begin{bmatrix}
											\mu_{\max} \lambda_L - O(\mu_{\max}^2)
													&	O(\mu_{\max})	\\
											O(\mu_{\max}^2)		
													&	I- \Gamma_{e,4} - O(\mu_{\max}^2)
										\end{bmatrix}^{-1}	
		\label{Equ:StabBound4thOrder:IF4_inverse_interm1}	
	\end{align}
Applying relation \eqref{Equ:Appendix:Inverse_2by2Matrix} to \eqref{Equ:StabBound4thOrder:IF4_inverse_interm1}, we have
	\begin{align}
		E_4		&=			\left(
								I - \Gamma_{e,4} - O(\mu_{\max}^2)
								- 
								\frac{O(\mu_{\max}^2) O(\mu_{\max})}
								{\mu_{\max} \lambda_L - O(\mu_{\max}^2)}
							\right)^{-1}
							\nn\\
				&=			\left(
								I - \Gamma_{e,4} + O(\mu_{\max}^2)
							\right)^{-1}
							\\
		(I&-F_4)^{-1}		
							\nn\\
				&=			
%							\begin{bmatrix}
%											\frac{1}{\mu_{\max} \lambda_L - O(\mu_{\max}^2)}
%											+ 
%											\frac{
%												O(\mu_{\max}) 
%												\cdot
%												E_4
%												\cdot
%												O(\mu_{\max}^2)
%											}
%											{
%												(\mu_{\max} \lambda_L - O(\mu_{\max}^2))^2
%											}
%												&	
%												- \frac{O(\mu_{\max}) \cdot E_4}
%												{\mu_{\max} \lambda_L - O(\mu_{\max}^2)}
%													\\
%											-\frac{E_4 \cdot O(\mu_{\max}^2)}
%											{\mu_{\max} \lambda_L - O(\mu_{\max}^2)}
%												&	E_4
%										\end{bmatrix}
%										\nn\\
%							&=			
										\begin{bmatrix}
											\frac{1}{\mu_{\max} \lambda_L - O(\mu_{\max}^2)}
											+ 
											O(\mu_{\max})
												&	
												- \frac{O(1) \cdot E_4}
												{\lambda_L - O(\mu_{\max})}
													\\
											-\frac{E_4 \cdot O(\mu_{\max})}
											{\lambda_L - O(\mu_{\max})}
												&	E_4
										\end{bmatrix}
		\label{Equ:Appendix:IF4_inv_final}
	\end{align}
Furthermore, recall from \eqref{P1-Equ:FirstOrderAnal:Gamma_def}--\eqref{P1-Equ:FirstOrderAnal:Gamma0_def} of Part I\cite{chen2013learningPart1} for the expression of $\Gamma$:
	\begin{align}
		\Gamma		&=
							\begin{bmatrix}
								\gamma_c		&	\mu_{\max} h_c(\mu_{\max}) \cdot \mathds{1}^T
												\\
								0			&	\Gamma_e
							\end{bmatrix} + \mu_{\max}^2 \psi_0 \cdot \one \one^T
							\nn\\
					&=
							\begin{bmatrix}
								1-\mu_{\max}\lambda_L
											&	O(\mu_{\max})
												\\
								0			&	\Gamma_e
							\end{bmatrix} + O(\mu_{\max}^2)
							\nn
	\end{align}
Observing that $\Gamma$ and $F_4$ have a similar structure, we can similarly get the expression for $(I-\Gamma)^{-1}$ as
	\begin{align}
		(I-\Gamma)^{-1}
					&=						
							\begin{bmatrix}
								\frac{1}{\mu_{\max} \lambda_L - O(\mu_{\max}^2)}
								+ 
								O(\mu_{\max})
									&	
									- \frac{O(1) \cdot E_2}
									{\lambda_L - O(\mu_{\max})}
										\\
								-\frac{E_2 \cdot O(\mu_{\max})}
								{\lambda_L - O(\mu_{\max})}
									&	E_2
							\end{bmatrix}
		\label{Equ:Appendix:IGamma_inv_final}	
							\\
		E_2		&=			\left[
								I - \Gamma_{e} - O(\mu_{\max}^2)
								- 
								\frac{O(\mu_{\max}^2) O(\mu_{\max})}
								{\mu_{\max} \lambda_L - O(\mu_{\max}^2)}
							\right]^{-1}
							\nn\\
				&=			\left(
								I - \Gamma_{e} + O(\mu_{\max}^2)
							\right)^{-1}
	\end{align}
In addition, by substituting \eqref{Equ:StabBound4thOrder:hcc_def}--\eqref{Equ:StabBound4thOrder:hce_Order} into \eqref{Equ:StabBound4thOrder:H4_def}, we note that
	\begin{align}
		H_4		=			\begin{bmatrix}
								O(\mu_{\max}^2)	&	O(\mu_{\max}^3)	\\
								0				&	0
							\end{bmatrix}		 
		\label{Equ:Appendix:H4_order}
	\end{align}
Substituting \eqref{Equ:Appendix:IF4_inv_final}, \eqref{Equ:Appendix:IGamma_inv_final} and \eqref{Equ:Appendix:H4_order} into the right-hand side of \eqref{Equ:StabBound4thOrder:AsymptoticBound_4ndOrder} and using 
we obtain 
	\begin{align}
		&\limsup_{i \rightarrow \infty}
		\check{\mW}_{4,i}'
									\nn\\
							&\preceq	\begin{bmatrix}
											\frac{1}{\mu_{\max} \lambda_L - O(\mu_{\max}^2)}
											+ 
											O(\mu_{\max})
												&	
												- \frac{O(1) \cdot E_4}
												{\lambda_L - O(\mu_{\max})}
													\\
											-\frac{E_4 \cdot O(\mu_{\max})}
											{\lambda_L - O(\mu_{\max})}
												&	E_4
										\end{bmatrix}	
										\nn\\
										&\quad
										\times\!
										\begin{bmatrix}
											O(\mu_{\max}^2)	&	O(\mu_{\max}^3)	\\
											0				&	0
										\end{bmatrix} 
										\nn\\
										&\quad 
										\times\!
										\mu_{\max}^2 
										\!\cdot \!
										\begin{bmatrix}
											\frac{1}{\mu_{\max} \lambda_L \!- O(\mu_{\max}^2)}
											\!+ \!
											O(\mu_{\max})
												&	
												-\! \frac{O(1) \cdot E_2}
												{\lambda_L \!- O(\mu_{\max})}
													\\
											-\!\frac{E_2 \!\cdot\! O(\mu_{\max})}
											{\lambda_L \!- O(\mu_{\max})}
												&	E_2
										\end{bmatrix}\!\!
										\begin{bmatrix}
											b_{v,c}	\\
											b_{v,e} \!\cdot\! \one
										\end{bmatrix}
										\nn\\
										&\quad
										+\!
										\mu_{\max}^4 \!\cdot\!
										\begin{bmatrix}
											\frac{1}{\mu_{\max} \lambda_L \!- O(\mu_{\max}^2)}
											\!+ \!
											O(\mu_{\max})
												&	
												- \frac{O(1) \cdot E_4}
												{\lambda_L - O(\mu_{\max})}
													\\
											-\! \frac{E_4 \cdot O(\mu_{\max})}
											{\lambda_L \!- O(\mu_{\max})}
												&	E_4
										\end{bmatrix}\!\!
										\begin{bmatrix}
											b_{v_4,c}	\\
											b_{v_4,e} \!\cdot\! \one
										\end{bmatrix}
										\nn\\
							&=			
										\begin{bmatrix}
											O(\mu_{\max}^2)		\\
											O(\mu_{\max}^4)	
										\end{bmatrix}	
		\label{Equ:Appendix:limsup_mW_4i_prime_check_ub_final}															
	\end{align}
where the last step follows from basic matrix algebra. Recalling the definition of $\check{\mW}_{4,i}'$ in \eqref{Equ:StabBound4thOrder:mW_4i_prime_def}, we conclude \eqref{Equ:Lemma:AsympBound_4thOrderMoments_wc}--\eqref{Equ:Lemma:AsympBound_4thOrderMoments_we} from \eqref{Equ:Appendix:limsup_mW_4i_prime_check_ub_final}.

\section{Proof of Lemma \ref{Lemma:Recursion4thOrderMoments}}
\label{Appendix:Proof_Lemma_IneqRec4thOrderMoments}

\subsection{Perturbation Bounds}

Before pursuing the proof of Lemma \ref{Lemma:Recursion4thOrderMoments}, we first state a result that bounds the fourth-order moments of the perturbation terms that appear in \eqref{Equ:Appendix:wci_check_recursion_Part2}, in a manner similar to the bounds we already have for the second-order moments in \eqref{Equ:Lemma:BoundsPerturbation:P_z}--\eqref{Equ:Lemma:BoundsPerturbation:P_v}.
	\begin{lemma}[Fourth-order bounds on the perturbation terms]
		\label{Lemma:BoundsPerturbation_4thMoment}
		Referring to \eqref{Equ:Appendix:wci_check_recursion_Part2}, the following bounds hold for any $i \ge 0$.
			\begin{align}
				\label{Equ:Lemma:BoundsPerturbation4th:P_z}
				\Pm[\bm{z}_{i\!-\!1}]	
							&\preceq			
											\lambda_U^4
											\!\cdot\!
											\left\|
												\bP[\mc{A}_1^T \mc{U}_L]
											\right\|_{\infty}^4
											\!\cdot\!
											\mathds{1}\mathds{1}^T
											\!\cdot\!
											\Pm[\bm{w}_{e,i-1}]
											\\
				\label{Equ:Lemma:BoundsPerturbaton4th:P_s}
				\Pm[s(\mathds{1}\otimes& \bm{w}_{c,i\!-\!1})]
											\nn\\
								&\preceq		
											27\lambda_U^4 \!\cdot\!
											\|\check{\bm{w}}_{c,i-1}\|^4
											\!\cdot\! \mathds{1}
											\!+\!
											27\lambda_U^4 \|\tilde{w}_{c,0}\|^4
											\!\cdot\!
											\mathds{1}
											\!+\!
											27 \!\cdot \!								
											g_4^o
											\\
				\Expt \big\{ \Pm[\bm{v}_i&]	\big| \mc{F}_{i-1} \big\}
											\nn\\
							&\preceq			
											216 \alpha_4 \cdot \mathds{1}
											\cdot
											\Pm[\check{\bm{w}}_{c,i-1}]
											\nn\\
											&\quad
											+
											216
											\alpha_4
											\cdot
											\left\|
												\bP[\mc{A}_1^T\mc{U}_L]
											\right\|_{\infty}^4
											\!\cdot\!
											\mathds{1}\mathds{1}^T
											\!\cdot\!
											\Pm[\bm{w}_{e,i\!-\!1}]
											\nonumber\\		
											&\quad
											+\!
											27
											\alpha_4
											\!\cdot\!
											(
												\|\tilde{w}_{c,0}\|^4
												\!+\!
												\|w^o\|^4
											) 
											\!\cdot\!
											\one
											\!+\!
											\sigma_{v4}^4
											\!\cdot\!
											\mathds{1}		
				\label{Equ:Lemma:BoundsPerturbation4th:P_v}
			\end{align}
		where $g_4^o		\defeq	\Pm[s(\mathds{1}\otimes w^o)]$.
%			\begin{align}
%				\label{Equ:Lemma:BoundsPerturbaton:g_o}
%				g^o		\defeq	P[s(\mathds{1} \otimes w^o)]
%			\end{align}
	\end{lemma}
	\begin{proof}
		See Appendix \ref{Appendix:Proof_BoundsPerturbation4th}.
	\end{proof}

% ========================================================================
%  Subsections for the inequaity recursion bound for wcicheck and we
% ========================================================================

\subsection{Recursion for the 4th order moment of $\check{\w}_{c,i}$}

\label{Appendix:RecursionEPwcicheck}

To begin with, note that by evaluating the squared Euclidean norm of both sides of \eqref{Equ:Appendix:wci_check_recursion_Part2} we obtain:
	\begin{align}
		&\|\check{\w}_{c,i}\|^2	
									\nn\\
								&=	
									\big\|
										T_c(\w_{c,i-1}) \!-\! T_c(\bar{w}_{c,i-1})
										\!-\!
										\mu_{\max}
										\!\cdot\!
										(p^T \!\otimes\! I_M)
										(\z_{i-1} \!+\! \bv_i)
									\big\|^2	
									\nn\\
								&=	\big\|
										T_c(\w_{c,i-1}) - T_c(\bar{w}_{c,i-1})
										-
										\mu_{\max}
										\cdot
										(p^T \otimes I_M)
										\z_{i-1}
									\big\|^2
									\nn\\
									&\quad
									+
									\mu_{\max}^2
									\cdot
									\big\|
										(p^T \otimes I_M) 
										\bv_i
									\big\|^2
									\nn\\
									&\quad
									-\!
									2\mu_{\max}
									\!\cdot\!
									\big[
										T_c(\w_{c,i-1}) \!-\! T_c(\bar{w}_{c,i-1})
										\!-\!
										\mu_{\max}
										\!\cdot\!
										(p^T \!\otimes \!I_M)
										\z_{i-1}
									\big]^T
									\nn\\
									&\qquad
									\cdot
									(p^T \otimes I_M)
									\bv_i
									\nn
	\end{align}	 
By further squaring both sides of the above expression, we get
	\begin{align}
		&\| \check{\w}_{c,i} \|^4
								\nn\\
					&=			\big\|
									T_c(\w_{c,i-1}) \!- \!T_c(\bar{w}_{c,i-1})
									\!-\!
									\mu_{\max}
									\!\cdot\!
									(p^T \!\otimes\! I_M)
									\z_{i-1}
								\big\|^4
								\nn\\
								&\quad+\!
								\big\{
									\mu_{\max}^2
									\!\cdot\!
									\| (p^T \!\otimes\! I_M) \bv_i \|^2
									\!-\!
									2\mu_{\max}
									\!\cdot\!
									\big[
										T_c(\w_{c,i-\!1}) 
										\!-\! 
										T_c(\bar{w}_{c,i-\!1})
										\nn\\
										&\qquad
										-\!
										\mu_{\max}
										\!\cdot\!
										(p^T \!\otimes\! I_M)
										\z_{i-\! 1}
									\big]
									(p^T \!\otimes\! I_M) \bv_i
								\big\}^2
								\nn\\
								&\quad-\!
								4\mu_{\max}
								\!\cdot\!
								\left\|
									T_c(\w_{c,i-1}) \!-\! T_c(\bar{w}_{c,i-1})
									\!-\!
									\mu_{\max}
									\!\cdot\!
									(p^T \!\otimes\! I_M)
									\z_{i-1}
								\right\|^2
								\nn\\
								&\qquad
								\cdot
								\left[
									T_c(\w_{c,i-1}) - T_c(\bar{w}_{c,i-1})
									-
									\mu_{\max}
									\cdot
									(p^T \otimes I_M)
									\z_{i-1}
								\right]^T
								\nn\\
								&\qquad
								\cdot
								(p^T \otimes I_M)\bv_i
								\nn\\
								&\quad+\!
								2\mu_{\max}^2
								\!\!\cdot\!
								\left\|
									T_c(\w_{c,i-1}) \!-\! T_c(\bar{w}_{c,i-1})
									\!-\!
									\mu_{\max}
									\!\cdot\!
									(p^T \otimes I_M)
									\z_{i-1}
								\right\|^2
								\nn\\
								&\qquad
								\cdot
								\|(p^T \otimes I_M) \bv_i\|^2
								\nn
	\end{align}
Taking the conditional expectation of both sides of the above expression given $\mF_{i-1}$ and recalling that $\Expt[\bv_i | \mF_{i-1}] = 0$ based on \eqref{Equ:Assumption:Randomness:MDS}, we get
	\begin{align}
		&\Expt[ \|\check{\w}_{c,i}\|^4 | \mc{F}_{i-1} ] 
								\nn\\
				&=				\Expt\Big\{\!
									\big\|
										T_c(\w_{c,i-1}) \!-\! T_c(\bar{w}_{c,i-1})
										\!-\!
										\mu_{\max}
										\!\cdot\!
										(p^T \!\!\otimes\! I_M)
										\z_{i-1}
									\big\|^4
								\Big| \mc{F}_{i-1}\!
								\Big\}
								\nn\\
								&\quad+\!
								\Expt\Big(
									\big\{
										\mu_{\max}^2
										\| (p^T \! \otimes \! I_M) \bv_i \|^2
										\!-\!
										2\mu_{\max}
										\big[
											T_c(\w_{c,i- 1}) 
											\!-\! 
											T_c(\bar{w}_{c,i-1})
											\nn\\
											&\qquad\quad
											\!-\!
											\mu_{\max}
											(p^T \!\otimes\! I_M)
											\z_{i-1}
										\big]
										(p^T \!\otimes\! I_M) \bv_i
									\big\}^2 
									\big|
									\mc{F}_{i-1}
								\Big)
								\nn\\
								&\quad+\!
								2\mu_{\max}^2
								\!\cdot\!
								\left\|
									T_c(\w_{c,i-1}) \!-\! T_c(\bar{w}_{c,i-1})
									\!-\!
									\mu_{\max}
									\!\cdot\!
									(p^T \!\otimes\! I_M)
									\z_{i-1}
								\right\|^2
								\nn\\
								&\qquad\quad
								\cdot
								\Expt\big[
									\|(p^T \otimes I_M) \bv_i\|^2	
									\big|
									\mc{F}_{i-1}
								\big]
								\nn\\
			&\overset{(a)}{\le}
								\big\|
									T_c(\w_{c,i-1}) - T_c(\bar{w}_{c,i-1})
									-
									\mu_{\max}
									\cdot
									(p^T \otimes I_M)
									\z_{i-1}
								\big\|^4
								\nn\\
								&\quad
								+
								2\mu_{\max}^4
								\cdot
								\Expt \| (p^T \otimes I_M) \bv_i \|^4
								\nn\\
								&\quad
								+\!
								8\mu_{\max}^2
								\!\cdot\!
								\big\|
									T_c(\w_{c,i-1}) \!-\! T_c(\bar{w}_{c,i-1})
									\!-\!
									\mu_{\max}
									\!\cdot\!
									(p^T \!\otimes\! I_M)
									\z_{i-1}
								\big\|^2
								\nn\\
								&\qquad\quad
								\cdot
								\Expt\big[
									\|(p^T \otimes I_M) \bv_i\|^2
									\big|
									\mc{F}_{i-1}
								\big]
								\nn\\
								&\quad+\!
								2\mu_{\max}^2
								\!\cdot\!
								\left\|
									T_c(\w_{c,i-1}) \!-\! T_c(\bar{w}_{c,i-1})
									\!-\!
									\mu_{\max}
									\!\cdot\!
									(p^T \!\otimes\! I_M)
									\z_{i-1}
								\right\|^2
								\nn\\
								&\qquad\quad
								\cdot
								\Expt\big[
									\|(p^T \otimes I_M) \bv_i\|^2	
									\big|
									\mc{F}_{i-1}
								\big]
								\nn\\
			&=					\big\|
									T_c(\w_{c,i-1}) - T_c(\bar{w}_{c,i-1})
									-
									\mu_{\max}
									\cdot
									(p^T \otimes I_M)
									\z_{i-1}
								\big\|^4
								\nn\\
								&\quad
								+
								2\mu_{\max}^4
								\cdot
								\Expt \big[
									\| (p^T \otimes I_M) \bv_i \|^4
									\big|
									\mc{F}_{i-1}
								\big]
								\nn\\
								&\quad
								\!+\!
								10\mu_{\max}^2
								\!\cdot\!
								\big\|
									T_c(\w_{c,i-1}) \!-\! T_c(\bar{w}_{c,i-1})
									\!-\!
									\mu_{\max}
									\!\cdot\!
									(p^T \!\otimes\! I_M)
									\z_{i-1}
								\big\|^2
								\nn\\
								&\qquad\quad
								\cdot\!
								\Expt\big[
									\|(p^T \!\otimes\! I_M) \bv_i\|^2
									\big|
									\mc{F}_{i-1}
								\big]
			\label{Equ:EP_wcicheck_bound_interm1}
	\end{align}
where step (a) uses the inequality $(x+y)^2 \le 2 x^2 + 2 y^2$. To proceed, we call upon the following bounds.
	\begin{lemma}[Useful bounds]
	\label{Lemma:UsefulBounds}
	The following bounds hold for arbitrary $i$:
		\begin{align}
			\big\|&
				T_c(\w_{c,i-1}) - T_c(\bar{w}_{c,i-1})
				-
				\mu_{\max}
				\cdot
				(p^T \otimes I_M)
				\z_{i-1}
			\big\|^4
									\nn\\
					&\le 			\gamma_c \cdot \| \check{\w}_{c,i-1} \|^4
									\nn\\
									&\quad
									+
									\frac{\mu_{\max}}
									{(\lambda_L - \frac{1}{2}\mu_{\max} \|p\|_1^2 \lambda_U^2)^3}	
									\cdot
									\|p\|_1^4 \cdot \lambda_U^4
									\cdot
									\big\|
										\bP[\mA_1^T \mU_L ]
									\big\|_{\infty}^4
									\nn\\
									&\qquad
									\cdot
									\one^T 
									\Pm[ \w_{e,i-1} ]
			\label{Equ:Lemma:UsefulBounds:Tc_z_4th}
									\\
			\Expt \big[&
				\left\|
					(p^T \otimes I_M) \bv_i
				\right\|^4
				\big|
				\mc{F}_{i-1}
			\big]					\nn\\
					&\le 			
									216\alpha_4 \|p\|_1^4
								\cdot
								\| \check{\w}_{c,i-1} \|^4
								\nn\\
								&\quad
								+
								216\alpha_4 \|p\|_1^4 
								\cdot
								\big\|
									\bP[ \mA_1^T \mU_L ]
								\big\|_{\infty}^4
								\cdot
								\one^T
								\cdot
								\Pm[ \w_{e,i-1} ]
								\nn\\
								&\quad+
								27\alpha_4\|p\|_1^4
								\cdot
								\|\tilde{w}_{c,0}\|^4
								+
								27\alpha_4 \cdot
								\|p\|_1^4
								\cdot
								\|w^o\|^4
								\nn\\
								&\quad
								+
								\sigma_{v4}^4
								\cdot
								\|p\|_1^4
			\label{Equ:Lemma:UsefulBounds:vi_4th}
									\\
			\big\|&
				T_c(\w_{c,i-1}) - T_c(\bar{w}_{c,i-1})
				-
				\mu_{\max}
				\cdot
				(p^T \otimes I_M)
				\z_{i-1}
			\big\|^2
									\nn\\
					&\le 			\gamma_c \cdot \| \check{\w}_{c,i-1} \|^2
									\nn\\
									&\quad
									+
									\frac{\mu_{\max}}
									{\lambda_L - \frac{1}{2}\mu_{\max} \|p\|_1^2 \lambda_U^2}	
									\cdot
									\|p\|_1^2 \cdot \lambda_U^2
									\cdot
									\big\|
										\bP[\mA_1^T \mU_L ]
									\big\|_{\infty}^2
									\nn\\
									&\qquad
									\cdot
									\one^T 
									P[ \w_{e,i-1} ]
			\label{Equ:Lemma:UsefulBounds:Tc_z_2nd}
									\\
			\Expt \big[&
				\left\|
					(p^T \otimes I_M) \bv_i
				\right\|^2
				\big|
				\mc{F}_{i-1}
			\big]					\nn\\
					&\le 			4 \alpha \|p\|_1^2
								\cdot
								P[ \check{\w}_{c,i-1} ]
								\nn\\
								&\quad
								+
								4\alpha
								\cdot
								\| \bP[ \mA_1^T \mU_L ] \|_{\infty}^2
								\cdot
								\|p\|_1^2
								\cdot
								\one^T
								P[\w_{e,i-1} ]
								\nn\\
								&\quad
								+\!
								4 \alpha \|\tilde{w}_{c,0}\|^2
								\!\cdot\!
								\|p\|_1^2
								\!+\!
								4\alpha \|p\|_1^2 \!\cdot\! \|w^o\|^2
								\!+\!
								\sigma_v^2 \!\cdot\! \|p\|_1^2
			\label{Equ:Lemma:UsefulBounds:vi_2nd}
		\end{align}
	\end{lemma}
	\begin{proof}
		See Appendix \ref{Appendix:Proof_UsefulBound}.
	\end{proof}
Substituting \eqref{Equ:Lemma:UsefulBounds:Tc_z_4th}--\eqref{Equ:Lemma:UsefulBounds:vi_2nd} into \eqref{Equ:EP_wcicheck_bound_interm1}, we obtain
	\begin{align}
		\Expt&[ \|\check{\w}_{c,i}\|^4 | \mc{F}_{i-1} ] 
									\nn\\
				&\preceq				\gamma_c \cdot \| \check{\w}_{c,i-1} \|^4
									\nn\\
									&\quad
									+
									\frac{\mu_{\max}}
									{(\lambda_L - \frac{1}{2}\mu_{\max} \|p\|_1^2 \lambda_U^2)^3}	
									\cdot
									\|p\|_1^4 \cdot \lambda_U^4
									\cdot
									\big\|
										\bP[\mA_1^T \mU_L ]
									\big\|_{\infty}^4
									\nn\\
									&\qquad
									\cdot
									\one^T 
									\Pm[ \w_{e,i-1} ]
									\nn\\
									&\quad
									+
									2\mu_{\max}^4 \cdot	
									\Big\{
										216 \alpha_4 \|p\|_1^4
										\cdot
										\| \check{\w}_{c,i-1} \|^4
										\nn\\
										&\qquad
										+
										216 \alpha_4 \|p\|_1^4
										\cdot
										\big\|
											\bP[\mA_1^T \mU_L]
										\big\|_{\infty}^4
										\cdot
										\one^T \cdot
										\Pm [ \w_{e,i-1} ]
										\nonumber\\
										&\qquad
										+ \!
										27\alpha_4 \|p\|_1^4 
										\!\cdot\!
										\|\tilde{w}_{c,0}\|^4
										\!+\!
										27 \alpha_4 \|p\|_1^4
										\!\cdot\!
										\| w^o \|^4
										\!+\!
										\sigma_{v4}^4
										\!\cdot\!
										\|p\|_1^4
									\Big\}
									\nn\\
									&\quad
									+
									10\mu_{\max}^2
									\cdot
									\Big\{
										\gamma_c \cdot \| \check{\w}_{c,i-1} \|^2
										\nn\\
										&\quad
										+
										\frac{\mu_{\max}}
										{\lambda_L - \frac{1}{2}\mu_{\max} \|p\|_1^2 \lambda_U^2}	
										\cdot
										\|p\|_1^2 \cdot \lambda_U^2
										\cdot
										\big\|
											\bP[\mA_1^T \mU_L ]
										\big\|_{\infty}^2
										\nn\\
										&\qquad
										\cdot
										\one^T 
										P[ \w_{e,i-1} ]
									\Big\}
									\nn\\
									&\quad
									\!\cdot\!
									\Big\{
										4 \alpha \|p\|_1^2
										\!\cdot\!
										P[ \check{\w}_{c,i-1} ]
										\nn\\
										&\quad
										+\!
										4\alpha
										\!\cdot\!
										\| \bP[ \mA_1^T \mU_L ] \|_{\infty}^2
										\!\cdot\!
										\|p\|_1^2
										\!\cdot\!
										\one^T
										P[\w_{e,i-1} ]
										\nn\\
										&\quad
										+\!
										4 \alpha \|p\|_1^2 
										(
											\|\tilde{w}_{c,0}\|^2
											\!+\!
											\|w^o\|^2
										)
										\!+\!
										\sigma_v^2 
										\!\cdot\! \|p\|_1^2
									\Big\}
			\label{Equ:EP_wcicheck_bound_interm2}
	\end{align}
We further call upon the following inequality to bound the last term in \eqref{Equ:EP_wcicheck_bound_interm2}:
	\begin{align}
		(a &\cdot x + b \cdot y ) ( c \cdot x + d \cdot y + e )
						\nn\\
				&=		ac \cdot x^2 + bd \cdot y^2 + (ad + bc) xy + ae \cdot x + be \cdot y
						\nn\\
				&\le 	
						ac \!\cdot\! x^2 \!+\! bd \!\cdot\! y^2 + (ad + bc) \frac{1}{2} (x^2 + y^2) 
						+ ae \cdot x + be \cdot y
						\nn\\
				&=		\left(
							ac \!+\! \frac{ad \!+\! bc}{2}
						\right)
						x^2
						\!+\!
						\left(
							bd \!+\! \frac{ad \!+\! bc}{2}
						\right)
						y^2
						\!+\!
						ae \!\cdot\! x
						\!+\!
						be \!\cdot \! y
						\nn
	\end{align}
Applying the above inequality to the last term in \eqref{Equ:EP_wcicheck_bound_interm2} with 
	\begin{align}
		a		&=		\gamma_c		
						\nn\\
		b		&=		\frac{\mu_{\max}}
						{\lambda_L - \frac{1}{2}\mu_{\max} \|p\|_1^2 \lambda_U^2}	
						\cdot
						\|p\|_1^2 \cdot \lambda_U^2
						\cdot
						\big\|
							\bP[\mA_1^T \mU_L ]
						\big\|_{\infty}^2
						\nn\\
		c		&=		4 \alpha \|p\|_1^2
						\nn\\
		d		&=		4 \alpha \|p\|_1^2
						\cdot
						\big\|
							\bP[\mA_1^T \mU_L]
						\big\|_{\infty}^2
						\nn\\
		e		&=		4 \alpha \|p\|_1^2 
						\cdot
						\|\tilde{w}_{c,0}\|^2
						+
						4 \alpha \|p\|_1^2
						\cdot
						\|w^o\|^2
						+
						\sigma_{v}^2
						\cdot
						\|p\|_1^2
						\nn\\
		x		&=		\| \check{\w}_{c,i-1} \|^2
						\nn\\
		y		&=		\one^T \cdot \Pm[ \w_{e,i-1} ]
				=		\| \w_{e,i-1} \|^4
						\nn
	\end{align}
we get
	\begin{align}
		&\Big\{
			\gamma_c \cdot \| \check{\w}_{c,i-1} \|^2
			\nn\\
			&\quad
			+
			\frac{\mu_{\max}}
			{\lambda_L - \frac{1}{2}\mu_{\max} \|p\|_1^2 \lambda_U^2}	
			\cdot
			\|p\|_1^2 \cdot \lambda_U^2
			\cdot
			\big\|
				\bP[\mA_1^T \mU_L ]
			\big\|_{\infty}^2
			\nn\\
			&\qquad
			\cdot
			\one^T 
			P[ \w_{e,i-1} ]
		\Big\}
		\nn\\
		&
		\!\times\!
		\Big\{
			4 \alpha \|p\|_1^2
			\!\cdot\!
			P[ \check{\w}_{c,i-1} ]
			\!+\!
			4\alpha
			\!\cdot\!
			\| \bP[ \mA_1^T \mU_L ] \|_{\infty}^2
			\!\cdot\!
			\|p\|_1^2
			\!\cdot\!
			\one^T
			P[\w_{e,i-1} ]
			\nn\\
			&\qquad
			+\!
			4 \alpha \|p\|_1^2 
			(
				\|\tilde{w}_{c,0}\|^2
				\!+\!
				\|w^o\|^2
			)
			\!+\!
			\sigma_v^2 
			\!\cdot\! \|p\|_1^2
		\Big\}
							\nn\\
				&\le 		\left(
								ac + \frac{ad + bc}{2}
							\right)
							\!\cdot\!
							\| \check{\w}_{c,i-1} \|^4
							\!+\!
							\left(
								bd \!+\! \frac{ad + bc}{2}
							\right)
							\!\cdot\!
							\| \w_{e,i-1} \|^4
							\nn\\
							&\qquad
							+
							ae \cdot \| \check{\w}_{c,i-1} \|^2
							+
							be \cdot \| \w_{e,i-1} \|^2
							\nn\\
				&\overset{(a)}{\le} 		
							\left(
								c + \frac{d + bc}{2}
							\right)
							\cdot
							\| \check{\w}_{c,i-1} \|^4
							+
							\left(
								bd + \frac{d + bc}{2}
							\right)
							\cdot
							\| \w_{e,i-1} \|^4
							\nn\\
							&\quad
							+
							e \cdot \| \check{\w}_{c,i-1} \|^2
							+
							be \cdot \| \w_{e,i-1} \|^2
							\nn\\
				&=			\Big(
								4 \alpha \|p\|_1^2
								+
								2
								\alpha \|p\|_1^2
								\cdot
								\| \bP[ \mA_1^T \mU_L ] \|_{\infty}^2
								\nn\\
								&\quad
								+\!								
								\frac{
									2\alpha \|p\|_1^2\cdot\mu_{\max}
								}
								{\lambda_L \!-\! \frac{1}{2}\mu_{\max} \|p\|_1^2 \lambda_U^2}
								\|p\|_1^2 \!\cdot\! \lambda_U^2 
								\!\cdot\!
								\| \bP[ \mA_1^T \mU_L ] \|_{\infty}^2
							\Big) \!\cdot\! \| \check{\w}_{c,i-1} \|^4	
							\nn\\
							&\quad
							+\!\!
							\Big(								
								\frac{
									4\alpha \|p\|_1^4  \lambda_U^2 
									\mu_{\max}
								}
								{\lambda_L \!\!-\!\! \frac{1}{2}\mu_{\max} \|p\|_1^2 \lambda_U^2}
								\!\cdot\!
								\| \bP[ \mA_1^T \mU_L ] \|_{\infty}^4
								\!+\!
								2
								\alpha \|p\|_1^2
								\!\cdot\!
								\| \bP[ \mA_1^T \mU_L ] \|_{\infty}^2
								\nn\\
								&\qquad
								+								
								\frac{ 2\alpha \|p\|_1^4 \lambda_U^2  \cdot \mu_{\max}}
								{\lambda_L-\frac{1}{2}\mu_{\max} \|p\|_1^2 \lambda_U^2}
								\cdot
								\| \bP[ \mA_1^T \mU_L ] \|_{\infty}^2
							\Big)
							\cdot
							\|\w_{e,i-1} \|^4		
							\nn\\
							&\quad
							+\!	
							\Big(
								4 \alpha \|p\|_1^2 
								\!\cdot\!
								\|\tilde{w}_{c,0}\|^2
								\!+\!
								4 \alpha \|p\|_1^2
								\!\cdot\!
								\|w^o\|^2
								\!+\!
								\sigma_{v}^2
								\!\cdot\!
								\|p\|_1^2
							\Big)
							\!\cdot\!
							\|\check{\w}_{c,i-1}\|^2					
							\nn\\
							&\quad
							+\!
							\frac{\|p\|_1^4 \lambda_U^2 \cdot \mu_{\max}}
							{\lambda_L \!-\! \frac{1}{2}\mu_{\max} \|p\|_1^2 \lambda_U^2}	
							\!\cdot\!
							\big\|
								\bP[\mA_1^T \mU_L ]
							\big\|_{\infty}^2
							\nn\\
							&\qquad
							\cdot\!
							\Big(\!
								4 \alpha
								\!\cdot\!
								\|\tilde{w}_{c,0}\|^2
								\!+\!
								4 \alpha
								\!\cdot\!
								\|w^o\|^2
								\!+\!
								\sigma_{v}^2
								\!\cdot\!
							\Big)
							\!\cdot\!
							\| \w_{e,i-1} \|^2
		\label{Equ:EP_wcicheck_bound_LastTerm}
	\end{align}
where inequality (a) is using $a = \gamma_c < 1$, which is guaranteed for sufficiently small step-sizes. Substituting \eqref{Equ:EP_wcicheck_bound_LastTerm} into \eqref{Equ:EP_wcicheck_bound_interm2}, we get
	\begin{align}
		\Expt&[ \|\check{\w}_{c,i}\|^4 | \mc{F}_{i-1} ] 
									\nn\\
				&\preceq				
									\gamma_c \cdot \| \check{\w}_{c,i-1} \|^4
									\nn\\
									&\quad
									+
									\frac{\mu_{\max}}
									{(\lambda_L - \frac{1}{2}\mu_{\max} \|p\|_1^2 \lambda_U^2)^3}	
									\cdot
									\|p\|_1^4 \cdot \lambda_U^4
									\cdot
									\big\|
										\bP[\mA_1^T \mU_L ]
									\big\|_{\infty}^4
									\nn\\
									&\qquad
									\cdot
									\one^T 
									\Pm[ \w_{e,i-1} ]
									\nn\\
									&\quad
									+
									2\mu_{\max}^4 \|p\|_1^4 \cdot	
									\Big\{
										216 \alpha_4 
										\cdot
										\| \check{\w}_{c,i-1} \|^4
										\nn\\
										&\qquad
										+\!
										216 \alpha_4
										\!\cdot\!
										\big\|
											\bP[\mA_1^T \mU_L]
										\big\|_{\infty}^4
										\!\cdot\!
										\one^T \!
										\Pm [ \w_{e,i-1} ]
										\!+ \!
										27\alpha_4
										\!\cdot\!
										\|\tilde{w}_{c,0}\|^4
										\nn\\
										&\qquad
										+
										27 \alpha_4
										\cdot
										\|w^o\|^4
										+
										\sigma_{v4}^4
									\Big\}
									\nn\\
									&\quad
									+
									10\mu_{\max}^2
									\cdot
									\Big\{
									%%%%%%
										\Big(
											4 \alpha \|p\|_1^2
											\!+\!
											2
											\alpha \|p\|_1^2
											\cdot
											\| \bP[ \mA_1^T \mU_L ] \|_{\infty}^2
											\nn\\
											&\qquad
											+\!								
											\frac{
												2\alpha \|p\|_1^4 \lambda_U^2\cdot\mu_{\max}
											}
											{\lambda_L-\frac{1}{2}\mu_{\max} \|p\|_1^2 \lambda_U^2} 
											\!\cdot\!
											\| \bP[ \mA_1^T \mU_L ] \|_{\infty}^2
										\Big) 
										\!\cdot\! 
										\| \check{\w}_{c,i-1} \|^4	
										\nn\\
										&\quad
										+
										\Big(								
											\frac{
												4\alpha \|p\|_1^4 \cdot \lambda_U^2 
												\cdot\mu_{\max}
											}
											{\lambda_L-\frac{1}{2}\mu_{\max} \|p\|_1^2 \lambda_U^2}
											\cdot
											\| \bP[ \mA_1^T \mU_L ] \|_{\infty}^4
											\nn\\
											&\qquad
											+
											2
											\alpha \|p\|_1^2
											\cdot
											\| \bP[ \mA_1^T \mU_L ] \|_{\infty}^2
											\nn\\
											&\qquad
											+								
											\frac{ 2\alpha \|p\|_1^4 \lambda_U^2  \cdot \mu_{\max}}
											{\lambda_L-\frac{1}{2}\mu_{\max} \|p\|_1^2 \lambda_U^2}
											\cdot
											\| \bP[ \mA_1^T \mU_L ] \|_{\infty}^2
										\Big)
										\cdot
										\|\w_{e,i-1} \|^4		
										\nn\\
										&\quad
										+	
										\Big(
											4 \alpha \|p\|_1^2 
											\cdot
											\|\tilde{w}_{c,0}\|^2
											+
											4 \alpha \|p\|_1^2
											\cdot
											\|w^o\|^2
											+
											\sigma_{v}^2
											\cdot
											\|p\|_1^2
										\Big)
										\nn\\
										&\qquad
										\cdot
										\|\check{\w}_{c,i-1}\|^2					
										\nn\\
										&\quad
										+\!
										\frac{\|p\|_1^4 \lambda_U^2 \cdot \mu_{\max}}
										{\lambda_L \!-\! \frac{1}{2}\mu_{\max} \|p\|_1^2 \lambda_U^2}	
										\!\cdot\!
										\big\|
											\bP[\mA_1^T \mU_L ]
										\big\|_{\infty}^2
										\nn\\
										&\qquad
										\cdot\!
										\Big(\!
											4 \alpha
											\!\cdot\!
											\|\tilde{w}_{c,0}\|^2
											\!+\!
											4 \alpha
											\!\cdot\!
											\|w^o\|^2
											\!+\!
											\sigma_{v}^2
											\!\cdot\!
										\Big)
										\!\cdot\!
										\| \w_{e,i-1} \|^2
									%%%%%%
									\Big\}
									\nn\\
			&\overset{(a)}{=}		
									\gamma_c 
									\cdot 
									\Pm[ \check{\w}_{c,i-1} ]
									\nn\\
									&\quad
									+
									\frac{\mu_{\max}}
									{(\lambda_L - \frac{1}{2}\mu_{\max} \|p\|_1^2 \lambda_U^2)^3}	
									\cdot
									\|p\|_1^4 \cdot \lambda_U^4
									\cdot
									\big\|
										\bP[\mA_1^T \mU_L ]
									\big\|_{\infty}^4
									\nn\\
									&\qquad
									\cdot
									\one^T 
									\Pm[ \w_{e,i-1} ]
									\nn\\
									&\quad
									+
									2\mu_{\max}^4 \|p\|_1^4 \cdot	
									\Big\{
										216 \alpha_4 
										\cdot
										\Pm[ \check{\w}_{c,i-1} ]
										\nn\\
										&\qquad
										+
										216 \alpha_4
										\cdot
										\big\|
											\bP[\mA_1^T \mU_L]
										\big\|_{\infty}^4
										\cdot
										\one^T
										\Pm [ \w_{e,i-1} ]
										\nonumber\\
										&\qquad
										+ 
										27\alpha_4
										\cdot
										\|\tilde{w}_{c,0}\|^4
										+
										27 \alpha_4
										\cdot
										\|w^o\|^4
										+
										\sigma_{v4}^4
									\Big\}
									\nn\\
									&\quad
									+
									10\mu_{\max}^2
									\!\cdot\!
									\Big\{\!
									%%%%%%
										\Big(\!
											4 \alpha \|p\|_1^2
											\!+\!
											2
											\alpha \|p\|_1^2
											\cdot
											\| \bP[ \mA_1^T \mU_L ] \|_{\infty}^2
											\nn\\
											&\qquad
											+\!								
											\frac{
												2\alpha \|p\|_1^4 \lambda_U^2\cdot\mu_{\max}
											}
											{\lambda_L-\frac{1}{2}\mu_{\max} \|p\|_1^2 \lambda_U^2} 
											\!\cdot\!
											\| \bP[ \mA_1^T \mU_L ] \|_{\infty}^2
											\!
										\Big) 
										\!\cdot\! 
										\Pm[ \check{\w}_{c,i-1} ]
										\nn\\
										&\quad
										+
										%------
										\Big(								
											\frac{
												4\alpha \|p\|_1^4 \cdot \lambda_U^2 
												\cdot\mu_{\max}
											}
											{\lambda_L-\frac{1}{2}\mu_{\max} \|p\|_1^2 \lambda_U^2}
											\cdot
											\| \bP[ \mA_1^T \mU_L ] \|_{\infty}^4
											\nn\\
											&\qquad
											+
											2
											\alpha \|p\|_1^2
											\cdot
											\| \bP[ \mA_1^T \mU_L ] \|_{\infty}^2
											\nn\\
											&\qquad
											+								
											\frac{ 2\alpha \|p\|_1^4 \lambda_U^2  \cdot \mu_{\max}}
											{\lambda_L-\frac{1}{2}\mu_{\max} \|p\|_1^2 \lambda_U^2}
											\cdot
											\| \bP[ \mA_1^T \mU_L ] \|_{\infty}^2
										\Big)
										\nn\\
										&\qquad
										\cdot
										\one^T \Pm[ \w_{e,i-1} ]
										\nn\\
										&\quad
										+	
										%------
										\Big(
											4 \alpha \|p\|_1^2 
											\cdot
											\|\tilde{w}_{c,0}\|^2
											+
											4 \alpha \|p\|_1^2
											\cdot
											\|w^o\|^2
											+
											\sigma_{v}^2
											\cdot
											\|p\|_1^2
										\Big)
										\nn\\
										&\qquad
										\cdot
										P[ \check{\w}_{c,i-1} ]
										\nn\\
										&\quad
										+\!
										%-------
										\frac{\|p\|_1^4 \lambda_U^2 \cdot \mu_{\max}}
										{\lambda_L \!-\! \frac{1}{2}\mu_{\max} \|p\|_1^2 \lambda_U^2}	
										\!\cdot\!
										\big\|
											\bP[\mA_1^T \mU_L ]
										\big\|_{\infty}^2
										\nn\\
										&\qquad
										\cdot\!
										\Big(\!
											4 \alpha
											\!\cdot\!
											\|\tilde{w}_{c,0}\|^2
											\!+\!
											4 \alpha
											\!\cdot\!
											\|w^o\|^2
											\!+\!
											\sigma_{v}^2
											\!\cdot\!
										\Big)
										\!\cdot\!
										\one^T P[ \w_{e,i-1} ]
									%%%%%%
									\Big\}
									\nn\\
			&=						\bigg\{
										\gamma_c
										+
										432\mu_{\max}^4
										\alpha_4 \|p\|_1^4
										\nn\\
										&\qquad
										+
										10\mu_{\max}^2
										\cdot
										\Big(\!
											4 \alpha \|p\|_1^2
											\!+\!
											2
											\alpha \|p\|_1^2
											\cdot
											\| \bP[ \mA_1^T \mU_L ] \|_{\infty}^2
											\nn\\
											&\qquad
											+\!								
											\frac{
												2\alpha \|p\|_1^4 \lambda_U^2\cdot\mu_{\max}
											}
											{\lambda_L-\frac{1}{2}\mu_{\max} \|p\|_1^2 \lambda_U^2} 
											\!\cdot\!
											\| \bP[ \mA_1^T \mU_L ] \|_{\infty}^2
											\!
										\Big)
									\bigg\} \!\cdot\!
									\Pm[ \check{\w}_{c,i-1} ]
									\nn\\
									&\quad
									+
									\bigg\{
										\mu_{\max} \cdot
										\frac{
											\|p\|_1^4 \cdot \lambda_U^4
											\cdot
											\left\|
												\bP[\mA_1^T \mU_L ]
											\right\|_{\infty}^4
										}
										{
											(\lambda_L - \frac{1}{2}\mu_{\max} \|p\|_1^2 \lambda_U^2)^3
										}
										\nn\\
										&\qquad	
										+
										432\mu_{\max}^4
										 \alpha_4 \|p\|_1^4
										\cdot
										\| \bP[ \mA_1^T \mU_L ] \|_{\infty}^4
										\nn\\
										&\qquad
										+
										10\mu_{\max}^2
										\cdot
										\Big(								
											\frac{
												4\alpha \|p\|_1^4 \cdot \lambda_U^2 
												\cdot\mu_{\max}
											}
											{\lambda_L-\frac{1}{2}\mu_{\max} \|p\|_1^2 \lambda_U^2}
											\cdot
											\| \bP[ \mA_1^T \mU_L ] \|_{\infty}^4
											\nn\\
											&\qquad\quad
											+
											2
											\alpha \|p\|_1^2
											\cdot
											\| \bP[ \mA_1^T \mU_L ] \|_{\infty}^2
											\nn\\
											&\qquad\quad
											+								
											\frac{ 2\alpha \|p\|_1^4 \lambda_U^2  \cdot \mu_{\max}}
											{\lambda_L-\frac{1}{2}\mu_{\max} 
											\|p\|_1^2 \lambda_U^2}
											\cdot
											\| \bP[ \mA_1^T \mU_L ] \|_{\infty}^2
										\Big)
									\bigg\}
									\nn\\
									&\qquad
									\cdot
									\one^T
									\cdot
									\Pm[\w_{e,i-1}]
									\nn\\
									&\quad
									+\!	
									10\mu_{\max}^2
									\cdot
									\Big(
											4 \alpha \|p\|_1^2 
											\!\cdot\!
											\|\tilde{w}_{c,0}\|^2
											\!+\!
											4 \alpha \|p\|_1^2
											\!\cdot\!
											\|w^o\|^2
											\!+\!
											\sigma_{v}^2
											\!\cdot\!
											\|p\|_1^2
									\Big)
									\nn\\
									&\qquad
									\cdot
									P[\check{\w}_{c,i-1}]					
									\nn\\
									&\quad
									+
									\frac{10\|p\|_1^4 \lambda_U^2 \cdot \mu_{\max}^3}
									{\lambda_L \!-\! \frac{1}{2}\mu_{\max} \|p\|_1^2 \lambda_U^2}	
									\!\cdot\!
									\big\|
										\bP[\mA_1^T \mU_L ]
									\big\|_{\infty}^2
									\nn\\
									&\qquad
									\cdot\!
									\Big(\!
										4 \alpha
										\!\cdot\!
										\|\tilde{w}_{c,0}\|^2
										\!+\!
										4 \alpha
										\!\cdot\!
										\|w^o\|^2
										\!+\!
										\sigma_{v}^2
										\!\cdot\!
									\Big)
									\cdot
									 \one^T \cdot P[\w_{e,i-1}]
									\nn\\
									&\quad
									+
									2\mu_{\max}^4
									\cdot
									\|p\|_1^4
									\cdot
									\big(
										27 \alpha_4 
										\cdot 
										(
											\|\tilde{w}_{c,0}\|^4
											+
											\|w^o\|^4
										)
										+
										\sigma_{v4}^4
									\big)
									\nn\\
				&=					\bigg\{
										\gamma_c
										+
										\mu_{\max}^4
										\cdot
										432\alpha_4 \|p\|_1^4
										\nn\\
										&\qquad
										+
										\mu_{\max}^2
										\cdot
										20
										\alpha \|p\|_1^2
										\nn\\
										&\qquad\quad
										\cdot
										\Big(
											2
											+
											\| \bP[ \mA_1^T \mU_L ] \|_{\infty}^2
											\cdot
											\frac{
												\lambda_L+\frac{1}{2}\mu_{\max} \|p\|_1^2 \lambda_U^2
											}
											{
												\lambda_L-\frac{1}{2}\mu_{\max} \|p\|_1^2 \lambda_U^2
											}
										\Big) 
									\bigg\}
									\nn\\
									&\qquad
									\cdot
									\Pm[ \check{\w}_{c,i-1} ]
									\nn\\
									&\quad
									%--------------
									+
									\bigg\{
										\mu_{\max} \cdot
										\frac{
											\|p\|_1^4 \cdot \lambda_U^4
											\cdot
											\left\|
												\bP[\mA_1^T \mU_L ]
											\right\|_{\infty}^4
										}
										{
											(\lambda_L - \frac{1}{2}\mu_{\max} \|p\|_1^2 \lambda_U^2)^3
										}	
										\nn\\
										&\qquad
										+
										432 \mu_{\max}^4
										\alpha_4 \|p\|_1^4
										\cdot
										\| \bP[ \mA_1^T \mU_L ] \|_{\infty}^4
										\nn\\
										&\qquad
										+
										20\mu_{\max}^2
										\alpha \|p\|_1^2
										\cdot
										\| \bP[ \mA_1^T \mU_L ] \|_{\infty}^2
										\nn\\
										&\qquad\quad
										\cdot
										\Big(
											\mu_{\max} \cdot
											\frac{
													2 \|p\|_1^2 \cdot \lambda_U^2 
													\cdot
													\| \bP[ \mA_1^T \mU_L ] \|_{\infty}^2
											}
											{
												\lambda_L
												-
												\frac{1}{2}\mu_{\max} \|p\|_1^2 \lambda_U^2
											}	
											\nn\\
											&\qquad\qquad
											+\!
											\frac{
												\lambda_L 
												\!+\! 
												\frac{1}{2}\mu_{\max} \|p\|_1^2 \lambda_U^2
											}
											{
												\lambda_L \!-\! \frac{1}{2}\mu_{\max} 
												\|p\|_1^2 \lambda_U^2
											}
										\Big)
									\bigg\}
									\!\cdot\!
									\one^T
									\!\cdot\!
									\Pm[\w_{e,i-1}]
									\nn\\
									&\quad
									%----------
									+	
									10\mu_{\max}^2
									\cdot
									\Big(
											4 \alpha \|p\|_1^2 
											\cdot
											\|\tilde{w}_{c,0}\|^2
											+
											4 \alpha \|p\|_1^2
											\cdot
											\|w^o\|^2
											\nn\\
											&\qquad
											+
											\sigma_{v}^2
											\cdot
											\|p\|_1^2
									\Big)
									\cdot
									P[\check{\w}_{c,i-1}]					
									\nn\\
									&\quad
									%----------
									+
									\frac{10\|p\|_1^4 \lambda_U^2 \cdot \mu_{\max}^3}
									{\lambda_L \!-\! \frac{1}{2}\mu_{\max} \|p\|_1^2 \lambda_U^2}	
									\!\cdot\!
									\big\|
										\bP[\mA_1^T \mU_L ]
									\big\|_{\infty}^2
									\nn\\
									&\qquad
									\cdot\!
									\Big(\!
										4 \alpha
										\!\cdot\!
										\|\tilde{w}_{c,0}\|^2
										\!+\!
										4 \alpha
										\!\cdot\!
										\|w^o\|^2
										\!+\!
										\sigma_{v}^2
										\!\cdot\!
									\Big)
									\cdot
									 \one^T \cdot P[\w_{e,i-1}]
									\nn\\
									&\quad
									+
									%----------
									2\mu_{\max}^4
									\!\cdot\!
									\|p\|_1^4
									\!\cdot\!
									\big(
										27 \alpha_4 
										\!\cdot\!
										(
											\|\tilde{w}_{c,0}\|^4
											\!+\!
											\|w^o\|^4
										)
										\!+\!
										\sigma_{v4}^4
									\big)
			\label{Equ:EP_wcicheck_bound_interm3}
	\end{align}
where step (a) is using the following relations:
	\begin{align}
		\|\w_{e,i-1}\|^4 	&= 	\one^T \cdot \Pm[\w_{e,i-1}]
								\nn\\
		\|\w_{e,i-1}\|^2 	&= 	\one^T \cdot P[\w_{e,i-1}]
								\nn\\
		\|\check{\w}_{c,i-1}\|^4
							&=	\Pm[\check{\w}_{c,i-1}]
								\nn\\
		\|\check{\w}_{c,i-1}\|^2
							&=	P[\check{\w}_{c,i-1}]
								\nn
	\end{align}
Using the notation defined in \eqref{Equ:StabBound4thOrder:fcc_def}--\eqref{Equ:StabBound4thOrder:xic_def} and taking expectations of both sides of \eqref{Equ:EP_wcicheck_bound_interm3} with respect to $\mF_{i-1}$, we obtain
	\begin{align}
		\Expt \Pm[ \check{\w}_{c,i} ]
						&\preceq			f_{cc}(\mu_{\max}) 
										\cdot 
										\Expt
										\Pm[ \check{\w}_{c,i-1} ]
										\nn\\
										&\quad
										+
										f_{ce}(\mu_{\max})
										\cdot
										\one^T
										\cdot
										\Expt
										\Pm[ \w_{e,i-1} ]
										\nn\\
										&\quad
										+
										h_{cc}(\mu_{\max})
										\cdot
										\Expt 
										P[ \check{\w}_{c,i-1} ]
										\nn\\
										&\quad
										+
										h_{ce}(\mu_{\max})
										\cdot
										\one^T
										\cdot
										\Expt
										P[ \w_{e,i-1} ]
										\nn\\
										&\quad
										+
										\mu_{\max}^4
										\cdot
										b_{v_4,c}
		\label{Equ:EP_wcicheck_bound_final}
	\end{align}

\subsection{Recursion for the 4th order moment of $\w_{e,i}$}
\label{Appendix:RecursionEPwei}

We now derive an inequality recursion for $\Expt\|\w_{e,i}\|^4$. First, applying $\Pm[\cdot]$ operator to both sides of \eqref{Equ:Appendix:wei_check_recursion_Part2}, we get
	\begin{align}
		&\Pm[ \w_{e,i} ]
								\nn\\
						&=		\Pm\Big[
									\mD_{N-1} \w_{e,i-1} 
									\!-\!
									\mU_R \mA_2^T \mM
									\big(
										s(\one \!\otimes\! \w_{e,i-1})
										\!+\!
										\z_{i-1}
										\!+\!
										\bv_{i}
									\big)
								\Big]
								\nn\\
						&\overset{(a)}{\preceq}
								\Gamma_{e,4} \cdot
								\Pm[ \w_{e,i-1} ]
								\nn\\
								&\quad
								+\!
								\frac{8}{(1\!\!-\!\!|\lambda_2(A)|)^3}
								\!\cdot\!
								\Pm
								\big[
									\mU_R \mA_2^T \mM
									\big(
										s( \one \!\otimes\! \w_{c,i-1} )
										\!+\!
										\z_{i-1}
										\!+\!
										\bv_i
									\big)
								\big]
								\nn\\
						&\overset{(b)}{\preceq}
								\Gamma_{e,4} \!\cdot\!
								\Pm[ \w_{e,i-1} ]
								\nn\\
								&\quad
								+
								\frac{8}{(1\!\!-\!\! |\lambda_2(A)|)^3}
								\!\cdot\!
								\left\|
									\bP [ \mU_R \mA_2^T \mM ]
								\right\|_{\infty}^4
								\!\cdot\!
								\one 
								\one^T
								\nn\\
								&\qquad
								\cdot\!
								\Pm
								\big[
										s( \one \otimes \w_{c,i-1} )
										\!+\!
										\z_{i-1}
										\!+\!
										\bv_i
								\big]
								\nn\\
						&\overset{(c)}{\preceq}
								\Gamma_{e,4} \!\cdot\!
								\Pm[ \w_{e,i-1} ]
								\nn\\
								&\quad
								+\!
								\mu_{\max}^4
								\!\cdot\!
								\frac{8}{(1\!-\!|\lambda_2(A)|)^3}
								\!\cdot\!
								\left\|
									\bP[
										\mU_R \mA_2^T
									]
								\right\|_{\infty}^4
								\!\cdot\!
								\one 
								\one^T
								\nn\\
								&\qquad
								\cdot\!
								\Pm
								\big[
										s( \one \otimes \w_{c,i-1} )
										\!+\!
										\z_{i-1}
										\!+\!
										\bv_i
								\big]
								\nn\\
						&=
								\Gamma_{e,4} \cdot
								\Pm[ \w_{e,i-1} ]
								\nn\\
								&\quad
								+
								\frac{
									8\mu_{\max}^4
									\left\|
										\bP[ \mU_R \mA_2^T ]
									\right\|_{\infty}^4
								}
								{(1-|\lambda_2(A)|)^3}								
								\!\cdot\!
								\one 
								\one^T
								\nn\\
								&\qquad
								\cdot
								\Pm
								\Big[
										\frac{1}{3}
										\!\cdot\!
										3
										s( \one \otimes \w_{c,i-1} )
										\!+\!
										\frac{1}{3}
										\!\cdot\!
										3
										\z_{i-1}
										\!+\!
										\frac{1}{3}
										\cdot
										3
										\bv_i
								\Big]		
								\nn\\
						&\overset{(d)}{\preceq}
								\Gamma_{e,4} \cdot
								\Pm[ \w_{e,i-1} ]
								\nn\\
								&\quad
								+
								\frac{
									8\mu_{\max}^4 
									\left|
										\bP[ \mU_R \mA_2^T ]
									\right\|_{\infty}^4
								}{(1-|\lambda_2(A)|)^3}
								\cdot								
								\one 
								\one^T
								\nn\\
								&\qquad
								\times
								\Big\{
								\frac{1}{3}
								\cdot
								\Pm
								\big[			
										3
										s( \one \otimes \w_{c,i-1} )
								\big]
										+
								\frac{1}{3}
								\cdot
								\Pm
								\big[
										3
										\z_{i-1}
								\big]
								\nn\\
								&\qquad\quad
										+
								\frac{1}{3}
								\cdot
								\Pm
								\big[
										3
										\bv_i
								\big]
								\Big\}	
								\nn\\
						&\overset{(e)}{=}
								\Gamma_{e,4} \cdot
								\Pm[ \w_{e,i-1} ]
								\nn\\
								&\quad
								+
								\mu_{\max}^4
								\cdot
								\frac{8}{(1-|\lambda_2(A)|)^3}
								\cdot
								\big\|
									\bP[ \mU_R \mA_2^T ]
								\big\|_{\infty}^4
								\cdot
								\one 
								\one^T
								\nn\\
								&\qquad
								\times
								\big\{
								27
								\cdot
								\Pm
								[		
										s( \one \otimes \w_{c,i-1} )
								]
										+
								27
								\cdot
								\Pm
								[
										\z_{i-1}
								]
								\nn\\
								&\qquad\quad
										+
								27
								\cdot
								\Pm
								[
										\bv_i
								]
								\big\}		
								\nn
	\end{align}
where step (a) uses \eqref{Equ:VarMomentPropert:StableKronJordanOperator}, step (b) uses \eqref{Equ:VarMomentPropt:Linear_ub}, step (c) uses  the sub-multiplicative property \eqref{P1-Equ:Properties:PX_bar_SubMult} from Part I\cite{chen2013learningPart1} and the sub-multiplicative property of matrix norms:
	\begin{align}
		&\bP[\mU_R \mA_2^T \mM]		
					\preceq
								\bP[ \mU_R ] \cdot \bP[ \mA_2^T ] \cdot \bP[ \mM ]
								\nn\\
		&\Rightarrow 
		\left\| \bP[\mU_R \mA_2^T \mM] \right\|_{\infty}
					\le 
								\left\| \bP[ \mU_R ] \right\|_{\infty}
								\!\cdot \!
								\left\| \bP[ \mA_2^T ] \right\|_{\infty}
								\!\cdot \!
								\left\| \bP[ \mM ] \right\|_{\infty}
								\nn\\
		&\Rightarrow 
		\left\| \bP[\mU_R \mA_2^T \mM] \right\|_{\infty}
					\le 
								\mu_{\max}
								\!\cdot\!
								\left\| \bP[ \mU_R ] \right\|_{\infty}
								\!\!\cdot\! 
								\left\| \bP[ \mA_2^T ] \right\|_{\infty}  \!\!
								\nn
	\end{align}
step (d) uses the convex property \eqref{Equ:VarMomentPropt:Convexity}, and step (e) uses the scaling property in Lemma \ref{Lemma:VarMomentPropt}. Applying the expectation operator to both sides of the above inequality conditioned on $\mF_{i-1}$, we obtain
	\begin{align}
		\Expt&\big[
			\Pm[ \w_{e,i} ]
			\big|
			\mF_{i-1}
		\big]					\nn\\
					&\preceq
								\Gamma_{e,4} \cdot
								\Pm[ \w_{e,i-1} ]
								\nn\\
								&\quad+
								\frac{
									8\mu_{\max}^4 \left\|
										\bP[ \mU_R \mA_2^T ]
									\right\|_{\infty}^4
								}
								{(1-|\lambda_2(A)|)^3}
								\!\cdot\!
								\one 
								\one^T
								\!\cdot\!
								\Big\{
								27
								\!\cdot\!
								\Pm
								[		
										s( \one \!\otimes\! \w_{c,i-1} )
								]
								\nn\\
								&\quad
								\!+\!
								27
								\!\cdot\!
								\Pm
								[
										\z_{i-1}
								]
								\!+\!
								27
								\!\cdot\!
								\Expt\{
								\Pm
								[
										\bv_i
								]
								|
								\mF_{i-1}
								\}
								\Big\}
								\nn
	\end{align}
In the above expression, we are using the fact that $\w_{c,i-1}$ and $\z_{i-1}$ are determined by the history up to time $i-1$. Therefore, given $\mF_{i-1}$, these two quantities are deterministic and known so that
	\begin{align}
		\Expt \big\{ \Pm [ s( \one \!\otimes\! \w_{c,i-1} )]  \big| \mF_{i-1} \big\}
				&=		\Pm [ s( \one \! \otimes \! \w_{c,i-1} )]
						\nn\\
		\Expt \big\{ \Pm [ \z_{i-1} ] \big| \mF_{i-1} \big\}
				&=		\Pm [ \z_{i-1} ]		
						\nn
	\end{align}
Substituting \eqref{Equ:Lemma:BoundsPerturbation4th:P_z}--\eqref{Equ:Lemma:BoundsPerturbation4th:P_v} into the right-hand side of the above inequality, we get
	\begin{align}
		\Expt&\big[
			\Pm[ \w_{e,i} ]
			\big|
			\mF_{i-1}
		\big]					\nn\\
					&\preceq
								\Gamma_{e,4} \cdot
								\Pm[ \w_{e,i-1} ]
								\nn\\
								&\quad+
								\mu_{\max}^4
								\cdot
								\frac{8}{(1-|\lambda_2(A)|)^3}
								\cdot
								\big\|
									\bP[ \mU_R \mA_2^T ]
								\big\|_{\infty}^4
								\cdot
								\one 
								\one^T
								\nn\\
								&\qquad
								\cdot
								\bigg\{
									27
									\!\cdot\!
									\Big[
										27\lambda_U^4 \!\cdot\!
										\|\check{\bm{w}}_{c,i-1}\|^4
										\!\cdot\! \mathds{1}
										\!+\!
										27\lambda_U^4 \|\tilde{w}_{c,0}\|^4
										\!\cdot\!
										\mathds{1}
										\!+\!
										27 \!\cdot \!
										g_4^o
									\Big]
									\nn\\
									&\qquad\quad
									+
									27
									\cdot
									\Big[
										\lambda_U^4
											\cdot
											\left\|
												\bP[\mc{A}_1^T \mc{U}_L]
											\right\|_{\infty}^4
											\cdot
											\mathds{1}\mathds{1}^T
											\cdot
											\Pm[\bm{w}_{e,i-1}]
									\Big]
									\nn\\
									&\qquad\quad
									+
									27
									\cdot
									\Big[
										216 \alpha_4 \cdot \mathds{1}
										\cdot
										\Pm[\check{\bm{w}}_{c,i-1}]
										\nn\\
										&\qquad\qquad\qquad
										+
										216
										\alpha_4
										\cdot
										\left\|
											\bP[\mc{A}_1^T\mc{U}_L]
										\right\|_{\infty}^4
										\!\cdot\!
										\mathds{1}\mathds{1}^T
										\!\cdot\!
										\Pm[\bm{w}_{e,i\!-\!1}]
										\nonumber\\		
										&\qquad\qquad\qquad
										+
										27
										\alpha_4
										\cdot
										(
											\|\tilde{w}_{c,0}\|^4
											+
											\|w^o\|^4
										) 
										\cdot
										\one
										+
										\sigma_{v4}^4
										\cdot
										\mathds{1}	
									\Big]
								\bigg\}
								\nn\\
					&=			\Big[
									\Gamma_{e,4}
									+
									\mu_{\max}^4
									\cdot
									\frac{216N \cdot (\lambda_U^4 + 216 \alpha_4)}
									{(1-|\lambda_2(A)|)^3}
									\cdot
									\| \bP[\mA_1^T \mU_L] \|_{\infty}^4
									\nn\\
									&\qquad\qquad
									\cdot
									\| \bP[ \mU_R \mA_2^T ] \|_{\infty}^4
									\cdot
									\one \one^T
								\Big]
								\cdot
								\Pm[ \w_{e,i-1} ]
								\nn\\
								&\quad+\!
								\mu_{\max}^4
								\!\cdot\!
								\frac{5832N \cdot (\lambda_U^4 \!\!+\!\! 8\alpha_4) }
								{(1 \!\!-\! \!|\lambda_2(A)|)^3}
								\| \bP[ \mU_R \mA_2^T ] \|_{\infty}^4
								\!\cdot\!
								\Pm[ \w_{c,i-1} ]
								\!\cdot \! \one
								\nn\\
								&\quad
								+\!
								\mu_{\max}^4
								\cdot
								\frac{
									216 \!\cdot \!
									\| \bP[ \mU_R \mA_2^T ] \|_{\infty}^4
								}{(1\!-\!|\lambda_2(A)|)^3}
								\cdot
								\Big\{
									27
									\big[
										(\lambda_U^4 \!+\! \alpha_4)
										\!\cdot\!
										\|\tilde{w}_{c,0}\|^4
										\!\cdot\!
										N
										\nn\\
										&\qquad
										+\!
										\one^T g_4^o
										\!+\!
										\alpha_4
										\!\cdot\!
										N \| w^o \|^4
									\big]
									\!+\!
									\sigma_{v4}^4 \cdot N
								\Big\}
								\cdot
								\one
								\nn\\
				&\preceq			\Big[
									\Gamma_{e,4}
									+
									\mu_{\max}^4
									\cdot
									\frac{216N \cdot (\lambda_U^4 + 216 \alpha_4)}
									{(1-|\lambda_2(A)|)^3}
									\cdot
									\| \bP[\mA_1^T \mU_L] \|_{\infty}^4
									\nn\\
									&\qquad\qquad
									\cdot
									\| \bP[ \mU_R \mA_2^T ] \|_{\infty}^4
									\cdot
									\one \one^T
								\Big]
								\cdot
								\Pm[ \w_{e,i-1} ]
								\nn\\
								&\quad+
								\mu_{\max}^4
								\cdot
								\frac{5832N \cdot (\lambda_U^4+8\alpha_4) }
								{(1 \!-\! |\lambda_2(A)|)^3}
								\| \bP[ \mU_R \mA_2^T ] \|_{\infty}^4
								\nn\\
								&\qquad
								\cdot
								\Pm[ \w_{c,i-1} ]
								\cdot
								\one
								\nn\\
								&\quad
								+
								\mu_{\max}^4
								\cdot
								\frac{216N \cdot
								\| \bP[ \mU_R \mA_2^T ] \|_{\infty}^4}
								{(1\!-\!|\lambda_2(A)|)^3}
								\!\cdot\!
								\Big\{
									27
									\big[
										(\lambda_U^4 + \alpha_4)
										\!\cdot\!
										\|\tilde{w}_{c,0}\|^4
										\nn\\
										&\qquad
										+\!
										\|g_4^o\|_{\infty}
										\!+\!
										\alpha_4
										\!\cdot\!
										\|w^o\|^4
									\big]
									+
									\sigma_{v4}^4
								\Big\}
								\cdot
								\one
		\label{Equ:EP_wei_bound_interm1}
	\end{align}
where the last step uses $\one^T g_4^o \le | \one^T g_4^o | \le N \| g_4^o\|_{\infty}$. Using the notation defined in \eqref{Equ:StabBound4thOrder:Fee_def}--\eqref{Equ:StabBound4thOrder:fec_def} and applying the expectation operator to both sides of \eqref{Equ:EP_wei_bound_interm1} with respect to $\mF_{i-1}$, we arrive at
	\begin{align}
		\Expt \Pm[ \w_{e,i} ]
					&\preceq
									F_{ee}(\mu_{\max}) \cdot \Expt \Pm[\w_{e,i-1}]
									\nn\\
									&\quad
									+
									f_{ec}(\mu_{\max}) \cdot \one \cdot \Expt \Pm[\check{\w}_{c,i-1}]
									\nn\\
									&\quad
									+
									\mu_{\max}^4
									\cdot
									b_{v_4,e}
									\cdot
									\one
		\label{Equ:EP_wei_bound_final}
	\end{align}

\section{Proof of Lemma \ref{Lemma:BoundsPerturbation_4thMoment}}
\label{Appendix:Proof_BoundsPerturbation4th}

First, we establish the bound for $P[\bm{z}_{i-1}]$ in 
\eqref{Equ:Lemma:BoundsPerturbation4th:P_z}. To begin with, recall the following two relations from \eqref{P1-Equ:DistProc:relation_phi_w_wprime}
and \eqref{P1-Equ:DistProc:invTF_wi_wi_prime} in Part I\cite{chen2013learningPart1}:  
	\begin{align}
		\bm{\phi}_{i} 	&= 	
							\mA_1^T \w_{i}
		\label{Equ:Appendix:phi_A1_w}
							\\
		\w_i			&=
							\one \otimes \w_{c,,i} + (U_L \otimes I_M) \w_{e,i}
		\label{Equ:Appendix:wi_wci_we}
	\end{align}
By the definition of $\bm{z}_{i-1}$ in \eqref{Equ:Def:z_i1}, we get:
	\begin{align}
		\Pm[\bm{z}_{i-1}]	&=
									\Pm[ s(\bm{\phi}_{i-1}) - s(\one \otimes \w_{c,i-1}) ]
									\nn\\
						&\overset{(a)}{=}
									\Pm[ s( \mA_1^T \w_{i-1} ) - s(\one \otimes \w_{c,i-1}) ]
									\nn\\
						&\overset{(b)}{=}
									\Pm\big[
										s\left(
											\mathds{1} \otimes \bm{w}_{c,i-1}
											\!+\!
											(A_1^T U_L \otimes I_M)
											\bm{w}_{e,i-1}
										\right)
										\nn\\
										&\qquad\qquad
										-
										s(\mathds{1}\otimes \bm{w}_{c,i-1})
									\big]
									\nonumber\\
						&\overset{(c)}{\preceq}
									\lambda_U^4
									\cdot
									\Pm\left[
										(A_1^T U_L \otimes I_M) 
										\bm{w}_{e,i-1}
									\right]
									\nonumber\\
						&\overset{(d)}{\preceq}
									\lambda_U^4
									\!\cdot\!
									\left\|
										\bP[\mc{A}_1^T \mc{U}_L]
									\right\|_{\infty}^4
									\!\cdot\!
									\mathds{1}\mathds{1}^T
									\!\cdot\!
									\Pm[\bm{w}_{e,i-1}]
									\nn
	\end{align}
where step (a) substitutes \eqref{Equ:Appendix:phi_A1_w}, step (b) substitutes \eqref{Equ:Appendix:wi_wci_we},
step (c) uses the variance relation \eqref{Equ:VarMomentPropt:Update}, and step (d) uses property \eqref{Equ:VarMomentPropt:Linear_ub}.

Next, we prove the bound for $\Pm[s(\mathds{1}\otimes \bm{w}_{c,i-1})]$.
It holds that
	\begin{align}
		\Pm&[s(\mathds{1}\otimes \bm{w}_{c,i-1})]	\nonumber\\
			&=		\Pm\Big[
						\frac{1}{3} \!\cdot\! 3
						\big(
							s(\mathds{1} \otimes \bm{w}_{c,i-1})
							\!-\!
							s(\mathds{1}\otimes \bar{w}_{c,i-1})
						\big)
						\nn\\
						&\quad
						+\!
						\frac{1}{3} \!\cdot\! 3
						\big(
							s(\mathds{1}\otimes \bar{w}_{c,i-1})
							\!-\!
							s(\mathds{1}\otimes w^o)
						\big)
						\!+\!
						\frac{1}{3} \!\cdot\! 3 \!\cdot\!
						s(\mathds{1}\otimes w^o)
					\Big]
					\nonumber\\
			&\overset{(a)}{\preceq}
					\frac{1}{3} \cdot \Pm\big[
						 3
						\big(
							s(\mathds{1} \otimes \bm{w}_{c,i-1})
							-
							s(\mathds{1}\otimes \bar{w}_{c,i-1})
						\big)
					\big]
					\nn\\
					&\quad
					+
					\frac{1}{3}\cdot \Pm
					\big[
						 3
						\big(
							s(\mathds{1}\otimes \bar{w}_{c,i-1})
							-
							s(\mathds{1}\otimes w^o)
						\big)
					\big]
					\nonumber\\
					&\quad
					+
					\frac{1}{3}\cdot \Pm
					\big[
						3 \cdot
						s(\mathds{1}\otimes w^o)
					\big]
					\nonumber\\
			&\overset{(b)}{=}
					3^3 \cdot \Pm\big[
						s(\mathds{1} \otimes \bm{w}_{c,i-1})
						\!-\!
						s(\mathds{1}\otimes \bar{w}_{c,i-1})
					\big]
					\nn\\
					&\quad
					+
					3^3 \cdot \Pm
					\big[
						s(\mathds{1}\otimes \bar{w}_{c,i-1})
						\!-\!
						s(\mathds{1}\otimes w^o)
					\big]
					\nn\\
					&\quad
					+
					3^3 \cdot \Pm
					\big[
						s(\mathds{1}\otimes w^o)
					\big]
					\nonumber\\
			&\overset{(c)}{\preceq}
					27\lambda_U^4 \cdot
					\Pm\big[
						\mathds{1} \otimes (\bm{w}_{c,i-1} - \bar{w}_{c,i-1})
					\big]
					\nn\\
					&\quad
					+\!
					27\lambda_U^4 \!\cdot\!
					\Pm[\mathds{1}\!\otimes\!(\bar{w}_{c,i-1} \!-\! w^o)]
					\!+\!
					27 \cdot \Pm[s(\mathds{1}\otimes w^o)]
					\nonumber\\
			&\overset{(d)}{=}
					27\lambda_U^4 \cdot
					\|\check{\bm{w}}_{c,i-1}\|^4
					\cdot \mathds{1}
					+
					27\lambda_U^4 \cdot
					\|\bar{w}_{c,i-1} - w^o\|^4
					\cdot
					\mathds{1}
					\nn\\
					&\quad
					+
					27 \cdot \Pm[s(\mathds{1}\otimes w^o)]
					\nonumber\\
			&=
					27\lambda_U^4 \cdot
					\|\check{\bm{w}}_{c,i-1}\|^4
					\cdot \mathds{1}
					+
					27\lambda_U^4 \cdot
					\|\tilde{w}_{c,i-1}\|^4
					\cdot
					\mathds{1}
					\nn\\
					&\quad
					+
					27 \cdot \Pm[s(\mathds{1}\otimes w^o)]
					\nonumber\\
			&\overset{(e)}{\preceq}
					27\lambda_U^4 \cdot
					\|\check{\bm{w}}_{c,i-1}\|^4
					\!\cdot\! \mathds{1}
					\!+\!
					27\lambda_U^4 \|\tilde{w}_{c,0}\|^4
					\!\cdot\!
					\mathds{1}
					\nn\\
					&\quad
					+\!
					27 \cdot
					\Pm[s(\mathds{1}\otimes w^o)]
					\nn
	\end{align}
where step (a) uses the convexity property \eqref{Equ:VarMomentPropt:Convexity},
step (b) uses the scaling property in Lemma \ref{Lemma:VarMomentPropt},
step (c) uses the variance relation \eqref{Equ:VarMomentPropt:Update},
step (d) uses the definition of the operator $\Pm[\cdot]$,
and step (e) uses the bound $\|\tilde{w}_{c,i}\|^2 \le \gamma_c^{2i} \cdot \|\tilde{w}_{c,0}\|^2$ from \eqref{P1-Equ:Thm:ConvergenceRefRec:NonAsympBound} of Part I\cite{chen2013learningPart1}
and the fact that $\gamma_c<1$.

Finally, we establish the bound on $\Pm[\bm{v}_i]$ in
\eqref{Equ:Lemma:BoundsPerturbation4th:P_v}. Introduce the $MN \times 1$
vector $\bm{x}$ according to \eqref{Equ:Appendix:phi_A1_w}--\eqref{Equ:Appendix:wi_wci_we}:
	\begin{align}
		\bm{x}	\defeq	\mathds{1} \otimes \bm{w}_{c,i-1}
						+
						\mc{A}_1^T \mc{U}_L \bm{w}_{e,i-1}
				\equiv	\mA_1^T \bm{w}_{i-1}
				=		\bm{\phi}_{i-1}
						\nn
	\end{align}
We partition $\bm{x}$ into block form as
$\bm{x} = \col\{\bm{x}_1,\ldots,\bm{x}_N\}$, where each $\bm{x}_k$ is
$M\times 1$. Then, by the definition of $\bm{v}_i$ from 
\eqref{Equ:Def:v_i}, we have
	\begin{align}
		\Expt &\left\{ \Pm[\bm{v}_i]	\big| \mc{F}_{i-1} \right\}
								\nn\\
						&=		\Expt \left\{ 
									\Pm\big[\hat{\bm{s}}_i(\bm{x}) - s(\bm{x})\big]
									\big|
									\mc{F}_{i-1}
								\right\}
								\nonumber\\
						&=		\col\big\{
									\Expt\big[\|
										\hat{\bm{s}}_{1,i}(\bm{x}_1)
										-
										s_1(\bm{x}_1)
									\|^4 | \mc{F}_{i-1} 
									\big],
									\ldots, 
									\nn\\
									&\qquad\quad
									\Expt\big[ \|
										\hat{\bm{s}}_{N,i}(\bm{x}_N)
										-
										s_N(\bm{x}_N)
									\|^4 | \mc{F}_{i-1}
									\big]
								\big\}
								\nonumber\\
						&\overset{(a)}{\preceq}
								\begin{bmatrix}
									\alpha_4 
									\cdot 
									\|\bm{x}_1\|^4 
									+ \sigma_{v4}^4
									\\
									\vdots\\
									\alpha_4
									\cdot 
									\|\bm{x}_N\|^4 
									+ \sigma_{v4}^4
								\end{bmatrix}
								\nn\\
						&=		\alpha_4 \cdot \Pm[\bm{x}] 
								+ \sigma_{v4}^4 \cdot \mathds{1}
		\label{Equ:Appendix:P4_v_bound1}
	\end{align}
where step (a) uses \eqref{Equ:Assumption:GradientNoise4thOrderMoment}. Now we bound $\Expt \Pm[\bm{x}]$ to complete the proof:
	\begin{align}
		\Pm[\bm{x}]			
						&=		\Pm\left[
									\mathds{1} \otimes \bm{w}_{c,i-1}
									+
									\mc{A}_1^T \mc{U}_L \bm{w}_{e,i-1}
								\right]
								\nonumber\\
						&=		\Pm\Big[
									\frac{1}{3} \!\cdot\! 3
									\big(
											\mathds{1} \!\otimes\! \bm{w}_{c,i-1}
											\!+\!
											\mc{A}_1^T \mc{U}_L \bm{w}_{e,i-1}
											\!-\!
											\mathds{1} \!\otimes\! \bar{w}_{c,i-1}
									\big)
									\nn\\
									&\quad
									+\!
									\frac{1}{3} \!\cdot\! 3
									\big(									
										\mathds{1} \!\otimes\! \bar{w}_{c,i-1}
										\!-\! 
										\mathds{1} \!\otimes\! w^o
									\big)
									\!+\!									
									\frac{1}{3} \!\cdot\! 3 \!\cdot\!
									\mathds{1} \!\otimes\! w^o
								\Big]
								\nonumber\\
						&\overset{(a)}{\preceq}
								27 \cdot \Pm\big[
										\mathds{1} \otimes \bm{w}_{c,i-1}
										\!+\!
										\mc{A}_1^T \mc{U}_L \bm{w}_{e,i-1}
									\!-\!
									\mathds{1} \!\otimes\! \bar{w}_{c,i-1}
								\big]
								\nn\\
								&\quad
									\!+\!
								27 \cdot \Pm \big[
									\mathds{1} \otimes \bar{w}_{c,i-1}
									\!-\! 
									\mathds{1}\otimes w^o
								\big]
								\nn\\
								&\quad
									+ 
								27 \cdot \Pm \big[
									\mathds{1}\otimes w^o
								\big]
								\nonumber\\
						&\overset{(b)}{=}
								27 \cdot 
								\Pm\big[
										\mathds{1} \otimes 
										\check{\bm{w}}_{c,i-1}
										+
										\mc{A}_1^T \mc{U}_L \bm{w}_{e,i-1}
								\big]
								\nn\\
								&\quad
									+
								27 \cdot
								\Pm\big[
									\mathds{1} \otimes \tilde{w}_{c,i-1}
								\big]
									+ 
								27 \cdot \Pm\big[
									\mathds{1}\otimes w^o
								\big]
								\nonumber\\
						&\overset{(c)}{\preceq}
								27 \!\cdot\!
								\left(
								8 \!\cdot\!
								\Pm\big[
										\mathds{1} \!\otimes \!
										\check{\bm{w}}_{c,i-1}
								\big]
								\!+\!
								8 \!\cdot\!
								\Pm\big[
										\mc{A}_1^T \mc{U}_L \bm{w}_{e,i-1}
								\big]
								\right)
								\nn\\
								&\quad
									+\!
								27 \!\cdot\!
								\Pm\big[
									\mathds{1} \!\otimes\! \tilde{w}_{c,i-1}
								\big]
									\!+ \!
								27 \!\cdot\! \Pm\big[
									\mathds{1} \!\otimes\! w^o
								\big]
								\nonumber\\
						&\overset{(d)}{=}
								216 \!\cdot\!
								\|
										\check{\bm{w}}_{c,i-1}
								\|^4
								\!\cdot\!
								\mathds{1}
								\!+\!
								216 \!\cdot\!
								\Pm\big[
										\mc{A}_1^T \mc{U}_L \bm{w}_{e,i-1}
								\big]
								\nn\\
								&\quad
									+
								27 \cdot
								\|
									\tilde{w}_{c,i-1}
								\|^4
								\cdot
								\mathds{1}
									+ 
								27 \cdot 
								\Pm\big[
									\mathds{1}\otimes w^o
								\big]
								\nonumber\\
						&\overset{(e)}{\preceq}
								216 \!\cdot\!
								\|
										\check{\bm{w}}_{c,i-1}
								\|^4
								\!\!\cdot\!
								\mathds{1}
								\nn\\
								&\quad
								+\!
								216 \!\cdot\!
								\big\|
									\bP\big[\mc{A}_1^T \mc{U}_L\big]
								\big\|_{\infty}^4
								\!\!\cdot\!
								\one\one^T
								\!\cdot\!
								\Pm[
										\bm{w}_{e,i-1}
								]
								\nn\\
								&\quad
								+\!
								27 \!\cdot\!
								\|
									\tilde{w}_{c,i-1}
								\|^4
								\!\cdot\!
								\mathds{1}
									\!+ \!
								27 \!\cdot\!
								\Pm\big[
									\mathds{1} \!\otimes\! w^o
								\big]
								\nonumber\\
						&\overset{(f)}{\preceq}
								216 \!\cdot\!
								\|
										\check{\bm{w}}_{c,i-1}
								\|^4
								\!\cdot\!
								\mathds{1}
								\nn\\
								&\quad
								+\!
								216 \!\cdot\!
								\big\|
									\bP\big[\mc{A}_1^T \mc{U}_L\big]
								\big\|_{\infty}^4
								\!\cdot\!
								\one\one^T
								\!\cdot\!
								\Pm[
										\bm{w}_{e,i-1}
								]
								\nn\\
								&\quad
									+\!
								27 \cdot
								\|
									\tilde{w}_{c,0}
								\|^4
								\!\cdot\!
								\mathds{1}
									\!+ \!
								27 \!\cdot\!
								\Pm\big[
									\mathds{1} \!\otimes\! w^o
								\big]
								\nn\\
						&=
								216 \!\cdot\!
								\|
										\check{\bm{w}}_{c,i-1}
								\|^4
								\!\cdot\!
								\mathds{1}
								\nn\\
								&\quad
								+\!
								216 \!\cdot\!
								\big\|
									\bP\big[\mc{A}_1^T \mc{U}_L\big]
								\big\|_{\infty}^4
								\!\cdot\!
								\one\one^T
								\!\cdot\!
								\Pm[
										\bm{w}_{e,i-1}
								]
								\nn\\
								&\quad
								+\!
								27 \cdot
								(
									\|
										\tilde{w}_{c,0}
									\|^4
										\!+ \!
									\|w^o\|^4
								)
								\cdot \one
		\label{Equ:Appendix:P4_v_bound2}
	\end{align}
where step (a) uses the convexity property \eqref{Equ:VarMomentPropt:Convexity}
and the scaling property in Lemma \ref{Lemma:VarMomentPropt},
step (b) uses the variance relation \eqref{Equ:VarMomentPropt:Update},
step (c) uses the convexity property \eqref{Equ:VarMomentPropt:Convexity},
step (d) uses the definition of the operator $\Pm[\cdot]$,
step (e) uses the variance relation \eqref{Equ:VarMomentPropt:Linear}, and
step (f) uses the bound $\|\tilde{w}_{c,i}\|^2 \le \gamma_c^{2i} \cdot \|\tilde{w}_{c,0}\|^2$ from \eqref{P1-Equ:Thm:ConvergenceRefRec:NonAsympBound} of Part I\cite{chen2013learningPart1} and $\gamma_c<1$. Substituting \eqref{Equ:Appendix:P4_v_bound2} into \eqref{Equ:Appendix:P4_v_bound1},
we obtain \eqref{Equ:Lemma:BoundsPerturbation4th:P_v}.

\section{Proof of Lemma \ref{Lemma:UsefulBounds}}
\label{Appendix:Proof_UsefulBound}

First, we prove \eqref{Equ:Lemma:UsefulBounds:Tc_z_4th}. It holds that
	\begin{align}
		\big\|&
			T_c(\w_{c,i-1}) - T_c(\bar{w}_{c,i-1})
			-
			\mu_{\max}
			\cdot
			(p^T \otimes I_M)
			\z_{i-1}
		\big\|^4
								\nn\\
				&=				\Pm\big[
									T_c(\w_{c,i-1}) - T_c(\bar{w}_{c,i-1})
									-
									\mu_{\max}
									\cdot
									(p^T \otimes I_M)
									\z_{i-1}
								\big]
								\nn\\
				&=				\Pm\Big[
									\gamma_c \cdot
									\frac{1}{\gamma_c}
									\big(
										T_c(\w_{c,i-1}) - T_c(\bar{w}_{c,i-1})
									\big)											
									\nn\\
									&\qquad
									+
									(1-\gamma_c)\cdot
									\frac{-\mu_{\max}}{1-\gamma_c}
									\cdot
									(p^T \otimes I_M)
									\z_{i-1}
								\Big]
								\nn\\
				&\overset{(a)}{\preceq}
								\gamma_c \cdot
								\Pm\big[											
									\frac{1}{\gamma_c}
									\big(
										T_c(\w_{c,i-1}) - T_c(\bar{w}_{c,i-1})
									\big)
								\big]	
								\nn\\
								&\quad										
								+
								(1-\gamma_c)\cdot
								\Pm\big[
									\frac{-\mu_{\max}}{1-\gamma_c}
									\cdot
									(p^T \otimes I_M)
									\z_{i-1}
								\big]
								\nn\\
				&\overset{(b)}{=}
								\gamma_c \cdot
								\frac{1}{\gamma_c^4}
								\cdot
								\Pm\big[		
									\big(
										T_c(\w_{c,i-1}) - T_c(\bar{w}_{c,i-1})
									\big)
								\big]	
								\nn\\
								&\quad										
								+
								(1-\gamma_c)\cdot
								\frac{\mu_{\max}^4}{(1-\gamma_c)^4}
								\cdot
								\Pm\big[
									(p^T \otimes I_M)
									\z_{i-1}
								\big]
								\nn\\
				&\overset{(c)}{\preceq}
								\gamma_c
								\cdot
								\Pm[ \w_{c,i-1} - \bar{w}_{c,i-1} ]
								\nn\\
								&\quad
								+
								\frac{\mu_{\max}^4}{(1-\gamma_c)^3}
								\cdot
								\Pm\left[
									(p^T \otimes I_M)
									\z_{i-1}
								\right]
								\nn\\
				&\overset{(d)}{=}
								\gamma_c
								\cdot
								\Pm[ \tilde{\w}_{c,i-1} ]
								\nn\\
								&\quad
								+
								\frac{\mu_{\max}^4}{(1-\gamma_c)^3}
								\cdot
								\left\|
									\sum_{k=1}^N p_{k} \z_{k,i-1}
								\right\|^4
								\nn\\
				&=
								\gamma_c
								\cdot
								\Pm[ \tilde{\w}_{c,i-1} ]
								\nn\\
								&\quad
								+
								\frac{\mu_{\max}^4}{(1-\gamma_c)^3}
								\cdot
								\left(
									\sum_{l=1}^N p_l
								\right)^4
								\cdot
								\left\|
									\sum_{k=1}^N \frac{p_{k}}{\sum_{l=1}^N p_l} \z_{k,i-1}
								\right\|^4
								\nn\\
				&\overset{(e)}{\preceq}
								\gamma_c
								\cdot
								\Pm[ \tilde{\w}_{c,i-1} ]
								\nn\\
								&\quad
								+
								\frac{\mu_{\max}^4}{(1-\gamma_c)^3}
								\cdot
								\left(
									\sum_{l=1}^N p_l
								\right)^4
								\cdot
								\sum_{k=1}^N 
								\frac{p_{k}}{\sum_{l=1}^N p_l} 
								\|
									\z_{k,i-1}
								\|^4
								\nn\\
				&=
								\gamma_c
								\!\cdot\!
								\Pm[ \tilde{\w}_{c,i-1} ]
								\!+\!
								\frac{\mu_{\max}^4}{(1-\gamma_c)^3}
								\!\cdot\!
								\| p \|_1^3
								\!\cdot\!
								\sum_{k=1}^N p_{k} 
								\|
									\z_{k,i-1}
								\|^4
								\nn\\
				&\overset{(f)}{=}
								\gamma_c \!\cdot\!
								\Pm[ \tilde{\w}_{c,i-1} ]
								\!+\!
								\frac{\mu_{\max} \|p\|_1^3}
								{(\lambda_L \!\! - \!\! \frac{1}{2}\mu_{\max}\|p\|_1^2 \lambda_U^2)^3}
								\!\cdot\!
								p^T \!\cdot \!\Pm[ \z_{i-1} ]
								\nn\\
				&\overset{(g)}{\preceq}
								\gamma_c \cdot
								\Pm[ \tilde{\w}_{c,i-1} ]
								\nn\\
								&\quad
								+
								\frac{\mu_{\max} \|p\|_1^3}
								{(\lambda_L - \frac{1}{2}\mu_{\max}\|p\|_1^2 \lambda_U^2)^3}
								\cdot
								p^T \cdot \lambda_U^4
									\cdot
									\left\|
										\bP[\mc{A}_1^T \mc{U}_L]
									\right\|_{\infty}^4
								\nn\\
								&\qquad
									\cdot
									\mathds{1}\mathds{1}^T
									\Pm[\bm{w}_{e,i-1}]
								\nn\\
				&=				\gamma_c \cdot
								\| \tilde{\w}_{c,i-1} \|^4
								\nn\\
								&\quad
								+\!
								\frac{
									\mu_{\max} 
									\lambda_U^4 
									\!\cdot\!
									\|p\|_1^4
									\!\cdot\!
									\left\|
										\bP[\mc{A}_1^T \mc{U}_L]
									\right\|_{\infty}^4
								}
								{(\lambda_L \!\!-\!\! \frac{1}{2}\mu_{\max}\|p\|_1^2 \lambda_U^2)^3}
								\!\cdot\!
									\mathds{1}^T
									\Pm[\bm{w}_{e,i-1}]
								\nn
	\end{align}
where step (a) uses property \eqref{Equ:VarMomentPropt:Convexity}, step (b) uses the scaling property in Lemma \ref{Lemma:VarMomentPropt}, step (c) uses property \eqref{Equ:VarMomentPropt:Centralized}, step (d) introduces $\z_{k,i-1}$ as the $k$th $M\times 1$ sub-vector of $\z_{1-1} = \col\{ \z_{1,i-1}, \ldots, \z_{N,i-1}\}$, step (e) applies Jensen's inequality to the convex function $\|\cdot \|^4$, step (f) uses the definition of the operator $\Pm[\cdot]$, and step (g) uses bound \eqref{Equ:Lemma:BoundsPerturbation4th:P_z}.

Second, we prove \eqref{Equ:Lemma:UsefulBounds:vi_4th}. Let $\bv_{k,i}$ denote the $k$th $M \times 1$ sub-vector of $\bv_i = \col\{ \bv_{1,i}, \ldots,  \bv_{N,i} \}$. Then,
	\begin{align}
		\Expt \big[&
			\left\|
				(p^T \otimes I_M) \bv_i
			\right\|^4
			\big|
			\mc{F}_{i-1}
		\big]					\nn\\
				&=	 			
							\Expt \Big[
									\Big\|
										\sum_{k=1}^N 
										p_k \bv_{k, i}
									\Big\|^4
									\Big|
									\mc{F}_{i-1}
								\Big]			
								\nn\\
				&=
							\left(
								\sum_{l=1}^N p_l
							\right)^4
							\cdot
							\Expt \Big[
									\left\|
										\sum_{k=1}^N 
										\frac{p_k}{\sum_{l=1}^N p_l} 
										\bv_{k, i}
									\right\|^4
									\Big|
									\mc{F}_{i-1}
								\Big]			
								\nn\\
				&\overset{(a)}{\le}
							\left(
								\sum_{l=1}^N p_l
							\right)^4
							\cdot
							\sum_{k=1}^N 
							\frac{p_k}{\sum_{l=1}^N p_l} 
							\Expt \left[
									\left\|										
										\bv_{k, i}
									\right\|^4
									|
									\mc{F}_{i-1}
								\right]	
							\nn\\
				&\overset{(b)}{=}
								\| p \|_1^3
								\cdot
								p^T
								\cdot
								\Expt \big\{
									\Pm[ \bv_i ]
									\big|
									\mc{F}_{i-1}
								\big\}	
								\nn\\
				&\overset{(c)}{\preceq}
							\| p \|_1^3
								\cdot
								p^T
								\cdot
								\Big\{
									216 \alpha_4 \cdot \mathds{1}
									\cdot
									\Pm[\check{\bm{w}}_{c,i-1}]
									\nn\\
									&\qquad
									+
									216
									\alpha_4
									\cdot
									\left\|
										\bP[\mc{A}_1^T\mc{U}_L]
									\right\|_{\infty}^4
									\!\cdot\!
									\mathds{1}\mathds{1}^T
									\!\cdot\!
									\Pm[\bm{w}_{e,i\!-\!1}]
									\nn\\
									&\qquad
									+
									27
									\alpha_4
									\cdot
									(
										\|\tilde{w}_{c,0}\|^4
										+
										\|w^o\|^4
									)
									\cdot
									\one
									+
									\sigma_{v4}^4
									\cdot
									\mathds{1}	
								\Big\}	
								\nn\\
					&=			216\alpha_4 \|p\|_1^4
								\cdot
								\| \check{\w}_{c,i-1} \|^4
								\nn\\
								&\quad
								+
								216\alpha_4 \|p\|_1^4 
								\cdot
								\big\|
									\bP[ \mA_1^T \mU_L ]
								\big\|_{\infty}^4
								\cdot
								\one^T
								\cdot
								\Pm[ \w_{e,i-1} ]
								\nn\\
								&\quad+
								27\alpha_4\|p\|_1^4
								\cdot
								\|\tilde{w}_{c,0}\|^4
								+
								27\alpha_4 \cdot
								\|p\|_1^4
								\cdot
								\|w^o\|^4
								\nn\\
								&\quad
								+
								\sigma_{v4}^4
								\cdot
								\|p\|_1^4
								\nn
	\end{align}
where step (a) applies Jensen's inequality to the convex function $\|\cdot \|^4$, step (b) uses the definition of the operator $\Pm[\cdot]$, and step (c) substitutes \eqref{Equ:Lemma:BoundsPerturbation4th:P_v}.

Third, we prove \eqref{Equ:Lemma:UsefulBounds:Tc_z_2nd}:
	\begin{align}
		\big\|&
			T_c(\w_{c,i-1}) - T_c(\bar{w}_{c,i-1})
			-
			\mu_{\max}
			\cdot
			(p^T \otimes I_M)
			\z_{i-1}
		\big\|^2
								\nn\\
				&=				\Big\|
									\gamma_c
									\cdot
									\frac{1}{\gamma_c}
									\big(
										T_c(\w_{c,i-1}) - T_c(\bar{w}_{c,i-1})
									\big)
									\nn\\
									&\qquad
									+
									(1-\gamma_c)
									\cdot
									\frac{-\mu_{\max}}{1-\gamma_c}
									\cdot
									(p^T \otimes I_M)
									\z_{i-1}
								\Big\|^2
								\nn\\
				&\overset{(a)}{\le}
								\gamma_c
								\cdot
								\left\|									
									\frac{1}{\gamma_c}
									\big(
										T_c(\w_{c,i-1}) - T_c(\bar{w}_{c,i-1})
									\big)
								\right\|^2
								\nn\\
								&\quad
								+
								(1-\gamma_c)
								\cdot
								\left\|
									\frac{-\mu_{\max}}{1-\gamma_c}
									\cdot
									(p^T \otimes I_M)
									\z_{i-1}
								\right\|^2
								\nn\\
				&=
								\gamma_c
								\cdot
								\frac{1}{\gamma_c^2}
								\cdot
								\left\|		
										T_c(\w_{c,i-1}) - T_c(\bar{w}_{c,i-1})
								\right\|^2
								\nn\\
								&\quad
								+
								(1-\gamma_c)
								\cdot
								\frac{\mu_{\max}^2}{(1-\gamma_c)^2}
								\cdot
								\left\|										
									(p^T \otimes I_M)
									\z_{i-1}
								\right\|^2
								\nn\\
				&=
								\gamma_c
								\cdot
								\frac{1}{\gamma_c^2}
								\cdot
								P\left[	
										T_c(\w_{c,i-1}) - T_c(\bar{w}_{c,i-1})
								\right]
								\nn\\
								&\quad
								+
								(1-\gamma_c)
								\cdot
								\frac{\mu_{\max}^2}{(1-\gamma_c)^2}
								\cdot
								\left\|										
									(p^T \otimes I_M)
									\z_{i-1}
								\right\|^2
								\nn\\
				&\overset{(b)}{\preceq}
								\gamma_c
								\cdot
								P\left[	
										\w_{c,i-1} - \bar{w}_{c,i-1}
								\right]
								+
								\frac{\mu_{\max}^2}{1-\gamma_c}
								\cdot
								\left\|											
									(p^T \otimes I_M)
									\z_{i-1}
								\right\|^2
								\nn\\
				&=				\gamma_c
								\cdot
								\left\|		
										\check{\w}_{c,i-1}
								\right\|^2
								+
								\frac{\mu_{\max}}
								{\lambda_L - \frac{1}{2}\|p\|_1^2 \lambda_U^2}
								\cdot
								\big\|											
									(p^T \otimes I_M)
									\z_{i-1}
								\big\|^2
								\nn\\
				&=
								\gamma_c
								\cdot
								\left\|		
										\check{\w}_{c,i-1}
								\right\|^2
								+
								\frac{\mu_{\max}}
								{\lambda_L - \frac{1}{2}\|p\|_1^2 \lambda_U^2}
								\cdot
								\left\|											
									\sum_{k=1}^N
									p_k
									\z_{k,i-1}
								\right\|^2
								\nn\\
				&=
								\gamma_c
								\cdot
								\left\|		
										\check{\w}_{c,i-1}
								\right\|^2
								\nn\\
								&\quad
								+
								\frac{\mu_{\max}}
								{\lambda_L - \frac{1}{2}\|p\|_1^2 \lambda_U^2}
								\cdot
								\Big(
									\sum_{l=1}^N p_l
								\Big)^2
								\cdot
								\left\|											
									\sum_{k=1}^N
									\frac{p_k}{\sum_{l=1}^N p_l}
									\z_{k,i-1}
								\right\|^2
								\nn\\
				&\overset{(c)}{\le}
								\gamma_c
								\cdot
								\left\|		
										\check{\w}_{c,i-1}
								\right\|^2
								\nn\\
								&\quad
								+
								\frac{\mu_{\max}}
								{\lambda_L - \frac{1}{2}\|p\|_1^2 \lambda_U^2}
								\cdot
								\Big(
									\sum_{l=1}^N p_l
								\Big)^2
								\cdot
								\sum_{k=1}^N
								\frac{p_k}{\sum_{l=1}^N p_l}
								\left\|	
									\z_{k,i-1}
								\right\|^2
								\nn\\
				&=				\gamma_c
								\cdot
								\left\|		
										\check{\w}_{c,i-1}
								\right\|^2
								+
								\frac{\mu_{\max}}
								{\lambda_L - \frac{1}{2}\|p\|_1^2 \lambda_U^2}
								\cdot
								\| p \|_1
								\cdot
								p^T
								\cdot
								P[					
									\z_{i-1}
								] 
								\nn\\
				&\overset{(d)}{\le}
								\gamma_c
								\cdot
								\left\|		
										\check{\w}_{c,i-1}
								\right\|^2
								\nn\\
								&\quad
								+
								\frac{\mu_{\max}}
								{\lambda_L - \frac{1}{2}\|p\|_1^2 \lambda_U^2}
								\cdot
								\| p \|_1
								\cdot
								p^T
								\cdot
									\lambda_U^2 \cdot
									\| \bP[ \mA_1^T \mU_L ] \|_{\infty}^2
									\nn\\
									&\qquad
									\cdot
									\one \one^T
									\cdot
									P[ \w_{e,i-1} ]
								\nn\\
				&=				\gamma_c
								\cdot
								\|\check{\w}_{c,i-1}\|^2
								\nn\\
								&\quad
								+\!
								\frac{\mu_{\max}}
								{\lambda_L \! - \! \frac{1}{2}\|p\|_1^2 \lambda_U^2}
								\!\cdot\!
								\| p \|_1^2
								\!\cdot\!
									\lambda_U^2 \!\cdot\!
									\| \bP[ \mA_1^T \mU_L ] \|_{\infty}^2
									\!\cdot\!
									\one^T
									P[ \w_{e,i-1} ]
								\nn
	\end{align}
where steps (a) and (c) apply Jensen's inequality to the convex function $\|\cdot\|^2$, step (b) uses property $P[T_c(x) - T_c(y)] \preceq \gamma_c^2 \cdot P[x-y]$ from \eqref{P1-Equ:VarPropt:Tc} in Part I\cite{chen2013learningPart1}, and step (d) substitutes the bound in \eqref{Equ:Lemma:BoundsPerturbation:P_z}.

Finally, we prove \eqref{Equ:Lemma:UsefulBounds:vi_2nd}. With the block structure $\bv_i = \col\{ \bv_{1,i}, \ldots,  \bv_{N,i} \}$ defined previously, we have
	\begin{align}
		\Expt& \big[
			\left\|
				(p^T \otimes I_M) \bv_i
			\right\|^2
			\big|
			\mc{F}_{i-1}
		\big]					
								\nn\\
				&=			
								\Expt \Big[
									\Big\|
										\sum_{k=1}^N
										p_k
										\bv_{k,i}
									\Big\|^2
									\Big|
									\mc{F}_{i-1}
								\Big]
								\nn\\
				&=
								\Big(
									\sum_{l=1}^N p_l
								\Big)^2
								\cdot
								\Expt \Big[
									\Big\|
										\sum_{k=1}^N
										\frac{p_k}{\sum_{l=1}^N p_l}
										\bv_{k,i}
									\Big\|^2
									\Big|
									\mc{F}_{i-1}
								\Big]
								\nn\\
				&\overset{(a)}{\le}
								\Big(
									\sum_{l=1}^N p_l
								\Big)^2
								\cdot
								\sum_{k=1}^N
								\frac{p_k}{\sum_{l=1}^N p_l}
								\Expt \Big[
									\Big\|										
										\bv_{k,i}
									\Big\|^2
									|
									\mc{F}_{i-1}
								\Big]
								\nn\\
				&=
								\Big(
									\sum_{l=1}^N p_l
								\Big)
								\cdot
								\sum_{k=1}^N
								p_k
								\Expt \Big[
									\Big\|										
										\bv_{k,i}
									\Big\|^2
									|
									\mc{F}_{i-1}
								\Big]
								\nn\\
				&=				
								\|p\|_1 \cdot p^T
								\cdot
								\Expt\big\{
									P[ \bv_i ]
									\big|
									\mF_{i-1}
								\big\}
								\nn\\
				&\overset{(b)}{\le}
								\|p\|_1 \cdot p^T
								\cdot
								\Big\{
									4\alpha \cdot \one
									\cdot
									P[ \check{\w}_{c,i-1} ]
									\nn\\
									&\quad
									+
									4 \alpha
									\cdot
									\| \bP[ \mA_1^T \mU_L ] \|_{\infty}^2
									\cdot
									\one \one^T
									P[ \w_{e,i-1} ]
									\nn\\
									&\quad
									+
									\left[
										4\alpha
										\cdot
										( \|\tilde{w}_{c,0}\|^2 + \|w^o\|^2 )
										+
										\sigma_v^2
									\right]
									\cdot
									\one
								\Big\}
								\nn\\
				&=		
								4 \alpha \|p\|_1^2
								\cdot
								P[ \check{\w}_{c,i-1} ]
								\nn\\
								&\quad
								+
								4\alpha
								\cdot
								\| \bP[ \mA_1^T \mU_L ] \|_{\infty}^2
								\cdot
								\|p\|_1^2
								\cdot
								\one^T
								P[\w_{e,i-1} ]
								\nn\\
								&\quad
								\!+\!
								4 \alpha \|\tilde{w}_{c,0}\|^2
								\!\cdot\!
								\|p\|_1^2
								\!+\!
								4\alpha \|p\|_1^2\! \cdot \!\|w^o\|^2
								\!+\!
								\sigma_v^2 \!\cdot \!\|p\|_1^2
								\nn
	\end{align}
where step (a) applies Jensen's inequality to the convex function $\|\cdot\|^2$, and step (b) substituting \eqref{Equ:Lemma:BoundsPerturbation:P_v_E_Fiminus1}.

%%%%%%%%%%%%%%%%%%%%%%%%%%%%%%%%%%%%%%%%%%%%%%%%%%%%%%%%%

\bibliographystyle{IEEEbib}
\bibliography{DistOpt,IT}

\end{document}